\documentclass[manuscript,screen]{acmart}

\usepackage{tikz}
\usepackage{coeffects} 
\usepackage{stmaryrd}
\usepackage{listings}
\usepackage[capitalize]{cleveref}
\usepackage{proof}

  \newcommand{\EZComm}[1]{} 
  \newcommand{\FDComm}[1]{} 
  \newcommand{\RBComm}[1]{} 
  \newcommand{\PGComm}[1]{} 
  \def\bez{} \def\eez{} 
  \def\bfd{} \def\efd{} 
  \def\brb{} \def\erb{} 
  \def\bpa{} \def\epa{} 
\definecolor{violet}{rgb}{0.62, 0.0, 1.0}
\newif\ifsubmit
\newif\ifrev
\submittrue
\revfalse 
\ifsubmit
  \ifrev 
    \def\bez{\begin{color}{magenta}} 
    \def\eez{\end{color}} 
    \def\bfd{\bez} 
    \def\efd{\end{color}} 
    \def\brb{\bez} 
    \def\erb{\end{color}} 
    \def\bpa{\bez}  
    \def\epa{\end{color}} 
  \fi 
\else 
  \def\bez{\begin{color}{blue}} 
  \def\eez{\end{color}} 
  \def\bfd{\begin{color}{teal}} 
  \def\efd{\end{color}} 
  \def\brb{\begin{color}{red}} 
  \def\erb{\end{color}} 
  \def\bpa{\begin{color}{violet}} 
  \def\epa{\end{color}} 
  \renewcommand{\PGComm}[1]{{\scriptsize \textcolor{violet}{[Paola{:} #1]}}}
  \renewcommand{\RBComm}[1]{{\scriptsize \textcolor{red}{[Riccardo{:} #1]}}}
  \renewcommand{\EZComm}[1]{{\scriptsize \textcolor{blue}{[Elena{:} #1]}}}
  \renewcommand{\FDComm}[1]{{\scriptsize \textcolor{teal}{[Francesco{:} #1]}}}
\fi

\newcommand{\refToExample}[1]{Example~\ref{ex:#1}}
\newcommand{\refToFigure}[1]{Fig.~\ref{fig:#1}}

\newcommand{\refToProposition}[1]{Prop.~\ref{prop:#1}}
\newcommand{\refToDefinition}[1]{Def.~\ref{def:#1}}

\newcommand{\refToTheorem}[1]{Theorem~\ref{thm:#1}}

\newcommand{\refToSection}[1]{Sect.~\ref{sect:#1}}
\newcommand{\refToRule}[1]{\textsc{\small (#1)}}

\newcommand{\meta}[1]{\colorbox{lightgray}{$#1$}}
\newcommand{\Space}{\hskip 0.7em}
\newcommand{\BigSpace}{\hskip 1.5em}

\newcommand{\Tuple}[1]    {({#1})}
\newcommand{\Pair}[2]     {\Tuple{{#1},{#2}}}
\newcommand{\fv}[1]{\aux{fv}(#1)}
\newcommand{\Subst}[3]   {#1[#2/#3]}
\newcommand{\dom}[1]{\aux{dom}(#1)}
\newcommand{\LabelledRule}[3][\vdots]{\scriptstyle{\textsc{(#3)}}\Space\displaystyle\frac{#1}{#2}}
\newcommand{\NamedRule}[4]{\scriptstyle{\textsc{(#1)}}\Space
\displaystyle                  
\frac{#2}{#3}         
\begin{array}{l}
#4     
\end{array}
}
\newcommand{\der}{{\cal D}} 

\newcommand{\rord}{\preceq}
\newcommand{\rjoin}{\curlyvee} 
\newcommand{\rsum}{+}
\newcommand{\rmul}{\times}
\newcommand{\rzero}{0} 
\newcommand{\rone}{1} 
\newcommand{\mord}{\preceq}
\newcommand{\msum}{+}
\newcommand{\mmul}{\times} 
\newcommand{\CCtx}{\aux{CCtx}} 
\newcommand{\cctx}{\gamma}
\newcommand{\EmptyCtx}{\emptyset}
\newcommand{\VarType}[2]{#1:#2}
\newcommand{\VarCoeffect}[2]{#1:#2}
\newcommand{\VarTypeCoeffect}[3]{#1:_{#3}#2}
\newcommand{\coeff}{\metavariable{c}}
\newcommand{\ctxord}{\preceq}
\newcommand{\ctxsum}{+}
\newcommand{\ctxmul}{\times}

\newenvironment{grammatica}{$\begin{array}{lcll}}{\end{array}$}
\newcommand{\produzione}[3]{#1&::=&#2&\mbox{#3}}
\newcommand{\seguitoproduzione}[2]{&&#1&\mbox{#2}}
\newcommand{\terminale}[1]{\texttt{#1}}
\newcommand{\aux}[1]{\mathsf{#1}}
\newcommand{\metavariable}[1]{\mathit{#1}}

\newcommand{\AppExp}[2]{#1\, #2}
\newcommand{\LambdaExpWithType}[3]{\lambda #1{:}#2.#3}
\newcommand{\x}{\metavariable{x}}
\newcommand{\y}{\metavariable{y}} 
\newcommand{\z}{\metavariable{z}}

\newcommand{\num}{\metavariable{n}}
\newcommand{\te} {\mathit{t}} 
\newcommand{\T}{\metavariable{T}}
\newcommand{\INT}{\aux{int}}
\newcommand{\funType}[2]{#1 \rightarrow #2}
\newcommand{\funTypeCoeffect}[3]{#1\xrightarrow{#2}#3}
\newcommand{\ev}{\rightarrow}
\newcommand{\evstar}{\ev^\star}
\newcommand{\reduce}[2]{#1\ev#2}
\newcommand{\reducestar}[2]{#1\evstar#2}
\newcommand{\Vars}{\metavariable{V}} 

\renewcommand{\C}{\metavariable{C}}
\newcommand{\D}{\metavariable{D}}
\newcommand{\const}{\metavariable{k}}
\newcommand{\m}{\metavariable{m}}
\newcommand{\f}{\metavariable{f}} 
\newcommand{\e}{\metavariable{e}}
\newcommand{\es}{\metavariable{es}}
\newcommand{\loc} {\x} 
\newcommand{\this}{\terminale{this}}

\newcommand{\ConstrCall}[2]{\terminale{new}\; #1\terminale{(}#2\terminale{)}}
\newcommand{\ConstrCallTuple}[3]{{\terminale{new}\, #1(#2_1,  \ldots, #2_{#3})}}
\newcommand{\FieldAccess}[2]{#1\terminale{.}#2}
\newcommand{\Block}[4]{\{ #1 \, #2 =#3 \terminale{;}\, #4 \}}
\newcommand{\MethCall}[3]{#1{\terminale{.}}#2\terminale{(}#3\terminale{)}}
\newcommand{\MethCallTuple}[4]{#1.#2(#3_1,  \ldots, #3_{#4})}
\newcommand{\Field}[2]{#1\ #2\terminale{;}}
\newcommand{\FieldAssign}[3]{#1\terminale{.}#2 \terminale{=}\, #3}
\newcommand{\Seq}[2]{#1\terminale{;}\, #2}

\newcommand{\val}{\metavariable{v}}
\newcommand{\vs}{\metavariable{vs}}
\newcommand{\Ctx}[1]{\ctx[#1]}
\newcommand{\ctx}{{\cal{E}}}
\newcommand{\emptyctx}{[\ ]}
\newcommand{\fields}[1]{\aux{fields}(#1)}
\newcommand{\mem} {\mu} 
\newcommand{\ExpMem}[2]{#1{\mid}#2}
\newcommand{\UpdateMem}[4]{#1^{#2.#3=#4}}
\newcommand{\mbody}[2]{{\aux{mbody}(#1,#2)}}
\newcommand{\Object}[2]{[#2]^{#1}}

\newcommand{\sharingRelSymbol}[1]{ \bowtie_{#1}} 
\newcommand{\SharingRel}[4]{#1 \sharingRelSymbol{#2}{#3} #4} 
\newcommand{\reachableSymbol}[1]{ \triangleright_{#1}} 
\newcommand{\reachable}[4]{#1 \reachableSymbol{#2}{#3} #4}

\newcommand{\PT}{\metavariable{P}}
\newcommand{\intType}{\terminale{int}}
\newcommand{\IsWFExp}[3]{#1\vdash#2:#3}
\newcommand{\IsWFObject}[3]{#1\vdash#2:#3}
\newcommand{\IsWFInMem}[3]{#1\Vdash#2:#3}
\newcommand{\IsWFConf}[4]{#1\vdash\ExpMem{#2}{#3}:#4}
\newcommand{\IsWFMem}[2]{#1\vdash#2}
\newcommand{\mtype}[2]{{\aux{mtype}(#1,#2)}}
\newcommand{\prom}{\blacktriangleleft}
\newcommand{\Deriv}{{\cal D}}

\newcommand{\link}{\ell} 
\newcommand{\Lnk}{\aux{Lnk}}
\newcommand{\res}{\aux{res}}
\newcommand{\LnkSet}{\textit{L}}
\newcommand{\X}{\metavariable{X}}
\newcommand{\Y}{\metavariable{Y}} 
\newcommand{\Z}{\metavariable{Z}}
\newcommand{\shCCtx}{\CCtx^\LnkSet} 
\newcommand{\closure}[1]{{#1}^\star}
\newcommand{\SharingScalar}{\mathcal{L}}
\newcommand{\MulLnk}{\triangleleft} 
\newcommand{\clo}{\star}
\newcommand{\shord}{\mathrel{\hat{\subseteq}}}
\newcommand{\shsum}{+}
\newcommand{\shmul}{\mathrel{\triangleleft}}
\newcommand{\hatshmul}{\mathrel{\hat{\triangleleft}}}

\newcommand{\coeffectType}[2]{#1^{#2}}
\newcommand{\links}[1]{\aux{links}(#1)}
\newcommand{\getCoeff}[2]{\aux{coeff}(#1,#2)}

\newcommand{\Restr}[2]{#1{\upharpoonright}#2}
\newcommand{\Capsule}[1]{\aux{capsule}(#1)}
\newcommand{\eqfresh}{\equiv^\aux{fr}}

\newcommand{\TypeMod}[2]{#1^{#2}}
\newcommand{\modif}{\textsc{m}}
\newcommand{\mut}{\aux{mut}}
\newcommand{\capsule}{\aux{caps}}
\newcommand{\readonly}{\aux{read}}
\newcommand{\imm}{\terminale{imm}}
\newcommand{\seal}{\sigma}

\newcommand{\ctxlinsum}{\oplus}
\newcommand{\Modif}[2]{#1[#2]}
\newcommand{\ModifComb}[2]{#1[#2]}
\newcommand{\getModif}[2]{\aux{modif}(#1,#2)}
\newcommand{\modord}{\le} 
\newcommand{\Sealed}[1]{#1[\seal]}

\newcommand{\Erase}[1]{\aux{erase}(#1)}
\newcommand{\IsoOrImm}[1]{\aux{IsoOrImm}(#1)}

\newcommand{\Att}{{\tt A}}
\newcommand{\Ctt}{{\tt C}}
\newcommand{\Btt}{{\tt B}}

\newcommand{\xtt}{{\tt x}}
\newcommand{\xttOne}{{\tt x1}}
\newcommand{\ytt}{{\tt y}}
\newcommand{\ftt}{{\tt f}}
\newcommand{\fttOne}{{\tt f1}}
\newcommand{\ztt}{{\tt z}}
\newcommand{\zttOne}{{\tt z1}}
\newcommand{\zttTwo}{{\tt z2}}
\newcommand{\wtt}{{\tt w}}

\newcommand{\amem} {\mu'} 

\newcommand{\Gordon}{Gordon et al.}

\begin{document}

\title{Coeffects for sharing and mutation}         


\author{Riccardo Bianchini}
\orcid{0000-0003-0491-7652
} 
\affiliation{
  \department{DIBRIS}              
  \institution{Universit\`a di Genova}            
  \country{Italy}
}
\email{riccardo.bianchini@edu.unige.it}          

\author{Francesco Dagnino}
\email{francesco.dagnino@dibris.unige.it} 
\orcid{0000-0003-3599-3535} 
\affiliation{
  \department{DIBRIS}              
  \institution{Universit\`a di Genova}            
  \country{Italy}
}

\author{Paola Giannini}
\affiliation{
  \department{DiSSTE}              
  \institution{Universit\`a del Piemonte Orientale}
  \country{Italy}
}
\email{paola.giannini@uniupo.it}          
\orcid{0000-0003-2239-9529}
\author{Elena Zucca}
\email{elena.zucca@unige.it} 
\orcid{0000-0002-6833-6470} 
\affiliation{
  \department{DIBRIS}              
  \institution{Universit\`a di Genova}            
  \country{Italy}
}
\email{elena.zucca@unige.it}          

\author{Marco Servetto}
\affiliation{
  \department{ECS}              
  \institution{Victoria University of Wellington}            
  \country{New Zealand}
}
\email{marco.servetto@vuw.ac.nz}          
\orcid{0000-0003-1458-2868}

\thanks{This work was partially funded by the MUR project ``T-LADIES'' (PRIN 2020TL3X8X)}


\begin{abstract}

In \emph{type-and-coeffect systems}, contexts are enriched by \emph{coeffects} modeling how they are actually used,  typically through annotations on single variables. Coeffects are computed bottom-up, combining, for each term, the coeffects of its subterms, through a fixed set of algebraic operators. We show that this principled approach can be adopted to track \emph{sharing} in the imperative paradigm, that is, links among variables possibly introduced by the execution. This provides a significant example of non-structural coeffects, which cannot be computed by-variable, since the way a given variable is used can affect the coeffects of other variables. 
To illustrate the effectiveness of the approach, we  enhance the type system tracking sharing to model a sophisticated set of features related to uniqueness and immutability. 
Thanks to the coeffect-based approach, we can express such features in a simple way and prove related properties with standard techniques.
\end{abstract}

\ccsdesc[500]{Theory of computation~Program analysis}
\ccsdesc[500]{Theory of computation~Type structures} 

\keywords{Coeffect systems, sharing, Java} 


\maketitle


\section{Introduction}
Recently, \emph{coeffect systems} have received considerable interest as a mechanism to reason about resource usage \cite{PetricekOM14,BrunelGMZ14,Atkey18,GaboardiKOBU16,GhicaS14,OrchardLE19,ChoudhuryEEW21,DalLagoG22}.  They are, in a sense, the dual of effect systems:  given  a generic type judgment $\IsWFExp{\Gamma}{\e}{\T}$, effects can be seen as an enrichment of the type $\T$ (modeling side effects of the execution), whereas coeffects can be seen as an enrichment of the context $\Gamma$ (modeling how execution needs to use external resources).  In the typical case when $\Gamma$ is a map from variables to types, the type judgment takes shape $\IsWFExp{\VarTypeCoeffect{\x_1}{\T_1}{\coeff_1},\ldots,\VarTypeCoeffect{\x_n}{\T_n}{\coeff_n}}{\e}{\T}$, where the \emph{scalar} coeffect\footnote{Also called \emph{grade}, using the terminology \emph{graded type system rather than coeffect system}. }  $\coeff_i$ models how variable $\x_i$ is used in $\e$. Such annotations on variables are an output of the typing process rather than an input: they are computed bottom-up, as a linear combination, for each term, of those of  its  subterms.  

 The fact that a program introduces sharing between two variables, say $\x$ and $\y$, for instance through  a field  assignment $\FieldAssign{\x}{\f}{\y}$ in an object-oriented language,  clearly has  a coeffect nature (indeed, it is a particular way to use the resources $\x$ and $\y$). However,  to the best of  our knowledge no attempt has been made to use coeffects to track this information statically. 

 A likely reason  is that  this  kind of coeffect does not fit in the framework of \emph{structural} coeffect systems,  which predominate  in  the  literature, and where the coeffect of each single variable  is   computed independently. Clearly, this is not the case for sharing,  since a program introducing sharing between $\x$ and $\y$, and between $\y$ and $\z$, introduces sharing between $\x$ and $\z$ as well.  

In this paper, we show that sharing can be naturally modeled by coeffects following a more general schema distilling their fundamental ingredients; namely, a \emph{semiring} for scalar coeffects, and a \emph{module} structure for coeffect contexts, providing \emph{sum} and \emph{scalar multiplication}.
  Whereas  scalars are regularly assumed in  the  literature to form (some variant of) a semiring, the fact that coeffect contexts form a module has only  been noted,   to the best of  our knowledge, by \citet{McBride16} \bfd and subsequently by \citet{WoodA22}. \efd   In these papers,   however, only structural coeffects are considered, that is, modules where sum and scalar multiplication are defined pointwise.  Sharing coeffects provide a significant non-structural instance of the framework, motivating its generality. 
  

 To define the sharing coeffect system, we take as a reference language an imperative variant of Featherweight Java \cite{IgarashiPW99}, and we extend the standard type system of the language by adding coeffects which track sharing introduced by the execution of a program.  Following the guiding principle of coeffect systems, they are computed bottom up, starting from the rule for variables and constructing more complex coeffects by sums and scalar multiplications, where coefficients are determined by each language construct. Hence,  the typing rules  can be easily turned into an algorithm.  In the resulting type-and-coeffect system, we are able to detect, in a simple and static way, some relevant notions in  the  literature, notably that 
 an expression is a \emph{capsule}\footnote{ We adopt the terminology of \citet{GianniniSZC19,GianniniRSZ19}; in  the  literature there  are  many variants of this notion with different names \cite{ClarkeWrigstad03,Almeida97,ServettoEtAl13a,Hogg91,DietlEtAl07,GordonEtAl12}.}, that is, evaluates to the \emph{unique entry point} for a portion of memory.

To illustrate the effectiveness of this approach, we  enhance the type system tracking sharing to model a sophisticated set   of features related to uniqueness and immutability.  Notably, we integrate and formalize the language designs proposed in \cite{GianniniSZC19,GianniniRSZ19}, which have two common key ideas. The first is 
to use \emph{modifiers} ($\readonly$, $\capsule$, and $\imm$ for read-only, uniqueness, and immutability, respectively), allowing the programmer to specify the corresponding constraints/properties in variable/parameter declarations and method return types.
The second is that uniqueness and immutability ($\capsule$ and $\imm$ tags) are not \emph{imposed}, but \emph{detected} by the type system, supporting \emph{ownership transfer} rather than \mbox{\emph{ownership invariants} (see \refToSection{related}).}

 Because it is built  on top of the sharing coeffects, the type-and-coeffect system we design to formalize the above features  significantly  improves  the  previous work \cite{GianniniSZC19,GianniniRSZ19}\footnote{A more detailed comparison is provided in \refToSection{conclu}.}.  Notably, features from  both  papers are integrated; inference of $\capsule$ and $\imm$ types is \emph{straightforward} from the coeffects, through a simple \emph{promotion} rule; the design of the type system is guided, and rules can be easily turned into an algorithm. Finally, the coeffect system uniformly computes both sharing introduced by a program \emph{existing in current memory}, allowing us to express and prove relevant properties in a clean way and with standard techniques. 

 In summary,  the contributions of the paper are the following:
\begin{itemize}
\item A schema distilling the ingredients of coeffect systems, mentioned by \citet{McBride16} and  \citet{WoodA22}, but  never  used in its generality, that is,  beyond structural instances.  
\item The first,  to the best of  our knowledge, formalization of sharing by a coeffect system.  We prove subject reduction, stating that not only is type preserved, but also sharing.
\item On top of  such a  coeffect system, an enhanced type system supporting sophisticated features related to uniqueness and immutability. 
We prove subject reduction, stating that type, sharing, and modifiers are preserved, so that we can statically detect uniqueness and immutability. 
\end{itemize}
 We stress that the aim of the paper is \emph{not} to propose a novel design of memory-management features, but to provide,  via  a complex example, a proof-of concept that coeffects can be the basis for modeling such features, which could be fruitfully employed in other cases. In particular, we demonstrate the following:
\begin{itemize}
\item  The paradigm of
coeffects can be applied for useful purposes in an imperative/OO setting, whereas so far in the literature it has only been used in functional calculi and languages.
\item The views of ownership as  a  substructural capability and as  a  graph-theoretic property of heaps can be unified.
\item The expressive power of the complicated and ad-hoc type systems \bez in \cite{GianniniSZC19,GianniniRSZ19} \eez can be achieved in a much more elegant and principled way, using only  simple algebraic operations. 
Moreover,  the coeffect approach allows one to reuse existing general results regarding algorithms/implementations, like, e.g., those used in Granule \cite{OrchardLE19}.
\end{itemize}
In \refToSection{calculus} we present the reference language, and illustrate the properties we want to guarantee.
In \refToSection{coeffects} we illustrate the ingredients of coeffect systems through a classical example, and we define their general algebraic structure.   In  \refToSection{sharing} and \refToSection{extended} we describe the two type systems outlined above with the related results.  In \refToSection{comparison} we discuss the expressive power of our system, compared with closely related proposals.  In \refToSection{related} we outline other related work, and in \refToSection{conclu} we summarize our contribution and discuss future work. Omitted proofs are in the \bez Appendix. \eez


\section{Sharing and mutation in a Java-like calculus}\label{sect:calculus}
We illustrate the properties we want to guarantee with the coeffect system on a simple reference language, an imperative variant of Featherweight Java \cite{IgarashiPW99}.
\subsection{The language}
For  the reader's  convenience, syntax, reduction rules, and the standard type system are reported  in \refToFigure{calculus}. We write $\es$ as  a  metavariable for $\e_1, \ldots, \e_n$, $n\geq 0$, and analogously for other sequences. Expressions of primitive types, unspecified, include  constants  $\const$. We assume a unique set of \emph{variables} $\x,\y,\z,\ldots$ which occur both in source code (method parameters, including the special variable $\this$, and local variables in blocks) and as references in memory. Moreover, we assume sets of \emph{class names} $\C$, \emph{field names} $\f$, and \emph{method names} $\m$. 
In addition to the standard constructs of imperative object-oriented languages (field access, field assignment, and object creation), we have a block expression, consisting of a local  variable declaration, and the body in which this variable can be used. We will sometimes abbreviate $\Block{\T}{\x}{\e}{\e'}$ by $\Seq{\e}{\e'}$ when $\x$ does not occur \mbox{free in $\e$.}
 
To be concise, the class table is abstractly modeled as follows, omitting its (standard) syntax:
\begin{itemize}
\item $\fields{\C}$ gives, for each class $\C$, the sequence $\Field{\T_1}{\f_1}\ldots\Field{\T_n}{\f_n}$ of its fields with their types;
\item $\mbody{\C}{\m}$ gives, for each method $\m$ of class $\C$, its parameters and body
\item $\mtype{\C}{\m}$ gives, for each method $\m$ of class $\C$, its parameter types and return type.
\end{itemize}
For simplicity, we do not consider subtyping (inheritance), which is  an  orthogonal feature. 

Method bodies are expected to be well-typed with respect to method types.
Formally, $\mbody{\C}{\m}$ and $\mtype{\C}{\m}$ are either both defined or both undefined; in the first case $\mbody{\C}{\m}=\Pair{\x_1\dots\x_n}{\e}$, $\mtype{\C}{\m}=\funType{\T_1\ldots\T_n}{\T}$, and  
\begin{quote}
$\IsWFExp{\VarType{\this}{\Ctt},\VarType{\x_1}{\T_1}\ldots,\VarType{\x_n}{\T_n}}{\e}{\T}$
\end{quote}
holds. 
\begin{figure}
\begin{grammatica}
\produzione{\e}{\x\mid\const\mid\FieldAccess{\e}{\f}\mid\FieldAssign{\e}{\f}{\e'}\mid\ConstrCall{\C}{\es}\mid\MethCall{\e}{\m}{\es}\mid\Block{\T}{\x}{\e}{\e'}\mid\ldots}{expression}\\
\produzione{\T}{\C\mid\PT}{type}\\
\\
\produzione{\val}{\loc\mid\const}{value}\\
\produzione{\ctx}{\emptyctx\mid\FieldAccess{\ctx}{\f}\mid\FieldAssign{\ctx}{\f}{\e'}\mid\FieldAssign{\loc}{\f}{\ctx}\mid\ConstrCall{\C}{\vs,\ctx,\es}}{evaluation context}\\
\seguitoproduzione{\mid\MethCall{\ctx}{\m}{\es}\mid\MethCall{\loc}{\m}{\vs,\ctx,\es}\mid\Block{\T}{\x}{\ctx}{\e}\mid\ldots}{}\\[2ex]
\end{grammatica}

\hrule

\begin{math}
\begin{array}{l}
\NamedRule{ctx}{\reduce{\ExpMem{\e}{\mem}}{\ExpMem{\e'}{\mem'}}}{\reduce{\ExpMem{\Ctx{\e}}{\mem}}{\ExpMem{\Ctx{\e'}}{\mem'}}}{}
\BigSpace
\NamedRule{field-access}{}{\reduce{\ExpMem{\FieldAccess{\loc}{\f_i}}{\mem}}{\ExpMem{\val_i}{\mem}}}
{\mem(\loc)=\ConstrCall{\C}{\val_1,\dots,\val_n}\\
\fields{\C}=\Field{\T_1}{\f_1}\dots\Field{\T_n}{\f_n}\\
i\in 1..n}\\[4ex]
\NamedRule{field-assign}{}{\reduce{\ExpMem{\FieldAssign{\loc}{\f_i}{\val}}{\mem}}{\ExpMem{\val}{\UpdateMem{\mem}{\loc}{i}{\val}}}}{
\mem(\loc)=\Object{\C}{\val_1,\ldots,\val_n}\\
\fields{\C}=\Field{\T_1}{\f_1}\dots\Field{\T_n}{\f_n}\\
i\in 1..n}\\[4ex]
\NamedRule{new}{}{\reduce{\ExpMem{\ConstrCall{\C}{\vs}}{\mem}}{\ExpMem{\loc}{\Subst{\mem}{\ConstrCall{\C}{\vs}}{\loc}}}}{\loc\not\in\dom{\mem}}\\[4ex]
\NamedRule{invk}{}{\reduce{\ExpMem{\MethCall{\loc}{\m}{\val_1,\dots,\val_n}}{\mem}}\ExpMem{\Subst{\Subst{\e}{\loc}{\this}}{\val_1}{\x_1} \ldots [\val_n / \x_n]}{\mem}}
{\mem(\loc)=\ConstrCall{\C}{\vs}\\
\mbody{\C}{\m}=\Pair{\x_1\dots\x_n}{\e}}\\[4ex]
\NamedRule{block}{}{\reduce{\ExpMem{\Block{\T}{\loc}{\val}{\e}}{\mem}}{\ExpMem{\Subst{\e}{\val}{\loc}}{\mem}}}{}
\\[3ex]
\end{array}
\end{math}

\hrule

\begin{math}
\begin{array}{l}
\\
\NamedRule{t-var}{}{\IsWFExp{\Gamma}{\x}{\T}}{\Gamma(\x)=\T
}
\BigSpace
\NamedRule{t-const}{}{\IsWFExp{\Gamma}{\const}{\PT_\const}}{}\\[5ex] 
\NamedRule{t-field-access}{\IsWFExp{\Gamma}{\e}{\C}}{\IsWFExp{\Gamma}{\FieldAccess{\e}{\f_i}}{\T_i}}
{
\fields{\C}=\Field{\T_1}{\f_1} \ldots \Field{\T_n}{\f_n} \\
i\in 1..n\\
}
\\[5ex]
\NamedRule{t-field-assign}{\IsWFExp{\Gamma}{\e}{\C} \Space \IsWFExp{\Gamma}{\e'}{\T_i}}{\IsWFExp{\Gamma}{\FieldAssign{\e}{\f_i}{\e'}
}{\T_i}}
{
\fields{\C}=\Field{\T_1}{\f_1} \ldots \Field{\T_n}{\f_n}\\
  i\in 1..n\\ 
}
\\[5ex]
\NamedRule{t-new}{\IsWFExp{\Gamma}{\e_i}{\T_i}\Space \forall i\in 1..n}{\IsWFExp{\Gamma}{\ConstrCallTuple{\C}{\e}{n}}{\C}}
{\fields{\C}=\Field{\T_1}{\f_1} \ldots \Field{\T_n}{\f_n}}
\\[5ex]
\NamedRule{t-invk}{\IsWFExp{\Gamma}{\e_0}{\C} \BigSpace \IsWFExp{\Gamma}{\e_i}{\T_i}\Space \forall i\in 1..n}{\IsWFExp{\Gamma}{\MethCallTuple{\e_0}{\m}{\e}{n}}{\T}}{
\mtype{\C}{\m}=\funType{\T_1\ldots\T_n}{\T}
}
\\[5ex]
\NamedRule{t-block}{
\IsWFExp{\Gamma}{\e}{\T} \BigSpace
\IsWFExp{\Gamma, \VarType{\x}{\T}}{\e'}{\T'}
}
{
\IsWFExp{\Gamma}{\Block{\T}{\x}{\e}{\e'}}{\T'}
}
{}\\[3ex]
\end{array}
\end{math}

\hrule

\begin{math}
\begin{array}{l}
\\
\NamedRule{t-conf}{\IsWFExp{\Gamma}{\e}{\T} \BigSpace \IsWFMem{\Gamma}{\mem}}
{\IsWFConf{\Gamma}{\e}{\mem}{\T}} 
{}
\BigSpace
\NamedRule{t-obj}{\IsWFExp{\Gamma}{\val_i}{\T_i}\Space \forall i\in 1..n}{\IsWFObject{\Gamma}{\Object{\C}{\val_1,\ldots,\val_n}}{\C}}{\fields{\C}=\Field{\T_1}{\f_1} \ldots \Field{\T_n}{\f_n}}

\\[4ex]
\NamedRule{t-mem}{
\IsWFExp{\Gamma}{\mem(\x_i)}{\C_i\Space\forall i\in 1...n}
}
{\IsWFMem{\Gamma}{\mem}} 
{
\Gamma=\VarType{\x_1}{\C_1},\ldots,\VarType{\x_n}{\C_n}\\
\dom{\Gamma}=\dom{\mem}
}
\end{array}
\end{math}
\caption{Syntax, reduction rules, and standard type system of the Java-like calculus}\label{fig:calculus}
\end{figure}
Reduction is defined over \emph{configurations} of shape $\ExpMem{\e}{\mem}$, where a \emph{memory} $\mem$ is a map from references to \emph{objects} of shape $\Object{\C}{\val_1,\ldots,\val_n}$, and we assume free variables in $\e$ to be  in $\dom{\mem}$.   We denote by $\UpdateMem{\mem}{\loc}{i}{\val}$ the memory obtained from $\mem$ by updating the $i$-th field of the object associated to $\loc$ by $\val$, and by $\Subst{\e}{\val}{\loc}$ the usual capture-avoiding substitution.

Reduction and typing rules are straightforward. In rule \refToRule{t-conf}, a configuration is well-typed if the expression is well-typed, and the memory is well-formed in the same context (recall that free variables in the expression are bound in the domain of the memory). In rule \refToRule{t-mem}, a memory is well-formed in a context assigning a type to all and only references in memory, provided that, for each reference, the associated object has the same type. 

\subsection{Sharing and mutation}
In languages with state and mutations, keeping control of sharing is a key issue for correctness. This is exacerbated by concurrency mechanisms, since side-effects in one thread can affect the behaviour of another, hence unpredicted sharing can induce unplanned/unsafe communication.

\emph{Sharing} means that some portion of the memory can be reached through more than one reference, say through $\x$ and $\y$, so that manipulating the memory through $\x$ can affect $\y$ as well. 
\begin{definition}[Sharing in memory]\label{def:sharing-rel}
The \emph{sharing relation} in memory $\mem$, denoted by $\sharingRelSymbol{\mem}$, is the smallest equivalence relation on $\dom{\mem}$ such that:
\begin{quote}
$\SharingRel{\x}{\mem}{\y}$ if $\mem(\x)=\Object{\C}{\val_1,\ldots,\val_n}$ and $\y = \val_i$ for some  $i\in 1..n$
\end{quote}
\end{definition}

 Note that $\y = \val_i$ above means that $\y$ and $\val_i$ \emph{are the same reference}, that is, \bez it \eez corresponds to what is sometimes called pointer equality. 

It is important for a programmer to be able to rely on \emph{capsule} and \emph{immutability} properties. 
Informally, an expression has the capsule property if its result will be the \emph{unique entry point} for a portion of store. For instance, we expect the result of a \lstinline{clone}  method to be a capsule, see  \refToExample{ex2} below. This allows programmers to identify state
that can be safely 
 used by a thread since no other thread can access/modify it. 
A reference has the immutability property if its reachable object graph will be never modified. As a consequence,
an immutable reference can be safely shared 
 by threads. 

The following simple example illustrates the capsule property.

\begin{example}\label{ex:ex1}
Assume the following class table:
\begin{lstlisting}
class B {int f;}
class C {B f1; B f2;}
\end{lstlisting}
and consider the expression $\e={\small \Block{\Btt}{\ztt}{\ConstrCall{\Btt}{2}}{\Seq{\FieldAssign{\xtt}{\fttOne}{\ytt}}{\ConstrCall{\Ctt}{\ztt,\ztt}}}}$. 
This expression has two free variables (in other words, uses two external resources) $\x$ and $\y$. We expect such free variables to be bound to an outer declaration, if the expression occurs as  a  subterm of a program, or to represent references in current memory. This expression is a \emph{capsule}. Indeed, even though it has free variables (uses external resources) \lstinline{x} and \lstinline{y}, such variables will \emph{not} be connected to the final result. We say that they are \emph{lent} in $\e$. In other words, lent references can be manipulated during the evaluation, but cannot be permanently saved. So, we have the guarantee that the result of evaluating $\e$, regardless of the initial memory, will be a reference pointing to a \emph{fresh} portion of memory. For instance, evaluating $\e$ in $\mem=\{\xtt\mapsto\Object{\Ctt}{\xttOne,\xttOne},\xttOne\mapsto\Object{\Btt}{0},\ytt\mapsto\Object{\Btt}{1}\}$, the final result is  a fresh reference  $\wtt$, in the memory $\mem'=\{\xtt\mapsto\Object{\Ctt}{\ytt,\xttOne},\xttOne\mapsto\Object{\Btt}{0},\ytt\mapsto\Object{\Btt}{1},\ztt\mapsto\Object{\Btt}{2},\wtt\mapsto\Object{\Ctt}{\ztt,\ztt}\}$.
\end{example}

Lent and capsule properties are formally defined below.

\begin{definition}[Lent reference]\label{def:lent} 
For $\x\in\fv{\e}$,  $\x$ is lent in $\e$  if, for all $\mu$, $\reducestar{\ExpMem{\e}{\mem}}{\ExpMem{\y}{\mem'}}$ implies $\SharingRel{\x}{\mem'}{\y}$ does not hold.
\end{definition}

An expression $\e$ is a capsule if all its free variables are lent in $\e$.

\begin{definition}[Capsule expression]\label{def:caps}
An expression $\e$ is a \emph{capsule} if, for all $\mem$, $\reducestar{\ExpMem{\e}{\mem}}{\ExpMem{\y}{\mem'}}$ implies that, for all $\x\in\fv{\e}$, $\SharingRel{\x}{\mem'}{\y}$ does not hold.

%
\end{definition}

 The capsule property can be easily detected in simple situations, such as using a primitive deep clone operator, or a closed expression. However, the property  also  holds in many other cases, 
 which are {\em not} easily   detected (statically) since they depend on \emph{the way variables are used}.  To see this, we consider a more involved example,  adapted from \cite{GianniniSZC19}. 
\begin{example}\label{ex:ex2}\
\begin{lstlisting}
class B {int f; B clone() {new B(this.f)}
class A { B f;
  A mix (A a) {this.f=a.f; a} // this, a and result linked
  A clone () {new A(this.f.clone())} // this and result not linked
}
A a1 = new A(new B(0));
A mycaps = {A a2 = new A(new B(1));
  a1.mix(a2).clone() // (1) 
  // a1.mix(a2).clone().mix(a2) // (2)
}
\end{lstlisting}

The result of \lstinline{mix}, as the name suggests,  will be connected to both the receiver and the argument, whereas the result of \lstinline{clone}, as expected for such a method, will be a reference to a \emph{fresh} portion of memory
 which is  not connected to the receiver. 

 Now let us consider  the code after the class definition, where the programmer wants the guarantee that \lstinline{mycaps} will be initialized with a capsule, that is, an expression which evaluates to the entry point of a fresh portion of memory.  
\begin{figure}[ht]
\begin{center}
\includegraphics
[width=.8\textwidth]
{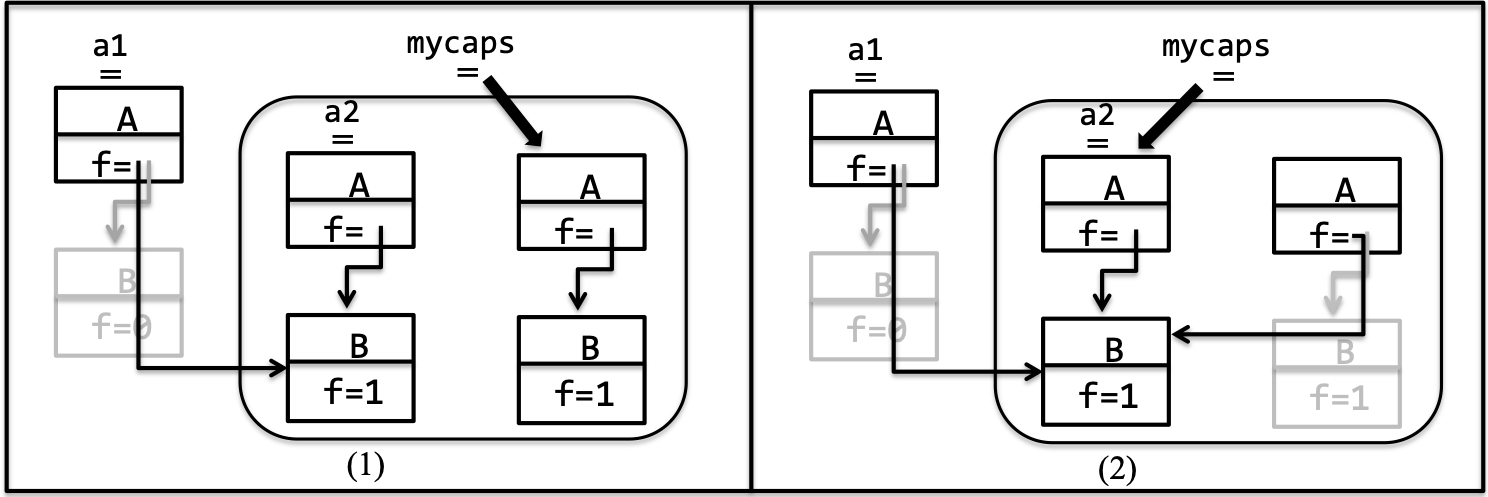}
\end{center}
\caption{Graphical representation of the store for \refToExample{ex2}} \label{fig:ex2}
\end{figure}
 \refToFigure{ex2} shows a graphical representation of the store after the evaluation of such code.
Side \bez (1) \eez shows the resulting store if we evaluate line (1) but not line (2), while side \bez (2) \eez shows the resulting store 
if we evaluate line (2) but not line (1). The 
thick arrow points to the result of the evaluation of the block and \lstinline{a2} is a local variable. 
In side \bez (1) \eez \lstinline{a1} is not in sharing with \lstinline{mycaps}, whereas
in side \bez (2) \eez \lstinline{a1} is in sharing with  \lstinline{a2} which is in sharing with \lstinline{mycaps} and so
\lstinline{a1} is in sharing with \lstinline{mycaps} as well.
 Set
 \begin{itemize}
\item $\e_1$ = \lstinline@{A a2 = new A(new B(1)); a1.mix(a2).clone()}@
\item  $\e_2$ = \lstinline@{A a2 = new A(new B(1)); a1.mix(a2).clone().mix(a2)}@
\end{itemize}
We can see that \lstinline{a1} is lent in $\e_1$, since its evaluation produces the object pointed \bez to \eez by the thick arrow which is not in sharing with \lstinline{a1}, 
whereas \lstinline{a1} is not lent in $\e_2$. Hence, $\e_1$ is a capsule, since its  free variable,  \lstinline{a1}, is not in sharing with the result of its
evaluation, whereas  \lstinline{a2} is not.
\end{example}

We consider now immutability. A reference $\x$ has the immutability property if the portion of memory reachable from $\x$ will never change during execution, as formally stated below.

\begin{definition}\label{def:reach}
The \emph{reachability relation} in memory $\mem$, denoted by $\reachableSymbol{\mem}$, is the reflexive and transitive closure of the relation on $\dom{\mem}$ such that:
\begin{quote}
$\reachable{\x}{\mem}{\y}$ if $\mem(\x)=\Object{\C}{\val_1,\ldots,\val_n}$ and $\y = \val_i$ for some  $i\in 1..n$
\end{quote}
\end{definition}

\begin{definition}[Immutable reference]\label{def:imm} 
For $\x\in\fv{\e}$, $\x$ is \emph{immutable} in $\e$ if $\reducestar{\ExpMem{\e}{\mem}}{\ExpMem{\e'}{\mem'}}$ and 
$\reachable{\x}{\mem}{\y}$ implies $\mem(\y)=\mem'(\y)$.
\end{definition}

A typical way to prevent mutation, as we will show in \refToSection{extended}, is by a type modifier $\readonly$, so that an expression with type tagged in this way   cannot occur  as the left-hand side of a field assignment.  However, to have the guarantee that a certain portion of memory is actually immutable, a type system should be able to detect that it cannot be modified through \emph{any} possibile reference. For instance, consider a variant of \refToExample{ex2} with the same classes \lstinline{A} and \lstinline{B}. 

\begin{example}\label{ex:ex3}\
\begin{lstlisting}
A a1 = new A(new B(0));
read A mycaps = {A a2 = new A(new B(1));
  a1.mix(a2).clone() // (1) 
  // a1.mix(a2).clone().mix(a2) // (2)
}
// mycaps.f.f=3 // (3)
a1.f.f=3 // (4)
\end{lstlisting}
\end{example}
The reference \lstinline{mycaps} is now declared  as  a $\readonly$ type, hence we cannot modify its reachable object graph through \lstinline{mycaps}. For instance, line (3) is ill-typed. 
However, if we replace line (1) with line (2), since in this case \lstinline{mycaps} and \lstinline{a1} share their \lstinline{f} field, the same effect of line (3) can be obtained by line (4). 
This example shows that the immutability property is, roughly, a conjunction of the $\readonly$ restriction and the capsule property.


\section{Coeffect systems}\label{sect:coeffects}
In \refToSection{informal} we illustrate the fundamental ingredients of coeffect systems through a classical example, and in \refToSection{sring-mod} we formally define their general algebraic structure.  
\subsection{An example}\label{sect:informal}

In \refToFigure{coeff-lambda} we  show  the example  which is generally used to illustrate how a coeffect system works\footnote{ More precisely, the structure of coeffects is that of most papers cited in the Introduction, and the calculus a variant/combination of examples in  those  papers.}.  Namely, a simple coeffect system for the  call-by-name  $\lambda$-calculus where we trace when a variable is either not used, or used linearly (that is, exactly once), or used in an unrestricted way, as expressed by assigning to the variable a \emph{scalar coeffect} $\coeff$. 

\begin{figure}[h]
\begin{grammatica}
\produzione{\te}{\num\mid\x\mid\LambdaExpWithType{\x}{\T}{\te}\mid\AppExp{\te_1}{\te_2}}{}\\
\produzione{\coeff}{\rzero\mid\rone\mid\omega}{}\\
\produzione{\T}{\INT\mid\funTypeCoeffect{\T_1}{\coeff}{\T_2}}{}\\
\produzione{\gamma}{\VarCoeffect{\x_1}{\coeff_1}, \ldots, \VarCoeffect{\x_n}{\coeff_n}}{}\\
\produzione{\Gamma,\Delta}{\VarTypeCoeffect{\x_1}{\T_1}{\coeff_1},\ldots,\VarTypeCoeffect{\x_n}{\T_n}{\coeff_n}}{}
\end{grammatica}
$
\LabelledRule[]{\reduce{\AppExp{(\LambdaExpWithType{\x}{\T}{\te})}{\te'}}{\Subst{\te}{\te'}{\x}}}{AppAbs}$\BigSpace
$\LabelledRule[\te_1\ev\te'_1]{\AppExp{\te_1}{\te_2}\ev\AppExp{\te'_1}{\te_2}  }{App}
$
\\[1ex]
\hrule
\begin{math}
\begin{array}{l}
\\
\NamedRule{t-const}{}
{\IsWFExp{\emptyset}{\num}{\INT}}
{}{}
\BigSpace
\NamedRule{t-var}{}
{\IsWFExp{\rzero\mmul\Gamma\msum\VarTypeCoeffect{\x}{\T}{\rone}}{\x}{\T}}
{}{}
\BigSpace
\NamedRule{t-sub}{\IsWFExp{\Gamma}{\te}{\T}}{\IsWFExp{\Gamma'}{\te}{\T}}{\Gamma'\ctxord\Gamma}
\\[3ex]
\NamedRule{t-abs}{\IsWFExp{\Gamma,\VarTypeCoeffect{\x}{\T_1}{\coeff}}{\te}{\T_2}}
{\IsWFExp{\Gamma}{\LambdaExpWithType{\x}{\T_1}{\te}}{\funTypeCoeffect{\T_1}{\coeff}{\T_2}}}
{}{}
\BigSpace
\NamedRule{t-app}{\IsWFExp{\Gamma_1}{\te_1}{\funTypeCoeffect{\T_2}{\coeff}{\T_1}} \qquad \IsWFExp{\Gamma_2}{\te_2}{\T_2}}
{\IsWFExp{\Gamma_1 \msum ((\coeff\meta{\rjoin\rone}) \mmul \Gamma_2)}{\AppExp{\te_1}{\te_2}}{\T_1}}
{}{}
\end{array}
\end{math}
\caption{A simple structural coeffect system}\label{fig:coeff-lambda}
\end{figure}
A \emph{coeffect context}, of shape $\cctx=\VarCoeffect{\x_1}{\coeff_1}, \ldots, \VarCoeffect{\x_n}{\coeff_n}$, where order is immaterial and $\x_i\neq\x_j$ for $i\neq j$, represents a map from variables  to  scalar coeffects where only a finite number of variables have non-zero coeffect. 
A \emph{(type-and-coeffect) context}, of shape $\Gamma=\VarTypeCoeffect{\x_1}{\T_1}{\coeff_1},\ldots,\VarTypeCoeffect{\x_n}{\T_n}{\coeff_n}$, with analogous conventions, represents the pair of the standard type context $\VarType{\x_1}{\T_1}\ldots,\VarType{\x_n}{\T_n}$, and the coeffect context $\VarCoeffect{\x_1}{\coeff_1}, \ldots, \VarCoeffect{\x_n}{\coeff_n}$. 
We write $\dom{\Gamma}$ for $\{\x_1,\ldots,\x_n\}$.

 Scalar coeffects usually form a preordered semiring \cite{BrunelGMZ14,GhicaS14,McBride16,Atkey18,GaboardiKOBU16,AbelB20,OrchardLE19,ChoudhuryEEW21,WoodA22}, that is, they are equipped with a \emph{preorder} $\rord$ (with binary join $\rjoin$), a \emph{sum} $\rsum$, and a \emph{multiplication} $\rmul$, satisfying  some  natural axioms, see \refToDefinition{semiring} in \refToSection{sring-mod}. 
In the example, the (pretty intuitive) definition of such  a  structure is given below.  

\begin{small}
\begin{center}
 $\rzero\rord\omega$, $\rone\rord\omega$ 
\BigSpace
\begin{tabular}{|c|c|c|c|}
\hline
$\rsum$ & $\rzero$ & $\rone$ & $\omega$ \\
\hline
$\rzero$ & $\rzero$ & $\rone$ & $\omega$ \\
\hline
$\rone$ & $\rone$ & $\omega$ & $\omega$ \\
\hline
$\omega$ & $\omega$ & $\omega$ & $\omega$ \\
\hline
\end{tabular}
\BigSpace
\begin{tabular}{|c|c|c|c|}
\hline
$\rmul$ & $\rzero$ & $\rone$ & $\omega$ \\
\hline
$\rzero$ & $\rzero$ & $\rzero$ & $\rzero$ \\
\hline
$\rone$ & $\rzero$ & $\rone$ & $\omega$ \\
\hline
$\omega$ & $\rzero$ & $\omega$ & $\omega$ \\
\hline
\end{tabular}
\end{center}
\end{small}

The typing rules use three operators on contexts: \emph{preorder} $\ctxord$,  \emph{sum} $\ctxsum$ and \emph{multiplication} $\ctxmul$ of a scalar coeffect with a context.  In the example, these operators are defined by first taking, 
on coeffect contexts, the pointwise application of the corresponding scalar operator, with the warning that the inverse preorder on scalars is  used,  see rule \refToRule{t-sub} below. Then, they are lifted  to type-and-coeffect contexts,  resulting in  the following definitions: 

\begin{itemize}
\item $\Gamma\ctxord\Delta$ is the preorder defined by
\begin{quote}
$(\rzero \ctxmul \Delta),\Gamma\ctxord\Gamma$\BigSpace\BigSpace\BigSpace
$(\VarTypeCoeffect{\x}{\T}{\coeff}, \Gamma) \ctxord(\VarTypeCoeffect{\x}{\T}{\coeff'}, \Delta)$ if $\coeff'\rord\coeff$  and $\Gamma\ctxord\Delta$ 
\end{quote}
\item $ \coeff \ctxmul \Gamma$ is the context defined by
\begin{quote}
$ \coeff \ctxmul \EmptyCtx = \EmptyCtx$ \BigSpace\BigSpace\BigSpace $\coeff \ctxmul (\VarTypeCoeffect{\x}{\T}{\coeff'},\Gamma) =  \VarTypeCoeffect{\x}{\T}{\coeff \rmul \coeff'}, (\coeff \ctxmul \Gamma)$
\end{quote}
\item $\Gamma\ctxsum\Delta$ is the context defined by
\begin{quote}
$\EmptyCtx \ctxsum\Gamma = \Gamma \BigSpace\BigSpace\BigSpace
(\VarTypeCoeffect{\x}{\T}{\coeff}, \Gamma) \ctxsum\Delta = \VarTypeCoeffect{\x}{\T}{\coeff}, (\Gamma \ctxsum\Delta)$ if $\x \notin \dom{\Delta} \\
(\VarTypeCoeffect{\x}{\T}{\coeff}, \Gamma) \ctxsum ( \VarTypeCoeffect{\x}{\T}{\coeff'}, \Delta) = \VarTypeCoeffect{\x}{\T}{\coeff \msum \coeff'}, (\Gamma \ctxsum\Delta)
$
\end{quote}
\end{itemize}
Note that when lifted to type-and-coeffect contexts the sum becomes partial, since we require a common variable to have the same type. 

In rule \refToRule{t-const} no variable is used. In rule \refToRule{t-var}, the coeffect context is one of those representing the map where the given variable is used exactly once, and no other is used.  Indeed, $\rzero\mmul\Gamma$ is a context where all variables have $\rzero$ coeffect.   
We include, to show the role of $\ctxord$, a standard subsumption rule \refToRule{t-sub}, allowing a well-typed expression to be typed in a more specific context, where coeffects are overapproximated.\footnote{Note that this rule partly overlaps with \refToRule{t-var}.}
This rule becomes useful, e.g., in  the presence of a conditional construct,  as its typing rule usually requires the two branches to be typed in the same context (that is, to use resources in the same way) and subsumption relaxes this condition. 
In rule \refToRule{t-abs}, the type of a lambda expression is decorated with the coeffect assigned to the binder when typechecking the body.
In rule  \refToRule{t-app}, the coeffects of an application are obtained  by summing  the coeffects of the first subterm, which is expected to have a functional type decorated with a coeffect,  and the coeffects of the argument multiplied by the decoration of the functional type. 
The part emphasized in  gray,  which shows the use of the join operator, needs to be added in a call-by-value strategy. For instance, without this addition, the judgment $\IsWFExp{\VarTypeCoeffect{\y}{\INT}{\rzero}}{\AppExp{(\LambdaExpWithType{\x}{\INT}{\num})}{\y}}{\INT}$ holds, meaning that $\y$ is not actually used, whereas it is used in call-by-value.

Extrapolating from the example, we can distill the following ingredients of a coeffect system:
\begin{itemize}
\item The typing rules use three operators on contexts (preorder, sum, and scalar multiplication) defined on top of the corresponding scalar operators.
\item Coeffects are computed bottom-up, starting from the rule for variable.
\item As exemplified in \refToRule{t-app}, the coeffects of a compound term are computed by a \emph{linear combination} (through sum and scalar multiplication) of those of the subterms. 
The coefficients are determined by the specific language construct considered in the typing rule.
\item The preorder is used for overapproximation.
\end{itemize}
Note also that, by just changing the semiring of scalars, we obtain a different coeffect system. For instance, an easy variant is to consider the natural numbers (with the usual preorder, sum, and multiplication) as scalar coeffects, tracking \emph{exactly how many times} a variable is used.    
The definition of contexts and their operations, and the typing rules, can be kept exactly the same. 

In the following section, we will provide a formal account of the ingredients described above.


\subsection{The algebra of coeffects}
\label{sect:sring-mod}

 As illustrated in the previous section, the first ingredient is a \emph{(preordered) semiring}, whose elements abstract a ``measure'' of resource usage.   

\begin{definition}[Semiring] \label{def:semiring}
A \emph{(preordered) semiring} is a tuple 
$\RR = \ple{\RSet,\rord,\rsum,\rmul,\rzero,\rone}$ where 
\begin{itemize}
\item \ple{\RSet,\rord} is a preordered set 
\item \ple{\RSet,\rsum,\rzero} is an ordered commutative monoid 
\item \ple{\RSet,\rmul,\rone} is an ordered monoid
\end{itemize}
such that the following equalities hold for all $\rel,\arel,\brel\in\RSet$
\begin{align*}
(\rel\rsum\arel)\rmul\brel &= (\rel\rmul\brel)\rsum(\arel\rmul\brel) 
&
\rel\rmul(\arel\rsum\brel) &= (\rel\rmul\arel)\rsum(\rel\rmul\brel) \\ 
\rel\rmul\rzero &= \rzero 
& 
\rzero\rmul\rel &= \rzero 
\end{align*}
\end{definition}
Spelling out the definition, this means that 
both $\rsum$ and $\rmul$ are associative and monotone with respect to $\rord$, and $\rsum$ is also commutative. 
In the following we will adopt the usual precedence rules for addition and multiplication. 

Let us assume a semiring $\RR = \ple{\RSet,\rord,\rsum,\rmul,\rzero,\rone}$ throughout this section. 
 Again, as exemplified in  the  previous section,  coeffect contexts have a preorder, a sum with a neutral element and a multiplication by elements of the semiring. 
Formally, they form a \emph{module over the semiring}, 
as already observed by \citet{McBride16} and \citet{WoodA22}.   

\begin{definition}[\RR-module] \label{def:module}
A \emph{(preordered) \RR-module} \MM is a tuple \ple{\MSet,\mord,\msum,
\mzero,\mmul} where 
\begin{itemize}
\item \ple{\MSet,\mord} is a preordered set 
\item \ple{\MSet,\msum,\mzero} is a commutative monoid 
\item \fun{\mmul}{\RSet\times\MSet}{\MSet} is a function, called \emph{scalar multiplication}, which is monotone in both arguments and satisfies  the following equalities:
\begin{align*}
(\rel\rsum\arel)\mmul\mel &= (\rel\mmul\mel)\msum(\arel\mmul\mel) 
  & \rel\mmul(\mel\msum\amel) &= (\rel\mmul\mel)\msum(\rel\mmul\amel) 
    & (\rel\rmul\arel)\mmul\mel &= \rel\mmul(\arel\mmul\mel)  \\ 
\rzero \mmul \mel &= \mzero 
  & \rel \mmul \mzero &= \mzero 
    & \rone \mmul \mel &= \mel 
\end{align*}
\end{itemize}
Given \RR-modules \MM and \aMM, a \emph{(lax) homomorphism} \fun{f}{\MM}{\aMM} is a monotone function \fun{f}{\MSet}{\aMSet} such that the following hold for all $\mel,\amel\in\MSet$ and $\rel \in \RSet$: 
\[  f(\mel)\msum f(\amel) \mord f(\mel\msum\amel) 
\qquad 
f(\rel\mmul\mel) = \rel \mmul f(\mel) \] 
\end{definition}
From the second equality it follows that 
$\mzero = f(\mzero)$ as 
$\mzero = \rzero \mmul f(\mzero) = f(\rzero\mmul\mzero) = f(\mzero)$. 
It is also easy to see that \RR-modules and their homomorphisms form a category, denoted by \RMod\RR. 
 Note that \citet{WoodA22} use a different notion of homomorphism built on relations. 
Here we prefered to stick to a more standard functional notion of homomorphism.
The comparison between these two notions is an interesting topic for future work.  

We show that the coeffects of the example in  the  previous section, and in general any structural coeffects, form an \RR-module. 

Let $X$ be a set and \fun{\alpha}{X}{\RSet} be a function. 
The \emph{support} of $\alpha$ is the set 
${\supp\alpha = \{ x \in X \mid \alpha(x) \ne \rzero \}}$. 
 Denote by $\RSet^X$ the set of functions \fun{\alpha}{X}{\RSet} with finite support, then we can define the \RR-module 
$\RR^X = \ple{\RSet^X,\hat\rord,\hat\rsum,\hat\rzero,\hat\rmul}$ where 
$\hat\rord$ and $\hat\rsum$ are the pointwise extension of $\rord$ and $\rsum$ to $\RSet^X$, 
$\hat\rzero$ is the constant function equal to $\rzero$ and 
$\rel\hat\rmul\alpha = x\mapsto \rel\rmul\alpha(x)$, for all $\rel\in\RSet$ and $\alpha\in\RSet^X$. 
Note that $\hat\rzero$, $\hat\rsum$ and $\hat\rmul$  are well-defined because 
$\supp{\hat\rzero} = \emptyset$,  
$\supp{\alpha \hat\rsum \beta} \subseteq \supp\alpha \cup \supp\beta$ and 
$\supp{\rel\hat\rmul\alpha}\subseteq \supp\alpha$. 
When $X$ is the set of variables, $\RR^X$  (with the inverse preorder)  is precisely the module of coeffect contexts in the structural case: 
they assign to each variable an element of the semiring and the requirement of finite support ensures that only finitely many variables have non-zero coeffect.  

%
%
 
Finally, the coeffect systems considered in this paper additionally assume that the preordered semiring, hence the associated module,  has binary joins. 
Since this assumption is completely orthogonal to the development in this section, we have omitted it. 
However, all definitions and results  also work  in presence of binary joins, hence they can be  added without issues.


\section{Coeffects for sharing}\label{sect:sharing}

Introducing sharing, e.g. by  a field  assignment $\FieldAssign{\x}{\f}{\y}$, can be clearly seen as adding an arc between $\x$ and $\y$ in an undirected graph where nodes are variables. However, such  a graphical   representation  would be a  \emph{global} one, whereas the representation we  are looking for must  be \emph{per variable}, and, moreover,  must support sum and scalar multiplication operators.
To achieve this, we introduce auxiliary entities called \emph{links}, and attach to each variable a set of them, so that an arc between $\x$ and $\y$ is represented by the fact that they have a common link.\footnote{This roughly corresponds to the well-known representation of a (hyper)graph by a bipartite graph.}  Moreover, there is a special link $\res$ which denotes a connection with the final result of the expression. 

For instance, considering again the classes of \refToExample{ex1}: 
\begin{lstlisting}
class B {int f;}
class C {B f1; B f2;}
\end{lstlisting}

and the program $\Seq{\FieldAssign{\xtt}{\fttOne}{\ytt}}{\ConstrCall{\Ctt}{\zttOne,\zttTwo}}$, the following typing judgment will be derivable: 
\begin{small}
\begin{quote}
$(\ast)\ \IsWFExp{\VarTypeCoeffect{\xtt}{\Ctt}{\{ \link\}}, \VarTypeCoeffect{\ytt}{\Btt}{\{ \link\}},\VarTypeCoeffect{\zttOne}{\Btt}{\{\res\}},\VarTypeCoeffect{\zttTwo}{\Btt}{\{ \res \}}}{\Seq{\FieldAssign{\xtt}{\fttOne}{\ytt}}{\ConstrCall{\Ctt}{\zttOne,\zttTwo}}}{\Ctt}$ \Space with $\link\neq\res$
\end{quote}
\end{small}
meaning that the program's execution introduces sharing between $\xtt$ and $\ytt$, as expressed by their common link $\link$, and between $\zttOne$, $\zttTwo$, and the final result, as expressed by their common link $\res$. 
The derivation for this judgment is shown later (\refToFigure{derivation}).

Formally, we  assume a countable set $\Lnk$, ranged over by $\link$, with a distinguished element $\res$. In the coeffect system for sharing, scalar coeffects $\X,\Y,$ and $\Z$ will be finite sets of links.  Let  $\LnkSet$ be  the finite powerset of $\Lnk$,  that is,  the set of scalar coeffects, and  let $\shCCtx$ be  the set of the corresponding coeffect contexts $\cctx$, that is (representations of) maps in  $\LnkSet^\Vars$, with $\Vars$ the set of variables.
Given $\cctx=\VarCoeffect{\x_1}{\X_1}, \ldots, \VarCoeffect{\x_n}{\X_n}$, the \emph{(transitive) closure} of $\cctx$, denoted $\closure{\cctx}$, is ${\VarCoeffect{\x_1}{\closure{\X}_1}, \ldots, \VarCoeffect{\x_n}{\closure{\X}_n}}$ where $\closure{\X}_1, \ldots,\closure{\X}_n$ are the smallest sets such that:
\begin{quote}
$\link\in\X_i$ implies $\link\in\closure{\X}_i$\\
$\link,\link'\in\closure{\X}_i$, $\link'\in\closure{\X}_j$ implies $\link\in\closure{\X}_j$
\end{quote}
For instance, if $\cctx=\VarCoeffect{\x}{\{\link\}},\VarCoeffect{\y}{\{\link,\link'\}},\VarCoeffect{\z}{\{\link'\}}$, then $\closure{\cctx}=\VarCoeffect{\x}{\{\link,\link'\}},\VarCoeffect{\y}{\{\link,\link'\}},\VarCoeffect{\z}{\{\link,\link'\}}$. 
That is, since $\x$ and $\y$ are connected by $\link$, and $\y$ and $\z$ are connected by $\link'$, then $\x$ and $\z$ are connected as well. 
Note that, if $\cctx$ is \emph{closed} ($\closure{\cctx}=\cctx$), then two variables have either the same, or disjoint coeffects. 

To sum two closed coeffect contexts, obtaining in turn a closed one, we need to apply the transitive closure after pointwise union. For instance, the above coeffect context $\cctx$ could have been obtained as pointwise union of $\VarCoeffect{\x}{\{\link\}},\VarCoeffect{\y}{\{\link\}}$ and $\VarCoeffect{\y}{\{\link'\}},\VarCoeffect{\z}{\{\link'\}}$.

Multiplication of a closed coeffect context with a scalar is defined in terms of an operator $\shmul$ on sharing coeffects, which replaces the $\res$ link (if any) in the second argument with the first:
\begin{quote} 
$\X\shmul\Y=\begin{cases} 
\emptyset & \mbox{if}\ \X=\emptyset  \\ 
\Y & \mbox{if}\ \X\ne\emptyset\ \mbox{and}\ \res\not\in\Y\\
(\Y\setminus\{ \res\})\cup\X&\mbox{if}\ \X\ne\emptyset\ \mbox{and}\ \res\in\Y 
\end{cases}$ 
\end{quote}

Similarly to sum, to  multiply  a coeffect context with a scalar $\X$, we need to apply the transitive closure after pointwise application of the operation $\shmul$. 
For instance, ${\{\link''\}\ctxmul(\VarCoeffect{\x}{\{\link,\res\}},\VarCoeffect{\y}{\{\link'\}})}=\VarCoeffect{\x}{\{\link,\link''\}},\VarCoeffect{\y}{\{\link'\}}$. To see that transitive closure can be necessary, consider, for instance, ${\{\link''\}\ctxmul(\VarCoeffect{\x}{\{\link,\res\}},\VarCoeffect{\y}{\{\link''\}})}=\VarCoeffect{\x}{\{\link,\link''\}},\VarCoeffect{\y}{\{\link,\link''\}}$.

When an expression $\e$, typechecked with context $\Gamma$, replaces a variable with coeffect $\X$ 
in an expression  $\e'$, the product $\X\ctxmul\Gamma$ computes the sharing 
introduced by the resulting expression on the variables in $\Gamma$.
For instance,  set $\e={\Seq{\FieldAssign{\xtt}{\fttOne}{\ytt}}{\ConstrCall{\Ctt}{\zttOne,\zttTwo}}}$ of $(\ast)$ and assume that $\e$ replaces  $\ztt$  in $\FieldAssign{\ztt}{\fttOne}{\wtt}$, for which the judgment
$\IsWFExp{\VarTypeCoeffect{\ztt}{\Ctt}{\{\res\}}, \VarTypeCoeffect{\wtt}{\Btt}{\{\res\}}}{\FieldAssign{\ztt}{\fttOne}{\wtt}}{\Btt}$ is derivable.
We expect that $\zttOne$ and $\zttTwo$, being connected to the result of $\e$,  are  connected to whatever
$\ztt$ is connected  to  ($\wtt$  and the result of $\FieldAssign{\ztt}{\fttOne}{\wtt}$), whereas the sharing of $\xtt$ and $\ytt$ would not be changed. 
In our example, we have $\{ \res\}\shmul\{ \link\}=\{ \link\}$ and  $\{ \res\}\shmul\{ \res\}=\{ \res\}$. Altogether we have the following formal definition:

\begin{definition}\label[Sharing coeffects]{def:sharing-coeff}
The \emph{sharing coeffect system} is defined by:
\begin{itemize}
\item
the semiring $\SharingScalar=\Tuple{\LnkSet,\subseteq,\cup,\shmul,\emptyset,\{\res\}}$

\item the $\SharingScalar$-module \ple{\shCCtx_\clo,\shord,\shsum,\emptyset,\ctxmul} where:
\begin{itemize}
\item $\shCCtx_\clo$ are the fixpoints of $\clo$, that is, the closed coeffect contexts
\item $\shord$ is the pointwise extension of  $\subseteq$  to $\shCCtx_\clo$ 
\item $\Gamma \shsum \Gamma'=\closure{(\Gamma \ \hat{\cup} \ \Gamma')}$, where $\hat{\cup}$ is the pointwise extension of $\cup$ to $\shCCtx_\clo$ 
\item $\X \ctxmul \Gamma = \closure{(\X \hatshmul \Gamma)}$, where $\hatshmul$ is the pointwise extension of $\shmul$ to $\shCCtx_\clo$.
\end{itemize}
\end{itemize}
\end{definition}

Operations on closed coeffect contexts can be lifted to type-and-coeffect contexts, exactly as we did in the introductory example in \refToSection{informal}.

It is easy to check that $\SharingScalar=\Tuple{\LnkSet,\subseteq,\cup,\shmul,\emptyset,\{\res\}}$ is actually a semiring with $\emptyset$ neutral element of $\ \cup \ $ and $\{\res\}$ neutral element of $\shmul$.  The fact that $\ple{\shCCtx_\clo,\shord,\shsum,\emptyset,\ctxmul}$ is actually an $\SharingScalar$-module can be proved as follows: first of all, $\SharingScalar^\Vars=\ple{\shCCtx,\shord,\hat{\cup},\emptyset,\hatshmul}$ is an $\SharingScalar$-module, notably, the structural one (all operations are pointwise); it is easy to see that $\closure{\_}$ is an idempotent homomorphism on $\SharingScalar^\Vars$; then, the thesis follows from \bez \refToProposition{module-fix} in the Appendix\eez, stating that an idempotent homomorphism on a module induces a module structure on the set of its fixpoints. 

In a judgment $\IsWFExp{\Gamma}{\e}{\T}$, the coeffects in $\Gamma$ describe an equivalence relation on $\dom{\Gamma}\cup\{\res\}$ where each coeffect corresponds to an equivalence class. 
Two variables, say $\x$ and $\y$, have the same coeffect if the evaluation of $\e$ possibly introduces sharing between $\x$ and $\y$. 
Moreover, $\res$ in the coeffect of $\x$ models possible sharing with the final result of $\e$. 
Intuitively, sharing only happens among variables of reference types (classes), since a variable $\x$ of a primitive type $\PT$ denotes an immutable value rather than a reference in memory. To have a uniform treatment, a judgment $\IsWFExp{\VarTypeCoeffect{\x}{\PT}{\{\link\}}}{\x}{\PT}$ with $\link$ fresh is derivable (by rules \refToRule{t-var} and \refToRule{t-prim}, as detailed below\footnote{ Alternatively, variables of primitive types could be in a separate context, with no sharing coeffects.}). 

The typing rules are given in \refToFigure{typing-sharing}. 
\begin{figure}
\begin{small}
\begin{grammatica}
\produzione{\Gamma,\Delta}{\VarTypeCoeffect{\x_1}{\T_1}{\X_1},\ldots,\VarTypeCoeffect{\x_n}{\T_n}{\X_n}}{context}\\
\produzione{\X}{\{\link_1,\ldots,\link_n\}}{coeffect (set of links)}
\\[1ex]
\end{grammatica}

\hrule

\begin{math}
\begin{array}{l}
\\[0.5ex]
 \NamedRule{t-var}{}{ \IsWFExp{ \emptyset\ctxmul \Gamma \shsum\VarTypeCoeffect{\x}{\T}{\{\res\}}}{\x}{\T}}{
}
\BigSpace
\NamedRule{t-const}{}{\IsWFExp{\emptyset \ctxmul \Gamma}{\const}{\PT_\const}}{}

\\[4ex] 
\NamedRule{t-field-access}{\IsWFExp{\Gamma}{\e}{\C}}{\IsWFExp{\Gamma}{\FieldAccess{\e}{\f_i}}{\T_i}}
{
\fields{\C}=\Field{\T_1}{\f_1} \ldots \Field{\T_n}{\f_n} \\
i\in 1..n\\
}
\\[4ex]
\NamedRule{t-field-assign}{\IsWFExp{\Gamma}{\e}{\C} \Space \IsWFExp{\Delta}{\e'}{\T_i}}{\IsWFExp{\Gamma \shsum \Delta}{\FieldAssign{\e}{\f_i}{\e'}
}{\T_i}}
{
\fields{\C}=\Field{\T_1}{\f_1} \ldots \Field{\T_n}{\f_n}\\
  i\in 1..n\\ 
}
\\[4ex]
\NamedRule{t-new}{\IsWFExp{\Gamma_i}{\e_i}{\T_i}\Space \forall i\in 1..n}{\IsWFExp{\Gamma_1 \shsum \ldots \shsum \Gamma_n}{\ConstrCallTuple{\C}{\e}{n}}{\C}}
{\fields{\C}=\Field{\T_1}{\f_1} \ldots \Field{\T_n}{\f_n}}
\\[4ex]

\NamedRule{t-invk}{\IsWFExp{\Gamma_0}{\e_0}{\C} \BigSpace \IsWFExp{\Gamma_i}{\e_i}{\T_i}\Space \forall i\in 1..n}{\IsWFExp{(\X_0\cup\{\link_0\}) \ctxmul \Gamma_0)  \shsum \ldots \shsum (\X_n\cup\{\link_n\})\ctxmul \Gamma_n)}{\MethCallTuple{\e_0}{\m}{\e}{n}}{\T}}{
\mtype{\C}{\m}\eqfresh\funType{\X_0,\coeffectType{\T_1}{\X_1} \ldots \coeffectType{\T_n}{\X_n}}{\T}\\
\link_0,\ldots,\link_n\ \mbox{fresh}
}

\\[4ex]
\NamedRule{t-block}{
\IsWFExp{\Gamma}{\e}{\T} \BigSpace
\IsWFExp{\Gamma', \VarTypeCoeffect{\x}{\T}{\X}}{\e'}{\T'}
}
{
\IsWFExp{(\X \cup \{ \link \})\ctxmul\Gamma \shsum \Gamma'}{\Block{\T}{\x}{\e}{\e'}}{\T'}
}
{\link \ \text{fresh}}
\BigSpace
\NamedRule{t-prim}{\IsWFExp{\Gamma}{\e}{\PT}}{\IsWFExp{\{\link\}\ctxmul\Gamma}{\e}{\PT}}
{
\link \ \text{fresh}
}
\\[3ex]
\end{array}
\end{math}

\hrule

\begin{math}
\begin{array}{l}
\\
\NamedRule{t-conf}{\IsWFExp{\Delta}{\e}{\T} \BigSpace \IsWFMem{\Gamma}{\mem}}
{\IsWFConf{\Delta\shsum\Gamma}{\e}{\mem}{\T}} 
{\dom{\Delta}\subseteq\dom{\Gamma}}
\\[4ex]

\NamedRule{t-obj}{\IsWFExp{\Gamma_i}{\val_i}{\T_i}\Space \forall i\in 1..n}{\IsWFObject{\Gamma_1 \shsum \cdots \shsum \Gamma_n}{\Object{\C}{\val_1,\ldots,\val_n}}{\C}}{\fields{\C}=\Field{\T_1}{\f_1} \ldots \Field{\T_n}{\f_n}}

\\[4ex]
\NamedRule{t-mem}{\IsWFExp{\Gamma_i}{\mem(\x_i)}{\C_i\Space\forall i\in 1...n}
}
{\IsWFMem{\Gamma_{\!\mem}\shsum\Gamma}{\mem}} 
{\Gamma_{\!\mem}=\VarTypeCoeffect{\x_1}{\C_1}{\{\link_1\}},\ldots,\VarTypeCoeffect{\x_n}{\C_n}{\{\link_n\}}\\
\dom{\Gamma_{\!\mem}}=\dom{\mem}\\
\Gamma=(\{\link_1\}\ctxmul\Gamma_1) \shsum \ldots \shsum (\{\link_n\}\ctxmul\Gamma_n)\\
\link_1,\ldots,\link_n\ \mbox{fresh}}
\end{array}
\end{math}
\end{small}
\caption{Coeffect system for sharing}\label{fig:typing-sharing}
\end{figure}
 In the rule for variable, the variable is obviously linked with the result (they coincide), hence its coeffect  is $\{\res\}$. 
In rule \refToRule{t-const}, no variable is used. 

In rule \refToRule{t-field-access}, the coeffects are those of the receiver expression. In rule \refToRule{t-field-assign}, the coffects of the two arguments are summed. 
In particular, the result of the receiver expression, of the right-side expression, and the final result, will be in  sharing.  For instance, we derive ${\IsWFExp{\VarTypeCoeffect{\xtt}{\Ctt}{\{\res\}}, \VarTypeCoeffect{\ytt}{\Btt}{\{ \res\}}}{\FieldAssign{\xtt}{\fttOne}{\ytt}}{\Btt}}$. 
In rule \refToRule{t-new}, analogously, the coeffects of the arguments of the constructor are summed. In particular, the results of the argument expressions and the final result will be in sharing.  For instance, we derive $\IsWFExp{\VarTypeCoeffect{\zttOne}{\Btt}{\{\res\}},\VarTypeCoeffect{\zttTwo}{\Btt}{\{\res\}}}{\ConstrCall{\Ctt}{\zttOne,\zttTwo}}{\Ctt}$. 

In rule \refToRule{t-invk}, the coeffects of the arguments are summed, after multiplying each of them with the coeffect of the corresponding parameter, where, to avoid clashes, we assume that links different from $\res$ are freshly renamed,  as indicated   by the notation $\eqfresh$. Moreover, a fresh link $\link_i$ is added\footnote{ Analogously to the rule \refToRule{t-app} in \refToFigure{coeff-lambda} in the call-by-value case.}, since otherwise, if the parameter is not used in the body (hence has empty coeffect), the links of the argument would be lost in the final context, see the example for rule \refToRule{t-block} below. 

The auxiliary function $\aux{mtype}$  now  returns an enriched method type,  where the parameter types are decorated with their coeffects, including the implicit parameter $\this$. 
 The condition that method bodies should be well-typed with respect to method types is extended by requiring that coeffects computed by typechecking the method body express no more sharing than those in the method type, formally: if $\mbody{\C}{\m}$ and $\mtype{\C}{\m}$ are defined, then  $\mbody{\C}{\m}=\Pair{\x_1\dots\x_n}{\e}$, ${\mtype{\C}{\m}=\funType{\X_0,\coeffectType{\T_1}{\X_1} \ldots \coeffectType{\T_n}{\X_n}}{\T}}$, and  
\begin{quote}
$\IsWFExp{\VarTypeCoeffect{\this}{\Ctt}{\X'_0},\VarTypeCoeffect{\x_1}{\T_1}{\X'_1},\ldots,\VarTypeCoeffect{\x_n}{\T_n}{\X'_n}}{\e}{\T}$\\[1ex]
$\X'_i=\X'_j\neq\emptyset$ implies $\X_i=\X_j\neq\emptyset$
\end{quote}
holds. As an example, consider the following method:
\begin{lstlisting}
class B {int f;}
class C {B f1; B f2;
  C m(B y, B z1, B z2) {this.f1=y; new C(z1,z2)}
}
\end{lstlisting}
where $\mtype{\Ctt}{\texttt{m}}{=}\funType{\{\link\},\coeffectType{\Btt}{\{\link\}}, \coeffectType{\Btt}{\{\res\}},\coeffectType{\Btt}{\{\res\}}}{\Ctt}$, with $\link{\neq}\res$.
The method body is well-typed, since we derive $\IsWFExp{\VarTypeCoeffect{\this}{\Ctt}{\{\link\}}, \VarTypeCoeffect{\ytt}{\Btt}{\{ \link\}},\VarTypeCoeffect{\zttOne}{\Btt}{\{\res\}},\VarTypeCoeffect{\zttTwo}{\Btt}{\{\res\}}}{\Seq{\FieldAssign{\xtt}{\fttOne}{\ytt}}{\ConstrCall{\Ctt}{\zttOne,\zttTwo}}}{\Ctt}$\mbox{, with $\link{\neq}\res$. }
 
Consider now the method call \lstinline{x.m(z,y1,y2)}. We get the following derivation:

{\footnotesize
\[
\NamedRule{t-invk}
{
	\NamedRule{t-var}
	{}
	{\IsWFExp{\VarTypeCoeffect{\xtt}{\Ctt}{\{\res\}}}{\xtt}{\Ctt}}
	{}
	\quad
	\NamedRule{t-var}
	{}
	{\IsWFExp{\VarTypeCoeffect{\ztt}{\Btt}{\{\res\}}}{\ztt}{\Btt}}
	{}
	\quad
	\NamedRule{t-var}
	{}
	{\IsWFExp{\VarTypeCoeffect{\ytt_1}{\Btt}{\{\res\}}}{\ytt_1}{\Btt}}
	{}
	\quad
	\NamedRule{t-var}
	{}
	{\IsWFExp{\VarTypeCoeffect{\ytt_2}{\Btt}{\{\res\}}}{\ytt_2}{\Btt}}
	{}
}
{\IsWFExp{\VarTypeCoeffect{\xtt}{\Ctt}{\X},\VarTypeCoeffect{\ztt}{\Btt}{\X},\VarTypeCoeffect{\ytt_1}{\Btt}{\Y},\VarTypeCoeffect{\ytt_2}{\Btt}{\Y}}{\MethCall{\xtt}{\texttt{m}}{\ztt,\ytt_1,\ytt_2}}{\Ctt}}
{}
\]
}
where $\X=\{\link',\link_{0},\link_{1}\}$ and $\Y=\{\res,\link_{2},\link_{3}\}$.

 The context of the call is obtained as follows
{\small\[
\begin{array}{lcl}
&&\{\link',\link_{0}\}\ctxmul(\VarTypeCoeffect{\xtt}{\Ctt}{\{\res\}})\ctxsum
\{\link',\link_{1}\}\ctxmul(\VarTypeCoeffect{\ztt}{\Btt}{\{\res\}})\ctxsum
\{\res,\link_{2}\}\ctxmul(\VarTypeCoeffect{\ytt_1}{\Btt}{\{\res\}})\ctxsum
\{\res,\link_{3}\}\ctxmul(\VarTypeCoeffect{\ytt_2}{\Btt}{\{\res\}})\\
&=&(\VarTypeCoeffect{\xtt}{\Ctt}{\{\link',\link_{0}\}})\ctxsum
(\VarTypeCoeffect{\ztt}{\Btt}{\{\link',\link_{1}\}})\ctxsum
(\VarTypeCoeffect{\ytt_1}{\Btt}{\{\res,\link_{2}\}})\ctxsum
(\VarTypeCoeffect{\ytt_2}{\Btt}{\{\res,\link_{3}\}})\\
&=&\VarTypeCoeffect{\xtt}{\Ctt}{\X},\VarTypeCoeffect{\ztt}{\Btt}{\X},\VarTypeCoeffect{\ytt_1}{\Btt}{\Y},\VarTypeCoeffect{\ytt_2}{\Btt}{\Y}
\end{array}
\]
}
where $\ell'$ is a fresh renaming of the (method) link $\ell$, and $\ell_i$, $0\leq i\leq 3$, are fresh links.

For a call \lstinline{x.m(z,z,y)}, instead, we get the following derivation:
{\footnotesize
\[
\NamedRule{t-invk}
{
	\NamedRule{t-var}
	{}
	{\IsWFExp{\VarTypeCoeffect{\xtt}{\Ctt}{\{\res\}}}{\xtt}{\Ctt}}
	{}
	\quad
	\NamedRule{t-var}
	{}
	{\IsWFExp{\VarTypeCoeffect{\ztt}{\Btt}{\{\res\}}}{\ztt}{\Btt}}
	{}
	\quad
	\NamedRule{t-var}
	{}
	{\IsWFExp{\VarTypeCoeffect{\ztt}{\Btt}{\{\res\}}}{\ztt}{\Btt}}
	{}
	\quad
	\NamedRule{t-var}
	{}
	{\IsWFExp{\VarTypeCoeffect{\ytt}{\Btt}{\{\res\}}}{\ytt}{\Btt}}
	{}
}
{\IsWFExp{\VarTypeCoeffect{\xtt}{\Ctt}{\X},\VarTypeCoeffect{\ztt}{\Btt}{\X},\VarTypeCoeffect{\ytt}{\Btt}{\X}}{\MethCall{\xtt}{\texttt{m}}{\ztt,\ztt,\ytt}}{\Ctt}}
{}
\]
}
where $\X=\{\link',\link_{0},\link_{1},\link_{2},\link_{3},\res\}$. That is, $\xtt$, $\ytt$, $\ztt$, and the result, are in sharing (note the role of the transitive closure here).
\\
In the examples that follow we will omit the fresh links unless necessary.

In rule \refToRule{t-block}, the coeffects of the expression in the declaration are multiplied by  the join (that is, the union) of those of the local variable in the body and the singleton of a fresh link, and then summed with those of the body.  The union with the fresh singleton is needed  when the variable is not used in the body (hence has empty coeffect), since otherwise its links, that is, the information about its sharing in $\e$, would be lost in the final context.  For instance,  consider the body of method \lstinline{m} above, which is an abbrevation for \lstinline{B unused = (this.f1=y); new (z1, z2)}. Without  the join with the fresh singleton, we could derive the judgment $\IsWFExp{\VarTypeCoeffect{\this}{\Ctt}{\emptyset}, \VarTypeCoeffect{\ytt}{\Btt}{\emptyset},\VarTypeCoeffect{\zttOne}{\Btt}{\{\res\}},\VarTypeCoeffect{\zttTwo}{\Btt}{\{\res\}}}{\texttt{B unused = (this.f1=y); new (z1, z2)}}{\Ctt}$, where the information that after the execution of the field assignment $\this$ and $\ytt$ are in sharing is lost. 
 
 Rule  \refToRule{t-prim}  allows the coeffects of an expression of primitive type to be changed  by removing the links with the result, as formally modeled by the product of the context with a fresh singleton coeffect. For instance, the following derivable judgment 
\begin{quote}
$\IsWFExp{\VarTypeCoeffect{\zttOne}{\Btt}{\{\link\}},\VarTypeCoeffect{\zttTwo}{\Btt}{\{ \link\}}}{\FieldAccess{\FieldAccess{\ConstrCall{\Ctt}{\zttOne,\zttTwo}}{\fttOne}}{\ftt}}{\intType}$, with $\link\neq\res$
\end{quote}
shows that there is no longer a link between the result and $\zttOne,\zttTwo$.

In rule \refToRule{t-conf}, the coeffects of the expression and those of the memory are summed. 
In rule \refToRule{t-mem}, a memory is well-formed in a context which is the sum of two parts.
The former assigns a type to all and only references in memory, as in the standard rule in \refToFigure{calculus}, and a fresh singleton coeffect. 
The latter sums the coeffects of the objects in memory, after multiplying each of them with that of the corresponding reference. 
For instance, \label{example-t-mem} for $\xtt\mapsto\Object{\Att}{\ytt},\ytt\mapsto\Object{\Btt}{0},\ztt\mapsto\Object{\Att}{\ytt}$, the former context is $\VarTypeCoeffect{\xtt}{\Att}{\{\link_\xtt\}},\VarTypeCoeffect{\ytt}{\Btt}{\{\link_\ytt\}},\VarTypeCoeffect{\ztt}{\Att}{\{\link_\ztt\}}$, the latter is the sum of the three contexts $\VarTypeCoeffect{\ytt}{\Att}{\{\link_\xtt\}}$, $\emptyset$, and $\VarTypeCoeffect{\ytt}{\Att}{\{\link_\ztt\}}$.
Altogether, we get $\VarTypeCoeffect{\xtt}{\Att}{\{\link_\xtt,\link_\ytt,\link_\ztt\}},\VarTypeCoeffect{\ytt}{\Btt}{\{\link_\xtt,\link_\ytt,\link_\ztt\}},\VarTypeCoeffect{\ztt}{\Att}{\{\link_\xtt,\link_\ytt,\link_\ztt\}}$, expressing that the three references are connected. Note that no $\res$ link occurs in memory; indeed, there is no final result.

As an example of a more involved derivation, consider the judgment
\[
\IsWFExp{\VarTypeCoeffect{\xtt}{\Ctt}{\{\link\}}, \VarTypeCoeffect{\ytt}{\Btt}{\{ \link\}}}{\Block{\Btt}{\ztt}{\ConstrCall{\Btt}{2}}{\Seq{\FieldAssign{\xtt}{\fttOne}{\ytt}}{\ConstrCall{\Ctt}{\ztt,\ztt}}}}{\Ctt}\mbox{ where $\link\neq\res$}.
\]
Here $\Seq{\FieldAssign{\xtt}{\fttOne}{\ytt}}{\ConstrCall{\Ctt}{\ztt,\ztt}}$ is shorthand for $\Block{\Btt}{\wtt}{(\FieldAssign{\xtt}{\fttOne}{\ytt})}{\ConstrCall{\Ctt}{\ztt,\ztt}}$. 
The derivation is in \refToFigure{derivation}, where the subderivations $\der_1$ and $\der_2$ are  given below for space reasons.

\begin{figure}
\begin{small}
$
\NamedRule{t-block}
{
	\NamedRule{t-new}
	{
		\NamedRule{t-const}
		{}
		{\IsWFExp{\emptyset}{2}{\intType}}
		{}
	}
	{\IsWFExp{\emptyset}{\ConstrCall{\Btt}{2}}{\Btt}}
	{}
\BigSpace\NamedRule{t-block}
	{
		 \der_1	
\BigSpace\BigSpace	  	\der_2 
	}
	{\IsWFExp{\Gamma}{\Block{\Btt}{\wtt}{(\FieldAssign{\xtt}{\fttOne}{\ytt})}{\ConstrCall{\Ctt}{\ztt,\ztt}}}{\Ctt}}
	{}
}
{\IsWFExp{\VarTypeCoeffect{\xtt}{\Ctt}{\{\link\}},\VarTypeCoeffect{\ytt}{\Btt}{\{\link\}}}{\Block{\Btt}{\ztt}{\ConstrCall{\Btt}{2}}{\Seq{\FieldAssign{\xtt}{\fttOne}{\ytt}}{\ConstrCall{\Ctt}{\ztt,\ztt}}}}{\Ctt}}
{}
$

\bigskip

\begin{tabular}{l}
$\link,\link'$ fresh\\
$ \VarTypeCoeffect{\xtt}{\Ctt}{\{\link\}},\VarTypeCoeffect{\ytt}{\Btt}{\{\link\}}= (\{\res\}\ctxsum\{\link'\})\ctxmul \emptyset\ctxsum\VarTypeCoeffect{\xtt}{\Ctt}{\{\link\}},\VarTypeCoeffect{\ytt}{\Btt}{\{\link\}}
$\\
$\Gamma = (\emptyset \ctxsum \{\link\})\ctxmul (\VarTypeCoeffect{\xtt}{\Ctt}{\{ \res\}},\VarTypeCoeffect{\ytt}{\Btt}{\{\res\}}) \ctxsum \VarTypeCoeffect{\ztt}{\Btt}{\{\res\}} = \VarTypeCoeffect{\xtt}{\Ctt}{\{ \link\}},\VarTypeCoeffect{\ytt}{\Btt}{\{\link\}},\VarTypeCoeffect{\ztt}{\Btt}{\{\res\}}$
\end{tabular}

\bigskip

$\der_1= \ 
\NamedRule{t-field-assign}
{
	\NamedRule{t-var}
	{}
	{\IsWFExp{\VarTypeCoeffect{\xtt}{\Ctt}{\{\res\}}}{\xtt}{\Ctt}}	
	{}
	\qquad
	\NamedRule{t-var}
	{}
	{\IsWFExp{\VarTypeCoeffect{\ytt}{\Btt}{\{\res\}}}{\ytt}{\Btt}}	
	{}
}
{\IsWFExp{\VarTypeCoeffect{\xtt}{\Ctt}{\{ \res\}},\VarTypeCoeffect{\ytt}{\Btt}{\{\res\}}}{\FieldAssign{\xtt}{\fttOne}{\ytt}}{\Btt}}
{}	
$

\bigskip

$
\der_2=   \ 
\NamedRule{t-new}
{
	\NamedRule{t-var}
	{}
	{\IsWFExp{\VarTypeCoeffect{\wtt}{\Btt}{\emptyset},\VarTypeCoeffect{\ztt}{\Btt}{\{\res\}}}{\ztt}{\Btt}}	
	{}
	\quad
	\NamedRule{t-var}
	{}
	{\IsWFExp{\VarTypeCoeffect{\wtt}{\Btt}{\emptyset},\VarTypeCoeffect{\ztt}{\Btt}{\{\res\}}}{\ztt}{\Btt}}	
	{}
}
{\IsWFExp{\VarTypeCoeffect{\wtt}{\Btt}{\emptyset},\VarTypeCoeffect{\ztt}{\Btt}{\{\res\}}}{\ConstrCall{\Ctt}{\ztt,\ztt}}{\Ctt}}
{}
$
\end{small}
\caption{Example  of derivation}\label{fig:derivation}
\end{figure}

The rules in \refToFigure{typing-sharing} immediately lead to an algorithm which inductively computes the coeffects of an expression. Indeed, all the rules except \refToRule{t-prim} are  syntax-directed, that is, the coeffects of the expression in the consequence are computed as a linear combination of those of the subexpressions,  where the basis is  the rule for variables. 
Rule \refToRule{t-prim} is assumed to be \emph{always} used in the algorithm, just once,  for expressions of primitive types. 

We assume  there are  coeffect annotations in method parameters to handle (mutual) recursion; for non-recursive methods, such coeffects can be computed (that is, in the coherency condition above, the $\X_i$s in $\aux{mtype}$ are exactly the $\X'_i$s). We leave to future work the investigation of a global fixed-point inference to compute coeffects across mutually recursive methods.

Considering again \refToExample{ex2}:
\begin{lstlisting}
class B {int f; B clone $\meta{[^{\{\link\}}]}$() {new B(this.f)} // $\link\neq\res$
class A { B f;
  A mix $\meta{[^{\{\res\}}]}$(A$\meta{^{\{\res\}}}$a) {this.f=a.f; a} // this, a and result linked
  A clone $\meta{[^{\{\link\}}]}$ () {new A(this.f.clone()) } // $\link\neq\res$
}
A a1 = new A(new B(0));
A mycaps = {A a2 = new A(new B(1));
  a1.mix(a2).clone()
  // a1.mix(a2).clone().mix(a2) 
}
\end{lstlisting}
 The parts emphasized in  gray  are the coeffects which can be computed for the parameters by typechecking the body (the coeffect for $\this$ is  in square brackets).
 In a real language, such coeffects  would be  declared by some concrete syntax, as part of the type information available
to clients. 
From such coeffects, a client knows that the result of \lstinline{mix} will be connected to both the receiver and the argument, whereas the result of \lstinline{clone} will be a reference to a \emph{fresh} portion of memory, not connected to the receiver. 
 
 Using  the sharing coeffects, we can discriminate \lstinline{a2.mix(a1).clone()} and  \lstinline{a1.mix(a2).clone().mix(a2)}, as desired. Indeed, for the first \lstinline{mix} call, the judgment $\IsWFExp{\VarTypeCoeffect{\text{\lstinline{a1}}}{\Att}{\{\res\}},\VarTypeCoeffect{\text{\lstinline{a2}}}{\Att}{\{\res\}}}{\text{\lstinline{a1.mix(a2)}}}{\text{\Att}}$ holds.
Then, the expression \lstinline{a1.mix(a2).clone()} returns a fresh result, hence $\IsWFExp{\VarTypeCoeffect{\text{\lstinline{a1}}}{\Att}{\{\link\}},\VarTypeCoeffect{\text{\lstinline{a2}}}{\Att}{\{\link\}}}{\text{\lstinline{a1.mix(a2).clone()}}}{\text{\Att}}$ holds, with $\link\neq\res$.
After the final call to \lstinline{mix}, since \lstinline{a1} and \lstinline{a2} have a link in common, the operation $\shsum$ adds to the coeffect of \lstinline{a1} the links of \lstinline{a2}, including $\res$, hence we get:
\begin{quote}
${\IsWFExp{\VarTypeCoeffect{\text{\lstinline{a1}}}{\Att}{\{ \link,\res\}}}{\text{\lstinline{\{A a2 = new A(new B(1));a1.mix(a2).clone().mix(a2)\}}}}{\Att}}$
\end{quote}
 expressing that \lstinline{a1} is linked to the result.
 
 We  now  state the properties of the coeffect system for sharing.
 
Given $\Gamma=\VarTypeCoeffect{\x_1}{\T_1}{\X_1},\ldots, \VarTypeCoeffect{\x_n}{\T_n}{\X_n}$, set $\getCoeff{\Gamma}{\x_i}=\X_i$  and  $\links{\Gamma}=\bigcup_{i\in 1..n}\X_i\cup\{\res\}$.
Finally, the \emph{restriction} of a context $\Gamma=\VarTypeCoeffect{\x_1}{\T_1}{\X_1},\ldots, \VarTypeCoeffect{\x_n}{\T_n}{\X_n}$ to the set of variables  $\Vars=\{\x_1,\ldots,\x_m\}$, with $m\leq n$, and the set of  links $\X$, denoted $\Restr{\Gamma}{\Pair{\Vars}{\X}}$, 
is the context ${\VarTypeCoeffect{\x_1}{\T_1}{\Y_1},\ldots, \VarTypeCoeffect{\x_m}{\T_m}{\Y_m}}$ where, for each $i\in 1..m$, $\Y_i=\X_i\cap\X$.
In the following, $\Restr{\Gamma}{\Delta}$ abbreviates $\Restr{\Gamma}{\Pair{\dom{\Delta}}{\links{\Delta}}}$.  

Recall that $\SharingRel{}{\mem}{}$ denotes the sharing relation in memory $\mem$ (\refToDefinition{sharing-rel}). 
The following result shows that the typing of the memory precisely captures the sharing relation. 

\begin{lemma}\label{lem:mem-ctx-sharing}
If $\IsWFMem{\Gamma}{\mem}$, then $\SharingRel{\x}{\mem}{\y}$ if and only if $\getCoeff{\Gamma}{\x}=\getCoeff{\Gamma}{\y}$.
\end{lemma}

Subject reduction states that not only type but also sharing  is preserved. 
More precisely, a reduction step may introduce new variables and new links, 
but the sharing between previous variables must be preserved,  as expressed by the following theorem.

\begin{theorem}[Subject reduction]\label{thm:subj-red}
If $\IsWFConf{\Gamma}{\e}{\mem}{\T}$ and $\reduce{\Pair{\e}{\mem}}{\Pair{\e'}{\mem'}}$, then 
$\IsWFConf{\Delta}{\e'}{\mem'}{\T}$, for some $\Delta$ such that 
$\Restr{(\Gamma\shsum\Delta)}{\Gamma}= \Gamma$. 
\end{theorem}

\begin{corollary}
If $\IsWFConf{\Gamma}{\e}{\mem}{\T}$ and $\reducestar{\Pair{\e}{\mem}}{\Pair{\e'}{\mem'}}$, then 
$\IsWFConf{\Delta}{\e'}{\mem'}{\T}$ for some $\Delta$ such that 
$\Restr{(\Gamma\shsum\Delta)}{\Gamma}= \Gamma$. 
\end{corollary} 

Indeed, coeffects in $\Gamma\shsum\Delta$ model the combined sharing before and after the computation step, hence 
the requirement $\Restr{(\Gamma\shsum\Delta)}{\Gamma} = \Gamma$ ensures that, on variables in $\Gamma$, the sharing remains the same. 
That is, the context $\Delta$ cannot connect variables that  were  disconnected  in $\Gamma$. 

Thanks to the fact that reduction preserves (initial) sharing, we can \emph{statically detect} lent references (\refToDefinition{lent}) and capsule expressions (\refToDefinition{caps}) just looking at coeffects, as stated below.
%

\begin{theorem}[Lent reference]\label{thm:lent} 
If $\IsWFConf{\Gamma}{\e}{\mem}{\C}$, $\x\in\dom{\Gamma}$ with $\res\not\in\getCoeff{\Gamma}{\x}$, and $\reducestar{\ExpMem{\e}{\mem}}{\ExpMem{\y}{\mem'}}$, then $\SharingRel{\x}{\mem'}{\y}$ does not hold.
\end{theorem}
\begin{proof}
By \cref{thm:subj-red} we have $\IsWFConf{\Delta}{\y}{\mem'}{\C}$, for some $\Delta$ such that 
 $\Restr{(\Gamma\shsum\Delta)}{\Gamma}= \Gamma$.  
By inversion, we have 
$\IsWFExp{\VarTypeCoeffect{\y}{\C}{\{\res\}}}{\y}{\C}$ with $\Delta = \Delta' \shsum \VarTypeCoeffect{\y}{\C}{\{\res\}}$, hence 
$\res\in\getCoeff{\Delta}{\y}$. 
Assume $\SharingRel{\x}{\mem'}{\y}$.
By \cref{lem:mem-ctx-sharing}, we have $\getCoeff{\Delta}{\x} = \getCoeff{\Delta}{\y}$, thus 
$\res\in\getCoeff{\Delta}{\x}$.
Since  $\Restr{(\Gamma\shsum\Delta)}{\Gamma}= \Gamma$, $\res\in\getCoeff{\Gamma\shsum\Delta}{\x}$ and $\x\in\dom{\Gamma}$, 
we also have $\res\in\getCoeff{\Gamma}{\x}$, contradicting the hypothesis. 
\end{proof}

We write $\Capsule{\Gamma}$ if, for each $\x\in\dom{\Gamma}$, $\res\not\in\getCoeff{\Gamma}{\x}$, that is, $\x$ is lent.
The theorem above immediately implies that an expression  which is  typable in such a context is a capsule.

\begin{corollary}[Capsule expression]\label{cor:capsule}
If $\IsWFConf{\Gamma}{\e}{\mem}{\C}$, with $\Capsule{\Gamma}$, and $\reducestar{\ExpMem{\e}{\mem}}{\ExpMem{\y}{\mem'}}$, then, 
for all $\x\in\dom{\Gamma}$, $\SharingRel{\x}{\mem'}{\y}$ does not hold.
\end{corollary}
\begin{proof}
Let $\x\in\dom{\Gamma}$.
The hypothesis $\Capsule{\Gamma}$ means that each variable in $\Gamma$ is lent that is, $\res\notin\getCoeff{\Gamma}{\x}$. Then, 
by \cref{thm:lent}, $\SharingRel{\x}{\mem'}{\y}$ does not hold.\end{proof}
\noindent Note that, in particular, \cref{cor:capsule} ensures that  no  free variable of $\e$ can access the reachable object graph of the final result $\y$.  
Notice also that assuming $\Capsule{\Gamma}$ is the same as assuming $\Capsule{\Delta}$ where $\Gamma = \Delta \shsum \Delta'$ and $\Delta$ is the context that types the expression $\e$, because no $\res$ link can occur in the context that types the memory.


\section{Case study: type modifiers for uniqueness and immutability}\label{sect:extended}

 The coeffect system in the previous section tracks sharing among variables possibly introduced by reduction. 
  In this section, we check the effectiveness of the approach to model specific language features related to sharing and mutation, taking as challenging case  study  those proposed by \citet{GianniniSZC19,GianniniRSZ19}, whose common key ideas are the following:
\begin{itemize}
\item types are decorated by \emph{modifiers} $\mut$ (default, omitted in code), $\readonly$, $\capsule$, and $\imm$ for read-only, capsule, and immutable, respectively, allowing the programmer to specify the corresponding contraints/properties for variables/parameters and method return types
\item  $\mut$ (resp. $\readonly$) expressions can be transparently \emph{promoted} to $\capsule$ (resp. $\imm$)
\item $\capsule$ expressions can be assigned to either mutable or immutable references.
\end{itemize}

 For instance, consider the following version of \refToExample{ex2} decorated with modifiers: 
\begin{example}\label{ex:ex2-decorated}\
\begin{lstlisting}
class B {int f; B clone $\meta{[\readonly^{\{\link\}}]}$() {new B(this.f)} // $\link\neq\res$
class A { B f;
  A mix $\meta{[^{\{\res\}}]}$(A$\meta{^{\{\res\}}}$a) {this.f=a.f; a} // this, a and the result linked
  A clone $\meta{[\readonly^{\{\link\}}]}$ () {new A(this.f.clone()) } // $\link\neq\res$
}
A a1=new A(new B(0));
read A mycaps = {A a2 = new A(new B(1));
  a1.mix(a2).clone()// (1) 
  // a1.mix(a2).clone().mix(a2) // (2)
}
// mycaps.f.f= 3 // (3)
a1.f.f=3 // (4)
\end{lstlisting}
 The modifier of $\this$  in \lstinline{mix} needs to be $\mut$, whereas in \lstinline{clone} it is $\readonly$ to allow invocations on arguments with any modifier. The result modifier in \lstinline{mix} is that of the parameter \lstinline{a}, chosen to be \lstinline{mut} since $\readonly$ would have made the result of the call less usable. The result modifier of \lstinline{clone} \emph{could} be $\capsule$, but even if it is $\mut$, the fact that there is no connection between the result and $\this$ is expressed by the coeffect. The difference is that with modifier $\capsule$ promotion takes place when typechecking the body of the method, whereas with modifier $\mut$ it takes place at the call site. 

As expected, an expression with type tagged $\readonly$  cannot occur as  the  left-hand side of  a field  assignment. To have the guarantee that a portion of memory is immutable, a type system should be able to detect that it cannot be modified through \emph{any} possibile reference. 
In the example, since \lstinline{mycaps} is declared $\readonly$, line (3) is ill-typed. 
However, if we replace line (1) with line (2), since in this case \lstinline{mycaps} and \lstinline{a1} share their \lstinline{f} field, the same effect of line (3) can be obtained by line (4). 
As previously illustrated, the sharing coeffect system 
detects that only in the version with line (1) does \lstinline{mycaps}  denote a capsule. 
Correspondingly, in the enhanced type system in this section, \lstinline{mycaps} can be correctly declared $\capsule$, hence $\imm$ as well, whereas this is not the case with line (2). 
By declaring \lstinline{mycaps} of an $\imm$ type, the programmer has the guarantee that the portion of memory denoted by \lstinline{mycaps} cannot be modified through another reference.  
That is, the immutability property is detected as a conjunction of the read-only restriction and the capsule property. 

Assume now that \lstinline{mycaps} is declared $\capsule$ rather than $\readonly$. Then, line (3) is well-typed. However, if \lstinline{mycaps} could be assigned to  both   a mutable and an immutable reference, e.g:
\begin{lstlisting}
A$^\imm$ imm = mycaps;
mycaps.f.f=3
\end{lstlisting}
the immutability guarantee for \lstinline{imm} would be broken.  For this reason, capsules can only be used linearly in the following type system.
\end{example}

We formalize the features illustrated above by a type-and-coeffect system built on top of that of the previous section, whose key  advantage  is that  detection of $\capsule$ and $\imm$ types is \emph{straightforward} from the coeffects, through a simple \emph{promotion}\footnote{ This terminology is chosen to  emphasize  the analogy with promotion in linear logic.} rule, since they exactly express the desired properties.  

Type-and-coeffect contexts are, as before, of shape $\VarTypeCoeffect{\x_1}{\T_1}{\X_1}, \ldots, \VarTypeCoeffect{\x_n}{\T_n}{\X_n}$, where types are either primitive types or of shape $\TypeMod{\C}{\modif}$, with $\modif$ modifier. 
We assume that fields can be declared either $\imm$ or $\mut$, whereas the modifiers $\capsule$ and $\readonly$ are only used for local variables. 
Besides those, which are written by the programmer in source code, modifiers include a numerable set of \emph{seals} $\seal$ which are only internally used by the type system, as will be explained  later. 

\begin{figure}

\framebox{
\begin{minipage}{0.4\textwidth}
\begin{tikzpicture}
\node (caps) at (10ex,7ex) {$\capsule$};
\node (mut) at (0ex,12ex) {$\mut$};
\node (imm) at (20ex,12ex) {$\imm$};
\node (read) at (10ex,17ex) {$\readonly$};
\node (seal) at (5ex,0ex) {$\seal$};
\node (sealP) at (15ex,0ex) {$\seal'$};
\draw [arrows={-latex}] (caps) -- (imm);
\draw [arrows={-latex}] (caps) -- (mut);
\draw [arrows={-latex}] (mut) -- (read);
\draw [arrows={-latex}] (imm) -- (read);
\draw [arrows={-latex}] (seal) -- (caps) ;
\draw [arrows={-latex}] (sealP) -- (caps);
\draw[->,double] (mut) to[bend right=20] (caps);
\draw[->,double](read) to[bend right=40] (imm); 
\node (a1) at (6ex,-0.3ex) {};
\node (a2) at (14ex,-0.3ex) {};
\node (b1) at (6ex,0.5ex) {};
\node (b2) at (14ex,0.5ex) {};
\draw [arrows={-latex}] (a1) -- (a2);
\draw [arrows={-latex}] (b2) -- (b1) ;
\end{tikzpicture}
\end{minipage}
\begin{minipage}{0.3\textwidth}
$
\begin{array}{|l}
\mbox{Arrows:}
\\
\begin{array}{l l}
\begin{tikzpicture}
\draw [arrows={-latex}] (0,0) -- (1,0);
\end{tikzpicture}
&\mbox{Subtype}
\\
\begin{tikzpicture}
\draw[->,double] (0,0) to (1,0);
\end{tikzpicture}
&\mbox{{Promotion}}
\end{array}
\end{array}
$
\end{minipage}
}
\caption{Type modifiers and their relationships}
\label{fig:hierarchy}

\end{figure}
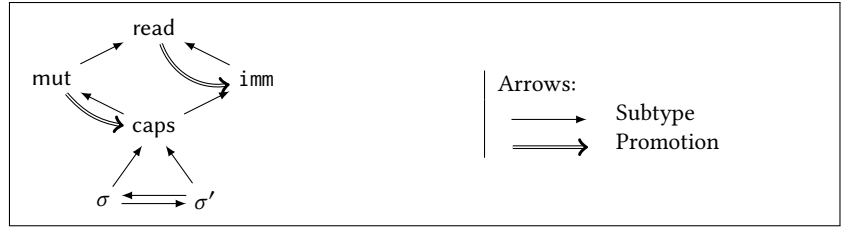

 Operations on coeffect contexts are lifted to type-and-coeffect contexts as in the previous case.  However, there are some novelties:
 \begin{itemize}
\item  The preorder must take into account subtyping as well, defined by
 \begin{quote}
$\T\leq\T'$ if either $\T=\T'$ primitive type, or $\T=\TypeMod{\C}{\modif}$, $\T'=\TypeMod{\C}{\modif'}$, and $\modif\leq\modif'$  induced by 
\mbox{$\seal\leq\seal'$, $\seal\leq\capsule$, $\capsule\leq\mut$, $\capsule\leq\imm$, $\mut\leq\readonly$, $\imm\leq\readonly$, see \refToFigure{hierarchy}.}
\end{quote}
\item In the sum of two contexts, denoted $\Gamma\ctxlinsum\Delta$, variables of a $\capsule$ or $\seal$ type cannot occur in both; that is, they are handled \emph{linearly}.
\end{itemize}

Combination of modifiers, denoted $\ModifComb{\modif}{\modif'}$, is the following operation:
\begin{quote}
$\ModifComb{\modif}{\modif'}=\modif$ if $\modif\leq \imm$\BigSpace
$\ModifComb{\mut}{\modif}=\modif$\BigSpace
$
\ModifComb{\readonly}{\modif}=
\begin{cases}
\imm \BigSpace \text{if } \modif=\imm\ \mbox{or}\ \modif=\capsule\\
\text{undefined}\BigSpace\text{if}\ \modif=\seal\\
\readonly \BigSpace \text{if } \mut\leq\modif
\end{cases}
$
\end{quote}

 Combination of modifiers is used in \refToRule{t-field-access}  to propagate the modifier of the receiver, and in \refToRule{t-prom} to promote the type and \emph{seal} mutable variables connected to the result, see below.

The typing rules are given in \refToFigure{typing-modifiers}. We only comment  on  the novelties with respect to \refToSection{sharing}. 
\begin{figure}
\begin{grammatica}
\produzione{\T}{\TypeMod{\C}{\modif}\mid\PT\mid\ldots}{type}\\
\produzione{\modif}{\mut\mid\readonly\mid\imm\mid\capsule\mid\seal}{modifier}\\
\end{grammatica}

\hrule

\begin{small}
\begin{math}
\begin{array}{l}
\\
\NamedRule{t-sub}{\IsWFExp{\Gamma}{\e}{\T'}}{\IsWFExp{\Gamma}{\e}{\T}}{
\T'\leq\T}
\BigSpace

\NamedRule{t-var}{}{ \IsWFExp{{ \emptyset\ctxmul \Gamma \ctxlinsum \VarTypeCoeffect{\x}{\T}{\{\res\}}}}{\x}{\T}}{
}
\BigSpace
\NamedRule{t-const}{}{\IsWFExp{\emptyset\shmul \Gamma}{\const}{\PT_\const}}{}
\BigSpace\\[4ex] 
\NamedRule{t-field-access}{\IsWFExp{\Gamma}{\e}{\TypeMod{\C}{\modif}}}{\IsWFExp{\Gamma}{\FieldAccess{\e}{\f_i}}{\Modif{\T_i}}{\modif}}
{
\fields{\C}=\Field{\T_1}{\f_1} \ldots \Field{\T_n}{\f_n} \\
i\in 1..n\\
}
\\[4ex]
\NamedRule{t-field-assign}{\IsWFExp{\Gamma}{\e}{\TypeMod{\C}{\mut}} \Space \IsWFExp{\Delta}{\e'}{\T_i}}{\IsWFExp{\Gamma\ctxlinsum \Delta}{\FieldAssign{\e}{\f_i}{\e'}
}{\T_i}}
{
\fields{\C}=\Field{\T_1}{\f_1} \ldots \Field{\T_n}{\f_n}\\
  i\in 1..n\\ 
}
\\[4ex]
\NamedRule{t-new}{\IsWFExp{\Gamma_i}{\e_i}{\T_i}\Space \forall i\in 1..n}{\IsWFExp{\Gamma_1 \ctxlinsum\ldots \ctxlinsum\Gamma_n}{\ConstrCallTuple{\C}{\e}{n}}{\TypeMod{\C}{\mut}}}
{\fields{\C}=\Field{\T_1}{\f_1} \ldots \Field{\T_n}{\f_n}}
\\[4ex]

\NamedRule{t-invk}{\IsWFExp{\Gamma_0}{\e_0}{\TypeMod{\C}{\modif}} \BigSpace \IsWFExp{\Gamma_i}{\e_i}{\T_i}\Space \forall i\in 1..n}{\IsWFExp{(\X_0\cup\{\link_0\} \ctxmul\Gamma_0) \ctxlinsum \ldots \ctxlinsum (\X_n\cup\{\link_n\} \ctxmul\Gamma_n)}{\MethCallTuple{\e_0}{\m}{\e}{n}}{\T}}{
\mtype{\C}{\m}\eqfresh\funType{\coeffectType{\modif}{\X_0},\coeffectType{\T_1}{\X_1} \ldots \coeffectType{\T_n}{\X_n}}{\T}\\
\link_0,\ldots,\link_n\ \mbox{fresh}
}

\\[4ex]
\NamedRule{t-block}{
\IsWFExp{\Gamma}{\e}{\T} \BigSpace
\IsWFExp{\Gamma', \VarTypeCoeffect{\x}{\T}{\X}}{\e'}{\T'}
}
{
\IsWFExp{(\X \cup\{ \link \}) \ctxmul\Gamma \ctxlinsum \Gamma'}{\Block{\T}{\x}{\e}{\e'}}{\T'}
}
{\link \ \text{fresh}}
\\[3ex]

\NamedRule{t-imm}{\IsWFExp{\Gamma}{\e}{\T}}{\IsWFExp{\{\link\}\ctxmul\Gamma}{\e}{\T}}
{
\link \ \text{fresh}\\
\T=\PT\ \mbox{or}\ \T=\TypeMod{\C}{\imm}
}

\BigSpace

\NamedRule{t-prom}{\IsWFExp{\Gamma}{\e}{\TypeMod{\C}{\modif}}}{\IsWFExp{\Sealed{\Gamma}}{\e}{\TypeMod{\C}{\ModifComb{\modif}{\capsule}}}}
{\mut\leq\modif\\
\seal\ \mbox{fresh}}

\\[3ex]
\end{array}
\end{math}

\hrule

\begin{math}
\begin{array}{l}
\\
\NamedRule{t-conf}{\IsWFExp{\Delta}{\e}{\T} \BigSpace \IsWFMem{\Gamma}{\mem}}
{\IsWFConf{\Delta\ctxsum\Gamma}{\e}{\mem}{\T}} 
{  \dom{\Delta}\subseteq\dom{\Gamma}
}
\\[4ex]
\NamedRule{ t-ref}{}{\IsWFInMem{\VarTypeCoeffect{\x}{\TypeMod{\C}{\modif}}{\{\res\}}}{\x}{\TypeMod{\C}{\modif}}}{\modif=\mut\ \mbox{or}\ \modif=\seal}
\BigSpace
\NamedRule{ t-imm-ref}{}{\IsWFInMem{\VarTypeCoeffect{\x}{\TypeMod{\C}{\imm}}{\{\link\}}}{\x}{\TypeMod{\C}{\imm}}}{\link\text{ fresh}} 
\\[4ex]
\NamedRule{t-mem-const}{}{\IsWFInMem{\EmptyCtx}{\const}{\PT_\const}}{}
\BigSpace
\NamedRule{ t-obj}{\IsWFInMem{\Gamma_i}{\val_i}{\Modif{\T_i}{\modif}}\Space \forall i\in 1..n}{\IsWFInMem{\Gamma_1\shsum\cdots\shsum\Gamma_n}{\Object{\C}{\val_1,\ldots,\val_n}}{\TypeMod{\C}{\modif}}}{\fields{\C}=\Field{\T_1}{\f_1} \ldots \Field{\T_n}{\f_n}}
\\[4ex]
\NamedRule{t-mem}{\IsWFInMem{\Gamma_i}{\mem(\x_i)}{\TypeMod{\C_i}{\modif_i}\Space\forall i\in 1...n}
}
{\IsWFMem{\Gamma_{\!\mem}\shsum\Gamma}{\mem}} 
{\Gamma_{\!\mem}=\VarTypeCoeffect{\x_1}{\TypeMod{\C_1}{\modif_1}}{\{\link_1\}},\ldots,\VarTypeCoeffect{\x_n}{\TypeMod{\C_n}{\modif_n}}{\{\link_n\}}\\
\dom{\Gamma_{\!\mem}}=\dom{\mem}\\
\Gamma=(\{\link_1\} \ctxmul\Gamma_1) \shsum \ldots \shsum (\{\link_n\}\ctxmul\Gamma_n)\\
\link_1,\ldots,\link_n\ \mbox{fresh}}
\end{array}
\end{math}
\end{small}
\caption{Adding modifiers and immutability}\label{fig:typing-modifiers}
\end{figure}

 Rule \refToRule{t-sub} uses the subtyping relation 
 defined above. For instance, an expression of type $\TypeMod{\C}{\capsule}$ has the types $\TypeMod{\C}{\mut}$ and $\TypeMod{\C}{\imm}$ as well.
 In rule \refToRule{t-field-access}, the notation $\Modif{\T}{\modif}$ denotes $\TypeMod{\C}{\ModifComb{\modif'}{\modif}}$ if $\T=\TypeMod{\C}{\modif'}$, and $\T$ otherwise, that is, if $\T$ is a primitive type. For instance, mutable fields referred  to  through an $\imm$ reference are $\imm$ as well. In other words, modifiers are \emph{deep}. 

In rule \refToRule{t-field-assign}, only a $\mut$ expression can occur as  the  left-hand side of a field assignment. 
In rule \refToRule{t-new}, a constructor invocation is $\mut$, hence $\mut$ is the default modifier of expressions of reference types. 
Note that the $\readonly$ modifier can only be introduced by variable/method declaration. The $\capsule$ and $\imm$ modifiers, on the other hand, in addition to variable/method declaration, can be introduced by the \emph{promotion} rule \refToRule{t-prom}.

 As in the previous type system, the auxiliary function $\aux{mtype}$ returns  an enriched method type where the parameter types are decorated with coeffects, including the implicit parameter $\this$.  The condition that method bodies should be well-typed with respect to method types is exactly as in the previous type system, with  only the difference that types have modifiers. 

Rule \refToRule{t-imm} generalizes rule \refToRule{t-prim} of the previous type system,  allowing the links with the result to be removed,   to immutable types. For instance, assuming  the following variant of \refToExample{ex1}  (recall that the default modifier $\mut$ can be omitted):
\begin{lstlisting}
class B {int f;}
class C {imm B f1; B f2;}
\end{lstlisting}
the following derivable judgment 
\begin{quote}
$\IsWFExp{\VarTypeCoeffect{\zttOne}{\TypeMod{\Btt}{\imm}}{\emptyset},\VarTypeCoeffect{\zttTwo}{\Btt}{\{ \link\}}}{\FieldAccess{\ConstrCall{\Ctt}{\zttOne,\zttTwo}}{\fttOne}}{\TypeMod{\Btt}{\imm}}$, with $\link\neq\res$
\end{quote}
shows that there is no longer a link between the result and $\zttOne$.

The new rule \refToRule{ t-prom} plays a key role, since, as already mentioned, it detects that an expression is a capsule thanks to  its  coeffects, and \emph{promotes} its type accordingly.  
 The basic idea is that a $\mut$ (resp. $\readonly$) expression can be promoted to $\capsule$ (resp. $\imm$) provided that there are no free variables connected to the result  with modifier  $\readonly$ or $\mut$. 
However, to guarantee that type preservation holds, the same promotion should be possible for runtime expressions, which may contain free variables which actually are $\mut$ references generated during reduction. To this end, the rule allows $\mut$ variables connected to the result\footnote{ Whereas $\readonly$ variables are still not allowed, as expressed by the fact that $\ModifComb{\readonly}{\seal}$ is undefined. }. Such variables become \emph{sealed} as an effect of the promotion, leading to the context
$\Sealed{\Gamma}$, obtained from $\Gamma$ by combining modifiers of variables connected to the result with $\seal$. 
Formally, if $\Gamma = \VarTypeCoeffect{\x_1}{\T_1}{\X_1}, \ldots,  \VarTypeCoeffect{\x_n}{\T_n}{\X_n}$,    
\[
\Sealed{\Gamma}=  \VarTypeCoeffect{\x_1}{\T'_1}{\X_1}, \ldots,  \VarTypeCoeffect{\x_n}{\T'_n}{\X_n}\quad \mbox{where $\T'_i=\Modif{\T_i}{\seal}$ if $\res\in\X_i$, $\T'_i=\T_i$ otherwise}
\]
The notation $\Modif{\T}{\seal}$ is the same used in rule \refToRule{t-field-access}.     

This highlights once again the analogy with the promotion rule for the (graded) bang modality of linear logic \cite{BreuvartP15}, where, in order to introduce a modality on the  right-hand side  of a sequent, one has to  modify the  left-hand side  accordingly.\footnote{ This is just an analogy, making it precise is an interesting direction for future work.}
We detail in the following how sealed variables are internally used by the type system to guarantee type preservation. 

Rule \refToRule{t-conf} is as in \cref{fig:typing-sharing}. Note that we use  the  sum of contexts $\shsum$ from the previous type system, since the linear treatment of $\capsule$ and $\seal$ variables is only required in source code.

Rule \refToRule{t-mem} is also analogous to that in \cref{fig:typing-sharing}. However, typechecking objects is modeled by an ad-hoc judgment $\Vdash$, where references can only be $\mut$, $\imm$, or $\seal$ ($\readonly$ and $\capsule$ are source-only notions), and subsumption is not included. As a consequence, rule \refToRule{t-obj} imposes that a reference reachable from an $\imm$ reference or field should be tagged $\imm$ as well, and analogously for seals. 

 As in the previous type system, the rules in \refToFigure{typing-modifiers} lead to an algorithm which inductively computes the coeffects of an expression. The only relevant novelty is  rule  \refToRule{ t-prom}, assumed to be applied \emph{only when needed}, that is, when we typecheck the initialization expression of a local variable declared $\capsule$, or the argument of a method call where the corresponding parameter is declared $\capsule$. Rule \refToRule{t-imm} is applied, as \refToRule{t-prim} before,  only once, whenever an expression has either a primitive or an immutable type. Subsumption rule \refToRule{t-sub} only handles types, and can be replaced by a more verbose version  of the rules with subtyping conditions where needed.  In the other cases, rules are  syntax-directed,  that is, the coeffects of the expression in the consequence are computed as a linear combination of those of the subexpressions,   where the basis is  the rule for variables.

We illustrate now the use of seals to preserve types during reduction.
For instance, consider again \refToExample{ex1}: 
\begin{lstlisting}
class B {int f;}
class C {B f1; B f2;}
\end{lstlisting}
\begin{quote}
$\e_0=\Block{\TypeMod{\Btt}{}}{\ztt}{\ConstrCall{\Btt}{2}}{\Seq{\FieldAssign{\xtt}{\fttOne}{\ytt}}{\ConstrCall{\Ctt}{\ztt,\ztt}}}$\\
 $\mem_0=\{\xtt\mapsto\Object{\Ctt}{\xttOne,\xttOne},\xttOne\mapsto\Object{\Btt}{0},\ytt\mapsto\Object{\Btt}{1}\}$
\end{quote}
Expression $\e_0$ is a \emph{capsule} since its free variables (external resources) $\xtt$ and $\ytt$ will not be connected to the final result. 
Formally, set $\Delta=\VarTypeCoeffect{\xtt}{\TypeMod{\Ctt}{}}{\{\link\}},\VarTypeCoeffect{\ytt}{\Btt}{\{\link\}}$, with $\link\neq\res$, we can derive the judgment $\IsWFExp{\Delta}{\e_0}{\TypeMod{\Ctt}{\mut}}$, and then apply the promotion rule \refToRule{ t-prom}, as shown below.
\begin{quote}
$\NamedRule{t-conf}{
{\NamedRule{ t-prom}{
\IsWFExp{\Delta}{\e_0}{\TypeMod{\Ctt}{\mut}}}{\IsWFExp{\Delta}{\e_0}{\TypeMod{\Ctt}{\capsule}}}{}
}
\BigSpace \IsWFMem{\Gamma}{\mem_0}
}{\IsWFConf{\Delta\shsum\Gamma}{\e_0}{\mem_0}{\TypeMod{\Ctt}{\capsule}}}{}$\BigSpace$\Gamma=\VarTypeCoeffect{\xtt}{\TypeMod{\Ctt}{}}{\{\link_\xtt\}},\VarTypeCoeffect{\xttOne}{\TypeMod{\Ctt}{}}{\{\link_\xtt\}},\VarTypeCoeffect{\ytt}{\Btt}{\{\link_\ytt\}}$
\end{quote}
where promotion does not affect the context $\Delta$ as there are no mutable variables connected to $\res$.

The first steps of the reduction of $\ExpMem{\e_0}{\mem_0}$ are as follows:
\begin{quote}
$\begin{array}{l@{}l} 
\ExpMem{\e_0}{\mem_0} 
  &\ev \ExpMem{\e_1}{\mem_1}=\ExpMem{\Block{\TypeMod{\Btt}{}}{\ztt}{\wtt}{\Seq{\FieldAssign{\xtt}{\fttOne}{\ytt}}{\ConstrCall{\Ctt}{\ztt,\ztt}}}}{\mem\cup\{\wtt\mapsto\Object{\Btt}{2}\}}\\
  &\ev \ExpMem{\e_2}{\mem_1}=\ExpMem{\Seq{\FieldAssign{\xtt}{\fttOne}{\ytt}}{\ConstrCall{\Ctt}{\wtt,\wtt}}}{\mem\cup\{\wtt\mapsto\Object{\Btt}{2}\}} 
\end{array}$
\end{quote} 
Whereas sharing preservation, in the sense of \refToTheorem{subj-red}, clearly still holds, to preserve the $\capsule$ type of the initial expression the \refToRule{ t-prom} promotion rule should be applicable to $\e_1$ and $\e_2$ as well.  However, in  the  next steps $\wtt$ is a free variable connected to the result; for instance for $\e_1$, we derive:
\begin{quote}
$\IsWFExp{\Delta,\VarTypeCoeffect{\wtt}{\Btt}{\{\res\}}}{\e_1}{\TypeMod{\Ctt}{\mut}}$
\end{quote}
Intuitively, $\e_1$ is still a capsule, since $\wtt$ is a fresh reference denoting a closed object in memory.
Formally, the promotion rule can still be applied, but variable $\wtt$ becomes \emph{sealed}:
\begin{quote}
$\NamedRule{t-conf}{
{\NamedRule{ t-prom}{
\IsWFExp{\Delta,\VarTypeCoeffect{\wtt}{\Btt}{\{\res\}}}{\e_1}{\TypeMod{\Ctt}{\mut}}}{\IsWFExp{\Delta,\VarTypeCoeffect{\wtt}{\TypeMod{\Btt}{\seal}}{\{\res\}}}{\e_1}{\TypeMod{\Ctt}{\capsule}}}{}
}
\BigSpace \IsWFMem{\Gamma,\VarTypeCoeffect{\wtt}{\TypeMod{\Btt}{\seal}}{\{\link_\wtt\}}}{\mem_1}
}{\IsWFConf{\Delta,\VarTypeCoeffect{\wtt}{\TypeMod{\Btt}{\seal}\shsum\Gamma}{\{\res\}}}{\e_1}{\mem_1}{\TypeMod{\Ctt}{\capsule}}}{}$
\end{quote}

Capsule guarantee is preserved since a sealed reference is handled linearly, and the typing rules for memory (judgment $\Vdash$) 
ensure that it can only be in sharing with another one with the same seal. Moreover, the relation $\seal\leq\seal'$ ensures type preservation in case a group of sealed references collapses during reduction in another one, as happens with a nested promotion.  
\bigskip

Let us denote by $\Erase{\Gamma}$ the context obtained from $\Gamma$ by erasing modifiers (hence, a context of the previous type-and-coeffect system). Subject reduction includes sharing preservation, as in the previous type system; in this case modifiers are preserved as well.  More precisely, they can decrease in the type of the expression, and increase in the type of references in the context.  We write $\Gamma \modord \Delta$ when, for all $\x\in\dom\Gamma$, we have $\getModif{\Gamma}{\x} \le \getModif{\Delta}{\x}$. 

\begin{theorem}[Subject Reduction]\label{thm:subj-red-extended}
 If $\IsWFConf{\Gamma}{\e}{\mem}{\T}$ and $\reduce{\Pair{\e}{\mem}}{\Pair{\e'}{\mem'}}$ then $\IsWFConf{\Delta}{\e'}{\mem'}{\T}$ for some $\Delta$ such that 
\begin{itemize}
\item
$\Restr{(\Gamma'\shsum\Delta')}{\Gamma'} = \Gamma'$, for $\Gamma'=\Erase{\Gamma}$ and $\Delta'=\Erase{\Delta}$; 
\item
 $\Gamma \modord \Delta$.  
\end{itemize}
\end{theorem}
We now focus on properties of the memory ensured by this extended type system. 
First of all, we prove two lemmas characterising how the typing of memory propagates type modifiers. 
Recall that $\reachableSymbol{\mem}$ denotes the reachability relation in memory $\mem$ (\cref{def:reach}). 

\begin{lemma}\label{lem:mod-deep} 
If $\IsWFMem{\Gamma}{\mem}$ and $\reachable{\x}{\mem}{\y}$, then 
\begin{itemize}
\item $\getModif{\Gamma}{\x} = \mut$  implies $\getModif{\Gamma}{\y} = \mut$ or $\getModif{\Gamma}{\y} = \imm$, 
\item $\getModif{\Gamma}{\x} = \seal$ implies $\getModif{\Gamma}{\y} = \seal$ or $\getModif{\Gamma}{\y} = \imm$, 
\item $\getModif{\Gamma}{\x} = \imm$  implies $\getModif{\Gamma}{\y} = \imm$. 
\end{itemize}
\end{lemma}
\begin{proof}
By induction on the definition of $\reachableSymbol{\mem}$.
\begin{description} 
\item [Case $\y=\x$] The thesis trivially holds. 
\item [Case $\mem(\x)=\Object{\C}{\val_1,\ldots,\val_n}$, $\z = \val_i$ for some  $i\in 1..n$ and $\reachable{\z}{\mem}{\y}$].
Since $\IsWFMem{\Gamma}{\mem}$,  inverting rule \refToRule{t-mem},   we have 
$\IsWFInMem{\Delta}{\Object{\C}{\val_1,\ldots,\val_n}}{\TypeMod{\C}{\modif}}$ for some $\Delta$ with $\Gamma = \Gamma' \shsum \{\link\}\ctxmul\Delta$ and $\modif = \getModif{\Gamma}{\x}$. 
 Inverting rule \refToRule{t-obj} and either \refToRule{t-ref} or \refToRule{t-imm-ref},  we have that $\IsWFInMem{\VarTypeCoeffect{\z}{\Modif{\T_i}{\modif}}{\X}}{\z}{\Modif{\T_i}{\modif}}$, with 
$\Delta = \Delta', \VarTypeCoeffect{\z}{\Modif{\T_i}{\modif}}{\X}$ and $\fields{\C}=\Field{\T_1}{\f_1} \ldots \Field{\T_n}{\f_n}$. 
Since $\z\in\dom{\mem}$, $\T_i$ is of shape $\TypeMod{\C'}{\modif'}$. 
We split cases on $\modif'$. 
\begin{itemize}
\item If $\modif' = \imm$, then $\Modif{\T_i}{\modif} = \imm$, hence $\getModif{\Gamma}{\z} = \imm$ and so, by induction hypothesis, we get $\getModif{\Gamma}{\y} = \imm$ as needed. 
\item If $\modif' = \mut$, then $\Modif{\T_i}{\modif} = \modif$, hence $\getModif{\Gamma}{\z} = \modif$ and so the thesis follows by induction hypothesis. 
\end{itemize}
\end{description} 
\end{proof}

\begin{lemma} \label{lem:sh-mem}
If $\IsWFMem{\Gamma}{\mem}$, then, 
for all $\x,\y{\in}\dom{\mem}$, 
$\getCoeff{\Gamma}{\x} {=} \getCoeff{\Gamma}{\y}$  implies ${\getModif{\Gamma}{\x} = \getModif{\Gamma}{\y}}$. 
\end{lemma} 

In this refined setting, 
the definition of the sharing relation needs to take into account modifiers. 
Indeed, if intuitively two references are in sharing when a mutation of either of the two affects the other, 
then  no sharing should be propagated through immutable references.
To do so, we need to assume a well-typed memory in order to know modifiers of references.\footnote{Actually, we do not need the full typing information, having just modifiers would be enough. } 

\begin{definition}[Sharing in memory with modifiers]\label{def:sharing-rel-with-imm}
The \emph{sharing relation} in memory $\IsWFMem{\Gamma}{\mem}$, denoted by $\sharingRelSymbol{\Gamma,\mem}$, is the smallest equivalence relation on $\dom{\mem}$ such that:
\begin{quote}
$\SharingRel{\x}{\Gamma,\mem}{\y}$ if 
\mbox{$\mem(\x)=\Object{\C}{\val_1,\ldots,\val_n}$, 
$\getModif{\Gamma}{\x},\getModif{\Gamma}{\y} \le \mut$ and 
$\y = \val_i$ \mbox{for some  $i\in 1..n$}}
\end{quote}
\end{definition}
\noindent Again, 
for a well-typed memory, 
coeffects characterize the sharing relation exactly. 

\begin{proposition}\label{prop:sh-mem}
If $\IsWFMem{\Gamma}{\mem}$, then 
 $\SharingRel{\x}{\Gamma,\mem}{\y}$  iff $\getCoeff{\Gamma}{\x} = \getCoeff{\Gamma}{\y}$, for all $\x,\y\in\dom{\mem}$. 
\end{proposition}

 In the extended type system, we can detect capsule expressions from the  modifier, without looking at coeffects of free variables, proving that the result of a $\capsule$ expression is not in sharing with the initial  mutable  variables. 

\begin{theorem}[Capsule expression]\label{thm:caps} 
If $\IsWFConf{\Gamma}{\e}{\mem}{\TypeMod{\C}{\capsule}}$, and 
$\reducestar{\ExpMem{\e}{\mem}}{\ExpMem{\y}{\mem'}}$, then 
there exists $\Gamma'$ such that $\IsWFMem{\Gamma'}{\mem'}$, $\Gamma \modord \Gamma'$ and, for all  $\x\in\dom{\mem}$, 
$\SharingRel{\x}{\Gamma',\mem'}{\y}$ implies $\getModif{\Gamma}{\x} \le \getModif{\Gamma'}{\x} \ne \mut$.
\end{theorem}  
\begin{proof}
By \cref{thm:subj-red-extended}, we get  
$\IsWFConf{\Delta}{\y}{\mem'}{\TypeMod{\C}{\capsule}}$ with $\Gamma \modord \Delta$. 
By inverting rule \refToRule{t-conf}, we get 
$\IsWFExp{\Delta_1}{\y}{\TypeMod{\C}{\capsule}}$ and $\IsWFMem{\Delta_2}{\mem'}$, with $\Delta = \Delta_1 \shsum \Delta_2$ and 
$\Gamma \modord \Delta_2$, as $\getModif{\Delta_2}{\z} = \getModif{\Delta}{\z}$ for all $\ \in \dom\Delta$. 
Since $\y \in\dom{\mem'}$, it cannot have modifier $\capsule$, hence 
$\IsWFExp{\Delta_1}{\y}{\TypeMod{\C}{\capsule}}$ holds by rule \refToRule{t-prom} or \refToRule{t-sub}. 
This implies $\Delta_1 = \emptyset \ctxmul\Delta' , \VarTypeCoeffect{\y}{\TypeMod{\C}{\seal}}{\{\res\}}$ and so 
$\getModif{\Delta}{\y} = \getModif{\Delta_2}{\y} = \getModif{\Delta_1}{\y} = \seal$. 
Set $\Gamma' = \Delta_2$. 
By \cref{prop:sh-mem,lem:sh-mem}, $\SharingRel{\x}{\Gamma',\mem'}{\y}$ implies $\getModif{\Gamma'}{x} = \seal$, thus we get 
$\getModif{\Gamma}{x} \le \getModif{\Delta}{\x} = \getModif{\Gamma'}{\x} \ne \mut$, hence the thesis. 
\end{proof}

It is important to notice that the notion of capsule expression in \cref{thm:caps} is different from the previous one  (\cref{def:caps}), as we now have $\imm$ references.
In particular, the previous notion prevented any access to the reachable object graph of the result from free variables, since, without modifiers, any access to a portion of memory can modify it. 
Here, instead, this is no longer true, hence the notion of capsule allows mutable references to access the reachable object graph of the result of a capsule expression, but only through $\imm$ references. 
Indeed, if two references access the same non-$\imm$ reference, they are necessarily in sharing, as shown below. 

\begin{proposition}\label{prop:non-imm-sharing}
Let $\IsWFMem{\Gamma}{\mem}$. 
If $\reachable{\x}{\mem}{\z}$ and  $\reachable{\y}{\mem}{\z}$ and $\getModif{\Gamma}{\z} \ne \imm$, then 
$\SharingRel{\x}{\Gamma,\mem}{\y}$. 
\end{proposition}
\begin{proof}
We first show that $\SharingRel{\x}{\mem}{\z}$. 
The proof is by induction on the definition of $\reachableSymbol{\mem}$. 
\begin{description} 
\item [Case $\x=\z$] The thesis trivially holds by reflexivity of $\sharingRelSymbol{\Gamma,\mem}$. 
\item [Case $\mem(\x)=\Object{\C}{\val_1,\ldots,\val_n}$, $\x' = \val_i$ for some  $i\in 1..n$ and $\reachable{\x'}{\mem}{\z}$]
We know that $\getModif{\Gamma}{\x}$, \\$\getModif{\Gamma}{\x'} \le \mut$ because $\getModif{\Gamma}{\x'} = \imm$ (or $\getModif{\Gamma}{\x} = \imm$) would imply $\getModif{\Gamma}{\z} = \imm$, by \cref{lem:mod-deep}, which is a contradiction. 
Therefore, by \cref{def:sharing-rel-with-imm}, we have $\SharingRel{\x}{\Gamma,\mem}{\x'}$ and by induction hypothesis, we get $\SharingRel{\x'}{\Gamma,\mem}{\z}$; 
then we get $\SharingRel{\x}{\Gamma,\mem}{\z}$ by transitivity of $\sharingRelSymbol{\Gamma,\mem}$. 
\end{description}
By the same argument, we also get $\SharingRel{\y}{\Gamma,\mem}{\z}$. 
Then, by transitivity of $\sharingRelSymbol{\Gamma,\mem}$, we get the thesis. 
\end{proof}

\begin{corollary}\label{cor:mut-imm-caps}
If $\IsWFConf{\Gamma}{\e}{\mem}{\TypeMod{\C}{\capsule}}$, and 
$\reducestar{\ExpMem{\e}{\mem}}{\ExpMem{\y}{\mem'}}$, then 
there exists $\Gamma'$ such that $\IsWFMem{\Gamma'}{\mem'}$, $\Gamma \modord \Gamma'$ and, for all  $\x\in\dom{\mem}$, 
$\getModif{\Gamma'}{\x} = \mut$ and $\reachable{\x}{\mem'}{\z}$ and $\reachable{\y}{\mem'}{\z}$ imply $\getModif{\Gamma'}{\z} = \imm$. 
\end{corollary}
\begin{proof}
By \cref{thm:caps}, we get $\IsWFMem{\Gamma'}{\mem'}$ and $\SharingRel{\x}{\Gamma',\mem'}{\y}$ imply $\getModif{\Gamma'}{\x} \ne \mut$. 
Suppose $\getModif{\Gamma'}{\z} \ne \imm$, then, by \cref{prop:non-imm-sharing}, we get $\SharingRel{\x}{\Gamma',\mem'}{\y}$, hence $\getModif{\Gamma'}{\x} \ne \mut$, which contradicts the hypothesis. 
Therefore, $\getModif{\Gamma'}{\z} = \imm$.
\end{proof}
   
In the extended type system, we can also nicely characterize the property guaranteed by the $\imm$ references. Notably, the reachable object graph of an $\imm$ modifier cannot be modified during the execution. 
We first show that fields of an $\imm$ reference cannot change in a single computation step. 
\begin{lemma} \label{lem:imm-step} 
If $\IsWFConf{\Gamma}{\e}{\mem}{\T}$, and $\getModif{\Gamma}{\x}=\imm$, and $\reduce{\ExpMem{\e}{\mem}}{\ExpMem{\e'}{\mem'}}$, then $\mem(\x)=\mem'(\x)$.
\end{lemma} 
\begin{proof}
By induction on reduction rules. 
The key case is rule \refToRule{field-assign}. 
We have $\e=\FieldAssign{\y}{\f}{\val}$ and $\IsWFConf{\Gamma}{\FieldAssign{\y}{\f}{\val}}{\mem}{\T}$.  Let $\getModif{\Gamma}{\y}=\modif$.
Either rule \refToRule{T-field-assign} was the last rule applied, or one of the non syntax-directed rules was applied after \refToRule{T-field-assign}. In the former case
$\modif=\mut$ or $\modif= \capsule$ if rule \refToRule{T-Sub} was applied before \refToRule{T-field-assign}.
In the latter case $\modif$ could only be equal to the previous modifier or $\modif=\seal$ if rule \refToRule{T-Prom} was applied and the previous modifier was $\mut$.  
Therefore, $\y\ne\x$ and so we have the thesis. 
For all other computational rules the thesis is immediate as they do not change the memory, and for \refToRule{ctx} the thesis immediately follows by induction hypothesis.  
\end{proof}
Thanks to \cref{lem:mod-deep}, we can show that the reachable object graph of an $\imm$ reference contains only $\imm$ references. 
Hence, by the above lemma we can characterise $\imm$ references as follows: 
\begin{theorem}[Immutable reference]
If $\IsWFConf{\Gamma}{\e}{\mem}{\T}$, $\getModif{\Gamma}{\x}=\imm$, and $\reducestar{\ExpMem{\e}{\mem}}{\ExpMem{\e'}{\mem'}}$, then
$\reachable{\x}{\mem}{\y}$ implies $\mem(\y)=\mem'(\y)$.
\end{theorem}
\begin{proof}
By induction on the definition of $\evstar$
\begin{description} 
\item [Case $\ExpMem{\e}{\mem}=\ExpMem{\e'}{\mem'}$] The thesis trivially holds.
\item [Case $\reduce{\ExpMem{\e}{\mem}}{\ExpMem{\e_1}{\mem_1}}\evstar{\ExpMem{\e'}{\mem'}}$]
Since $\getModif{\Gamma}{\x} = \imm$, for all $\y$ such that $\reachable{\x}{\mem}{\y}$, 
by \cref{lem:mod-deep} $\getModif{\Gamma}{\y} = \imm$, hence, by \cref{lem:imm-step}, $\mem_1(\y) = \mem(\y)$. 
Therefore, it is easy to check that $\reachable{\x}{\mem}{\y}$ implies $\reachable{\x}{\mem_1}{\y}$. 
By \cref{thm:subj-red-extended}, $\IsWFConf{\Delta}{\e_1}{\mem_1}{\T}$ and $\getModif{\Gamma}{\x} \le \getModif{\Delta}{\x}$, hence $\getModif{\Delta}{\x} = \imm$. 
Then, by induction hypothesis, $\mem'(\y) = \mem_1(\y)$, hence the thesis. 
\qedhere 
\end{description}
\end{proof} 



\section{Expressive power}\label{sect:comparison}

We discuss the expressive power of the type-and-coeffect system in \refToSection{extended}, \bez comparing it with the two most closely related proposals by  \citet{GordonEtAl12} and by \citet{ClebschEtAl15,Clebsch17}, abbreviated as \Gordon\  and Pony, respectively.  
The takeaway is that our promotion mechanism is much more powerful than their recovery, since sharing is taken into account; 
on the other hand, the expressive power allowed by some of their annotations on fields is beyond the scope of this paper. \eez
We assume a syntax enriched
by the  usual programming constructs. 
 

Before the work in \Gordon, the capsule property was only ensured in simple situations, such as using a primitive deep clone operator, or composing subexpressions with the same property.
The type system in \Gordon\ has been an important step, being the first to introduce \emph{recovery}. That is, this type system contains two typing rules which allow recovering isolated\footnote{Their terminology for capsule.} or immutable references from arbitrary code checked in contexts containing only isolated or immutable variables. Such rules are rephrased below in our style for better comparison.
\begin{small}
\begin{quote}
$\NamedRule{t-recov-iso}{\IsWFExp{\Gamma}{\e}{\TypeMod{\C}{\mut}}}{\IsWFExp{\Gamma}{\e}{\TypeMod{\C}{\capsule}}}
{\IsoOrImm{\Gamma}}$\BigSpace$\NamedRule{t-recov-imm}{\IsWFExp{\Gamma}{\e}{\TypeMod{\C}{\readonly}}}{\IsWFExp{\Gamma}{\e}{\TypeMod{\C}{\imm}}}
{\IsoOrImm{\Gamma}}$
\end{quote}
\end{small}
where $\IsoOrImm{\Gamma}$ means that, for all $\VarType{\x}{\TypeMod{\C}{\modif}}$ in $\Gamma$, $\modif\leq\imm$.

As the reader can note, this is exactly in the spirit of coeffects, since typechecking also takes into account the way the \emph{surrounding context} is used. By these rules \Gordon\ typechecks, e.g., the following examples, assuming  the language   has  threads with a parallel operator:
\begin{lstlisting}
isolated IntList l1 = ...
isolated IntList l2 = ...
l1.map(new Incrementor()); || l2.map(new Incrementor());
\end{lstlisting}
The two threads do not interfere, since they operate and can mutate disjoint object graphs.
\begin{lstlisting}
isolated IntBox increment(isolated IntBox b){
  b.value++;//b converted to mut by subtyping
  return b//convert b *back* to isolated by recovery
}
\end{lstlisting}
An isolated object can be mutated\footnote{We say that the capsule is \emph{opened}, see in the following.}, and then isolation can be recovered, since the context only contains isolated or immutable references.

In Pony, the ideas of \Gordon\ are extended to a richer set of modifiers. In their terminology \texttt{val} is immutable, \texttt{ref} is mutable, \texttt{box} is read-only. An ephemeral isolated reference \lstinline{iso^} is similar to a $\capsule$ reference in our calculus, whereas non ephemeral \texttt{iso} references are more
similar to  the  isolated fields discussed below. Finally, \texttt{tag} only allows object identity checks  and  asynchronous method invocation,   and \texttt{trn} (transition) is a subtype of \texttt{box} that can be converted to \texttt{val}, providing a way to create values without using isolated references. 
The last two modifiers have no equivalent in  \Gordon\ or our work.

The type-and-coeffect-system in \refToSection{extended} shares with \Gordon\ and Pony the modifiers  $\mut$, $\imm$, $\readonly$, and $\capsule$ with their subtyping relation, a similar operation to combine modifiers, and the  key role  of recovery.
However, rule \refToRule{t-prom} is much more powerful than the recovery rules reported above, which definitely forbid $\readonly$ and $\mut$ variables in the context. Rule \refToRule{t-prom}, instead, allows such variables when they are not connected to the final result, as smoothly derived from coeffects which compute sharing.  For instance, with the classes of \refToExample{ex1}, the following two examples would be ill-typed in \Gordon\ and Pony:

\begin{lstlisting}
caps C es1 = {B z = new B(2); x.f1=y; new C(z,z)}
caps C es2 = {B z = new B(y.f=y.f+1); new C(z,z) }
\end{lstlisting}
 Below  is the corresponding Pony code.

\begin{lstlisting}
class B
  var f: U64
  new create(ff:U64) => f=ff
class C
  var f1: B
  var f2: B 
  new create(ff1: B ref, ff2: B ref) =>  f1=ff1; f2=ff2    
var x: B ref = ...
var y: B ref = ...
var es1: C iso = recover iso var z = B(2); x.f1=y; C(z,z) end//ERROR
var es2: C iso = recover iso var z = B(y.f=y.f+1); C(z,z) end//ERROR
\end{lstlisting}

\bez For comparison on a more involved example, let us \eez add to class \lstinline{A}  of \refToExample{ex2-decorated}
the method  \lstinline{nonMix} that follows:
\begin{lstlisting}
  A nonMix $\meta{[^{\{\link\}}]}$(A$\meta{^{\{\res\}}}$a) {this.f.f=a.f.f; a} // $\link\neq\res$ \end{lstlisting}
Consider  the following code: 
\begin{lstlisting}
A a1= new A(new B(0));
caps A mycaps = {A a2 = new A(new B(1));
  a1.mix(a2).clone() // (1)
  // a1.mix(a2).clone().mix(a2) // (2)
  // a1.nonMix(a2) // (3)
}
\end{lstlisting}
The corresponding Pony code is as follows: 
\begin{lstlisting}
class B
  var f:U64
  new create(ff:U64) => f=ff
  fun box clone():B iso^ => recover B(f) end 
class A
  var f:B
  new create(ff:B) => f=ff
  fun ref mix(a:A):A => this.f=a.f; a
  fun ref nonMix(a:A):A => f.f=a.f.f; a
  fun box clone():A iso^ => var x:B iso = f.clone(); recover A(consume x) end
var a1 = A(B(0))
var a2 = A(B(1)); var l1:A iso = a1.mix(a2).clone() // (1) OK
var l2:A iso=recover var a2=A(B(1));a1.mix(a2).clone().mix(a2) end//(2) ERROR
var l3:A iso= recover var a2 = A(B(1)); a1.nonMix(a2) end // (3) ERROR
\end{lstlisting}

  As in our approach, Pony is able to discriminate line (1)  from line  (2), causing code to be well-typed and ill-typed, respectively. However, to do so,  Pony needs an explicit modifier \lstinline{iso^} in the return type of \lstinline{clone}, whereas, as  noted after the code of  \refToExample{ex2-decorated}, in our approach the return type of \lstinline{clone} can be $\mut$, since the fact that there is no connection between the result and $\this$ is expressed by the coeffect.
Moreover, in order to be able to obtain an \lstinline{iso} from the \lstinline{clone} method, Pony needs to insert explicit \lstinline{recover} instructions. In the case of class \lstinline{A} where the field is an  object, Pony needs to  explicitly use \lstinline{consume} to ensure uniqueness, whereas in our approach promotion takes place implicitly  and uniqueness is ensured by linearity.
Finally, Pony rejects line (3) as well, whereas, in our approach, this expression is correctly recognized to be a capsule, since the external variable \lstinline{a1} is modified, but not connected to the final result. 

\bez Moreover, our type system can prevent sharing of parameters, something which is not possibile in \Gordon\ and Pony. \eez
\label{ex:teams}
The following method takes  two teams, \lstinline{t1} and \lstinline{t2}, as parameters.  Both want to add a reserve player from their respective lists \lstinline{p1} and
\lstinline{p2}, sorted with best players first.  To keep
the game fair, the two reserve players can only be added if they have the same
skill level.
{\small
\begin{lstlisting}
static void addPlayer(Team$\meta{^{\{\link\}}}$t1, Team$\meta{^{\{\link'\}}}$t2, Players$\meta{^{\{\link\}}}$p1, Players$\meta{^{\{\link'\}}}$p2)
{/*$\link\neq\link'$*/} {while(true){
 if(p1.isEmpty()||p2.isEmpty()) {/*error*/}
 if(p1.top().skill==p2.top().skill){t1.add(p1.top());t2.add(p2.top());return;}
 else{removeMoreSkilled(p1,p2);}
}
\end{lstlisting}
}
\noindent The sharing coeffects express that each team can only  add players from its list 
of reserve players. 

\bez As mentioned at the beginning of the section, an important feature supported by \Gordon\ and Pony, and not by our type system, are isolated fields. To ensure that accessing an isolated field will not introduce aliasing, they use \eez an ad-hoc semantics, called \emph{destructive read}, see also \citet{Boyland10}.  In \Gordon,  an isolated field can only be read by a command \lstinline{x=consume(y.f)}, assigning  the value  to \lstinline{x} 
and updating the field to \lstinline{null}. Pony supports the command \lstinline{(consume x)}, with the semantics that the reference becomes empty. Since fields cannot be empty, they cannot be arguments of \lstinline{consume}. By relying on the fact that assignment returns the left-hand side value, in Pony  one writes  \lstinline{x=y.f=(consume z)}, with \lstinline{z} isolated. In this way, the field value is   assigned to  \lstinline{x}, and the field is updated to a new isolated reference. 

We  prefer  to avoid destructive reads since they can cause subtle bugs, see \citet{GianniniSZC19} for a discussion. 
\bez We leave to future work the development of an alternative solution, notably investigating how to extend our affine handling of $\capsule$ variables to fields. \eez
\bez Concerning this point, another feature allowing more flexibility in \Gordon\ and Pony is that iso variables can be ``consumed'' only once, but accessed more than once. For example in Pony we can write \lstinline{var c: C iso=recover var z=B(2); C(z,z) end; c.f1=recover B(1) end;c.f2=recover B(1) end}. 
To achieve this in our type system, one needs to explicitly open the capsule by assigning it to a local mutable variable, modify it and finally apply the promotion to recover the capsule property. \eez 



\section{Related work}\label{sect:related}

\subsection{Coeffect systems}

Coeffects were  first  introduced by \citet{PetricekOM13} and further analyzed by \citet{PetricekOM14}. 
In particular, \citet{PetricekOM14} develops a  generic  coeffect system which augments the simply-typed $\lambda$-calculus with context annotations indexed by \emph{coeffect shapes}. 
The proposed framework is very abstract, and the authors focus only on two opposite instances: 
structural (per-variable) and flat (whole context) coeffects, identified by specific choices of context shapes.

Most of the  subsequent  literature on coeffects focuses on structural ones, for which there is a clear algebraic description in terms of semirings. 
This was first noticed by \citet{BrunelGMZ14}, who developed a framework for structural coeffects 
for a functional language. 
This approach is inspired by a generalization of the exponential modality of linear logic, see, e.g., \citet{BreuvartP15}. 
That is, the distinction between linear and  unrestricted  variables of linear systems is generalized to have variables  
decorated by coeffects, or {\em grades}, that determine how much they can be used. 
In this setting, many advances have been made to combine coeffects with other programming features, such as 
computational effects \cite{GaboardiKOBU16,OrchardLE19,DalLagoG22}, 
dependent types \cite{Atkey18,ChoudhuryEEW21,McBride16}, and 
polymorphism \cite{AbelB20}. 
A fully-fledged functional programming language, called Granule, has been presented by \citet{OrchardLE19}, inspired by the principles of coeffect systems.

As already mentioned, \citet{McBride16} and \citet{WoodA22}  observed that contexts in a structural coeffect system form a module over the semiring of grades, event though they do not use this structure in its full generality, restricting themselves to free modules, that is, to structural coeffect systems.
This algebraic structure nicely describes operations needed in typing rules, and we believe  it  could be a clean framework for coeffect systems beyond structural ones. 
Indeed, the sharing coeffect system  in this paper  provides a non-structural instance. 

\subsection{Type systems controlling sharing and mutation}
 The literature  on type systems controlling sharing and mutation is  vast.   
 In \refToSection{comparison} we provided a  comparison with the most closely related approaches. We briefly discuss here other works.

The approach based on modifiers is extended in other proposals \cite{HallerOdersky10,CastegrenWrigstad16} to compositions of one or more \emph{capabilities}. The modes of the capabilities in a type control how resources of that
type can be aliased. The compositional aspect of capabilities is an important difference
from modifiers, as accessing different parts of an object through different capabilities in the same type gives different properties. 
By using capabilities it is possible to obtain an expressivity similar to our type system, although with different sharing notions and syntactic constructs. For instance, the \emph{full encapsulation} notion by \citet{HallerOdersky10}, apart from the fact that sharing of immutable objects is not allowed, is equivalent to our $\capsule$ guarantee.
Their model has a higher syntactic/logic overhead to explicitly  track regions.
As for all work  preceding~\citet{GordonEtAl12},  objects need to be born \lstinline{unique} and the type system 
permits manipulation of data preserving their uniqueness. With recovery/promotion,  instead,  we can use normal code designed to work on conventional shared data, and then
recover uniqueness.

An  alternative  approach to modifiers to restrict the usage of references is that of \emph{ownership}, based on \emph{enforcing invariants} rather than deriving properties. We refer to the recent work  of  \citet{MilanoMT22} for an up-to-date survey. The Rust language, considering its ``safe'' features \cite{JungJKD18}, belongs to this family as well, and uses ownership for memory management.
In Rust, all references which support mutation are required to be affine,  thus ensuring a unique entry point to a portion of mutable memory. 
This relies on a relationship between linearity and uniqueness recently clarified by \citet{MarshallVO22}, which proposes a linear calculus with modalities for non-linearity and uniqueness with a somewhat dual behaviour.    
In our approach, instead, the capsule concept models an efficient \emph{ownership transfer}. In other words, when an object $\x$ is ``owned'' by $\y$, it remains always true that $\y$  can only be accessed through  $\x$, whereas the capsule notion is dynamic: a capsule can be ``opened'', that is, assigned to a standard reference and modified, since we can always recover the capsule guarantee.\footnote{Other work in  the  literature supports ownership transfer, see, e.g., \citet{MullerRudich07} and \citet{ClarkeWrigstad03}, however not of {the whole} reachable object graph.}.
We also mention that, whereas in this paper all properties are \emph{deep}, that is, inherited by the reachable object graph, 
most ownership approaches allows one to distinguish
subparts of the reachable object graph that are referred  to  but not logically owned. This viewpoint has some advantages, for example Rust uses ownership to control deallocation without a garbage collector.

\section{Conclusion}\label{sect:conclu}
 The main achievement of  this  paper is to show that sharing and mutation can be tracked by a coeffect system, thus reconciling two distinct views in  the  literature on type systems for mutability control: substructural type systems, and graph-theoretic properties on heaps.  Specifically,  the contributions of the paper  are  the following:
\begin{itemize}
\item a minimal framework formalizing ingredients of coeffect systems
\item a coeffect system, for an imperative Java-like calculus, where coeffects express \emph{links} among variables and with the final result introduced by the execution
\item an enhanced type system modeling complex features for uniqueness and immutability.
\end{itemize}

The enhanced type system (\refToSection{extended}) cleanly integrates and formalizes language designs by \citet{GianniniSZC19} and \citet{GianniniRSZ19}, as detailed below: 
\begin{itemize}
\item \citet{GianniniSZC19} supports promotion through a very complex type system; moreover, sharing of two variables/parameters cannot be prevented, as, e.g., in the example  on  page \pageref{ex:teams}.
\item \citet{GianniniRSZ19} has a more refined tracking of sharing allowing  us  to express  this   example, but does not handle immutability.
\item In both works, significant properties are expressed and proved with respect to a non-standard reduction model where memory is encoded in the language itself. 
\end{itemize} 
Each of the contributions of the paper opens interesting research directions. The minimal framework we defined, modeling coeffect contexts as modules, includes structural coeffect systems, where the coeffect of each variable can be computed independently (that is, the module operators are defined pointwise), and coeffect systems  such  as those in this paper, which can be considered \emph{quasi-structural}. Indeed, coeffects cannot be computed per-variable (notably, the sum and multiplication operator of the module are \emph{not} defined pointwise), but can still be expressed by annotations on single variables. 
 This also shows a difference with existing graded type systems explicitly derived from bounded linear logic, which generally consider purely structural (that is, computable per-variable) grades.\footnote{ This corresponds to a free module, while the module of sharing coeffects is not free (just a retraction of a free one).}. 
In future work we plan to develop the metatheory of the framework, and to investigate its appropriateness both for other quasi-structural cases, and for coeffects which are truly \emph{flat}, that is, cannot be expressed on a per-variable basis. 
This could be coupled with the  design of a $\lambda$-calculus with a generic module-based coeffect system, substantially \mbox{revising \cite{PetricekOM14}. }

In the type system in \refToSection{extended}, coeffects and modifiers are distinct, yet interacting, features. We will investigate a framework where modifiers, or, more generally, \emph{capabilities}  \cite{HallerOdersky10,GordonEtAl12,Gordon20}, are formalized as graded modal types, which are, roughly, types annotated with coeffects (grades) \cite{BrunelGMZ14,OrchardLE19,DalLagoG22}, thus providing a formal foundation for the ``capability'' notion in  the   literature. A related interesting question is  which kinds of modifier/capability can be expressed \emph{purely} as coeffects. The read-only property, for instance, could be expressed by enriching the sharing coeffect  with  a Read/Write component (Read by default), so that in a field assignment, variables connected to  \lstinline{res}  in the context of the left-hand expression are marked as Write.


Concerning the specific type system in \refToSection{extended},  additional  features are necessary to have a more realistic language. 
 Two important examples are: suitable syntactic sugar for coeffect annotations in method types and
relaxation of coeffects when redefining methods in the presence of inheritance. 

Finally, an interesting objective in the context of Java-like languages is to allow variables (notably, method parameters) to be annotated by user-defined coeffects, written by the programmer by extending a predefined abstract class,
in the same way user-defined exceptions extend the \lstinline{Exception} predefined class.  This approach would be partly similar to that of Granule \cite{OrchardLE19}, where, however, coeffects cannot be extended by the programmer.  A first step in this direction is presented by \citet{BianchiniDGZ22}.


\bibliography{main}

\clearpage

\appendix

\section{Proof of Theorem~\ref{thm:subj-red}}

\brb
\begin{definition}
$\Gamma\prom\Gamma'$ if  
\begin{enumerate}
\item $\Gamma'= {\Gamma}$ or
\item $\Gamma'=\{\link\}\ctxmul\Gamma$ with $\link$ fresh
\end{enumerate}
\end{definition}

\begin{lemma}\label[lemma]{lem:nonStruct}
If $\Deriv:\IsWFExp{\Gamma}{{\e}}{\T}$, then there is a subderivation $\Deriv':\IsWFExp{\Gamma'}{{\e}}{\T}$ of 
$\Deriv$ ending with a syntax-directed rule and $\Gamma'\prom\Gamma$
\end{lemma}

\begin{proof}
By induction on $\Deriv:\IsWFExp{\Gamma}{{\e}}{\T}$
\begin{description}
\item if $\Deriv$ ends with a syntax-directed rule then $\Gamma'=\Gamma$ and $\Deriv=\Deriv'$. If $\Deriv$ ends with rule \refToRule{t-prim} we have $\Gamma=\{\link\}\ctxmul\Delta$ with $\link$ fresh and so $\Delta\prom\Gamma$ where $\IsWFExp{\Delta}{{\e}}{\T}$. By induction hypothesis on the premise we have a subderivation $\Deriv':\IsWFExp{\Gamma'}{{\e}}{\T}$ ending with a syntax-directed rule and $\Gamma'\prom\Delta$. We have 2 cases: 
\begin{itemize}
\item $\Gamma'=\Delta$: by this and by $\Delta\prom\Gamma$ we obtain $\Gamma'\prom\Gamma$, that is, the thesis
\item $\Delta=\{\link'\}\ctxmul\Gamma'$ with $\link'$ fresh: by this and $\Gamma=\{\link\}\ctxmul\Delta$ we obtain $\Gamma=\{\link\}\ctxmul\{\link'\}\ctxmul\Gamma'=\{\link'\}\ctxmul\Gamma'$ so we have $\Gamma'\prom\Gamma$ and so the thesis
\end{itemize}
\end{description}
\end{proof}

\begin{lemma}[Inversion]\label[lemma]{lem:inversion}
If $\IsWFExp{\Gamma}{\e}{\T}$ then exists a context $\Gamma'$ such that $\Gamma'\prom\Gamma$ and $\IsWFExp{\Gamma'}{\e}{\T}$ and the following properties holds:
\begin{enumerate}

\item \label{lem:inversion:x} If $\e = \x$ then $\Gamma' = \emptyset\ctxmul\Gamma''\ctxsum \VarTypeCoeffect{\x}{\T}{\{\res\}}$
\item \label{lem:inversion:cost} If $\e = \const$, then $\T=\PT_\const$.
\item \label{lem:inversion:acc}  If $\e = \FieldAccess{\e}{\f_i}$, then $\IsWFExp{\Gamma'}{\e}{\C}$ and $\fields{\C}=\Field{\T_1}{\f_1} \ldots \Field{\T_n}{\f_n}$ and $\T_i=\T$.
\item \label{lem:inversion:ass}  If $\e = \FieldAssign{\e}{\f_i}{\e'}$ then $\IsWFExp{\Gamma_1}{\e}{\C}$ and $\fields{\C}=\Field{\T_1}{\f_1} \ldots \Field{\T_n}{\f_n}$ and $\T_i=\T$ and $\IsWFExp{\Gamma_2}{\e'}{\T_i}$ with $\Gamma'= \Gamma_1 \shsum \Gamma_2$.
\item \label{lem:inversion:new} If $\e = \ConstrCallTuple{\C}{\e}{n}{\C}$, then we have $\fields{\C}=\Field{\T_1}{\f_1} \ldots \Field{\T_n}{\f_n}$ and
$\IsWFExp{\Gamma_i}{\e_i}{\T_i}$ for all  $i\in 1..n$ and
$\Gamma=\Gamma_1\shsum \ldots \shsum\Gamma_n$.

\item \label{lem:inversion:invk} If $\e = \MethCallTuple{\e_0}{\m}{\e}{n}$, then $\IsWFExp{\Gamma_0}{\e_0}{\C}$ and
$\mtype{\C}{\m}=\funType{\X_{\this},\coeffectType{\T_1}{\X_1} \ldots \coeffectType{\T_n}{\X_n}}{\T}$
and
$\IsWFExp{\Gamma_i}{\e_i}{\T_i}$ for all  $i\in 1..n$ and
$\Gamma'=(\X_{\this} \ctxmul \Gamma_0) \shsum(\X_1 \ctxmul \Gamma_1)\shsum \ldots \shsum (\X_n \ctxmul \Gamma_n)$.
\item \label{lem:inversion:bl} If $\e = \Block{\T}{\x}{\e}{\e'}$ then $\Gamma'=(\X \shsum \{ \link \}) \ctxmul \Gamma_1 \shsum \Gamma_2$ where $ \link$ is fresh and $\IsWFExp{\Gamma_1}{\e}{\T}$
and $\IsWFExp{\Gamma_2, \VarTypeCoeffect{\x}{\T'}{\X}}{\e'}{\T'}$.
\end{enumerate}
\end{lemma}

\begin{proof}
\begin{enumerate}
	\item If $\Deriv:\IsWFExp{\Gamma}{\x}{\T}$ then, by \cref{lem:nonStruct}, we know that exists a derivation $\Deriv':\IsWFExp{\Gamma'}{{\x}}{\T}$ subderivation of $\Deriv$ ending with a syntax-directed rule and $\Gamma'\prom\Gamma$. Since the last applied rule in $\Deriv'$ must be \refToRule{t-var}, we have $\Gamma' = \emptyset\ctxmul\Gamma''\ctxsum\VarTypeCoeffect{\x}{\T}{\{\res\}}\shord\VarTypeCoeffect{\x}{\T}{\{\res\}}$
\item If $\IsWFExp{\Gamma}{\const}{\T}$ then the last applied rule can be \refToRule{t-const} or \refToRule{t-prim}. In both cases we have the thesis

\item If $\Deriv:\IsWFExp{\Gamma}{\FieldAccess{\e}{\f_i}}{\T}$  then, by \cref{lem:nonStruct}, we know that exists a derivation $\Deriv':\IsWFExp{\Gamma'}{{\FieldAccess{\e}{\f_i}}}{\T}$ subderivation of $\Deriv$ ending with a syntax-directed rule and $\Gamma'\prom\Gamma$. Since the last applied rule in $\Deriv'$ must be \refToRule{t-access} we have that $\IsWFExp{\Gamma'}{\e}{\C}$ with $\fields{\C}=\Field{\T_1}{\f_1} \ldots \Field{\T_n}{\f_n}$ and $\T_i=\T$

\item If $\Deriv:\IsWFExp{\Gamma}{\FieldAssign{\e}{\f_i}{\e'}}{\T}$  then, by \cref{lem:nonStruct}, we know that exists a derivation $\Deriv':\IsWFExp{\Gamma'}{{\FieldAssign{\e}{\f_i}{\e'}}}{\T}$ subderivation of $\Deriv$ ending with a syntax-directed rule and $\Gamma'\prom\Gamma$. Since the last applied rule in $\Deriv'$ must be \refToRule{t-assign}, we have that $\Gamma' = \Delta_1\ctxsum\Delta_2$ such that $\IsWFExp{\Delta_1}{\e}{\C}$ and $\IsWFExp{\Delta_2}{\e'}{\T_i}$ with $\fields{\C}=\Field{\T_1}{\f_1} \ldots \Field{\T_n}{\f_n}$ and $\T_i=\T$

\item If $\Deriv:\IsWFExp{\Gamma}{\ConstrCallTuple{\C}{\e}{n}}{\C}$  then, since the last applied rule in $
\Deriv$ must be \refToRule{t-new}, we have that $\Gamma = \Delta_1\ctxsum...\ctxsum\Delta_n$ such that $\fields{\C}=\Field{\T_1}{\f_1} \ldots \Field{\T_n}{\f_n}$ and $\IsWFExp{\Delta_i}{\e_i}{\T_i}$ for all  $i\in 1..n$.

\item If $\Deriv:\IsWFExp{\Gamma}{\MethCallTuple{\e_0}{\m}{\e}{n}}{\T}$  then, by \cref{lem:nonStruct}, we know that exists a derivation $\Deriv':
\IsWFExp{\Gamma'}{\MethCallTuple{\e_0}{\m}{\e}{n}}{\T}$ subderivation of $\Deriv$ ending with a syntax-directed rule and $\Gamma'\prom\Gamma$. Since the last applied rule in $
\Deriv'$ must be \refToRule{t-invk}, we have $\IsWFExp{\Delta_0}{\e_0}{\C}$,
$\mtype{\C}{\m}=\funType{\X_{\this},\coeffectType{\T_1}{\X_1} \ldots \coeffectType{\T_n}{\X_n}}{\T}$
and
$\IsWFExp{\Delta_i}{\e_i}{\T_i}$ for all  $i\in 1..n$ and
$\Gamma'=(\X_{\this} \ctxmul \Delta_0) \shsum(\X_1 \ctxmul \Delta_1)\shsum \ldots \shsum (\X_n \ctxmul \Delta_n)$.

\item If $\Deriv:\IsWFExp{\Gamma}{\Block{\T}{\x}{\e}{\e'}}{\T'}$  then, by \cref{lem:nonStruct}, we know that exists a derivation $\Deriv':
\IsWFExp{\Gamma'}{\Block{\T}{\x}{\e}{\e'}}{\T}$ subderivation of $\Deriv$ ending with a syntax-directed rule and $\Gamma'\prom\Gamma$. Since the last applied rule in $
\Deriv'$ must be \refToRule{t-block}, we have $\Gamma'=(\X \shsum \{ \link \}) \ctxmul \Delta' \shsum \Gamma''$ where $ \link$ is fresh $\IsWFExp{\Delta'}{\e}{\T}$
and $\IsWFExp{\Gamma'', \VarTypeCoeffect{\x}{\T'}{\X}}{\e'}{\T'}$.
\end{enumerate}
\end{proof} 
\erb

\begin{lemma}\label{lem:closed-ctx}
If $\Gamma$ is a closed context then, for all $\x,\y\in\dom{\Gamma}$, 
$\getCoeff{\Gamma}{\x}\cap\getCoeff{\Gamma}{\y} \ne \emptyset$ implies $\getCoeff{\Gamma}{\x} = \getCoeff{\Gamma}{\y}$. 
\end{lemma}
\begin{proof}
Immediate from the definition of $\closure{\_}$. 
\end{proof}

\begin{lemma} \label{lem:links-closure}
Let $\Gamma$ be a closed context and $\Delta = \Gamma,\VarTypeCoeffect{\x}{\T}{\X}$. 
Then, for all $\y\in\dom{\Delta}$, we have 
\[ 
\getCoeff{\closure\Delta}{\y} =  \begin{cases}
\bigcup \{ \getCoeff{\Delta}{\z} \mid \getCoeff{\Delta}{\z}\cap\X\ne\emptyset \} & \getCoeff{\Delta}{\y}\cap\X\ne\emptyset \\ 
\getCoeff{\Delta}{\y} & \text{otherwise} 
\end{cases}\] 
\end{lemma}
\begin{proof}
Suppose $\Gamma = \VarTypeCoeffect{\x_1}{\T_1}{\X_1},\ldots,\VarTypeCoeffect{\x_n}{\T_n}{\X_n}$ with $n \ge 0$, $\VarTypeCoeffect{\x}{\T}{\X} = \VarTypeCoeffect{\x_{n+1}}{\T_{n+1}}{\X_{n+1}}$ and set 
$\Theta = \VarTypeCoeffect{x_1}{\T_1}{\Y_1},\ldots,\VarTypeCoeffect{\x_{n+1}}{\T_{n+1}}{\Y_{n+1}}$, where  for all $i \in 1..n+1$, we have 
\[ \Y_i = \begin{cases}
\bigcup \{ \X_j \mid j\in1..n+1, \X_j\cap\X_{n+1}\ne\emptyset \}  & \X_i\cap\X_{n+1}\ne\emptyset \\
\X_i & \text{otherwise} 
\end{cases}\] 
The inequality $\Theta\hat\subseteq\closure\Delta$ is trivial by definition of $\closure{\_}$.  We know that $\Delta \subseteq \Theta$, so  to get the other direction, we just have to show that $\Theta$ is closed. 
To this end, let $i\in 1..n+1$, $\link_1\in\Y_i$ and $\link_1,\link_2\in\Y_j$ for some $j \in 1..n+1$, then we have to prove that $\link_2\in\Y_i$. 
We distinguish to cases. 
\begin{itemize}
\item If $\X_j\cap\X_{n+1} = \emptyset$, then $\Y_j = \X_j$. 
We observe that $\link_1\notin\X_k$ for any $k \in 1..n+1$ such that $\X_k\cap\X_{n+1} \ne\emptyset$. 
This is obvious for $k = n+1$, as it is against the assumption $\X_j\cap\X_{n+1}=\emptyset$. 
For $k \in 1..n$, 
since $\Gamma$ is closed, by \cref{lem:closed-ctx}, we would bet $\X_j = \X_k$, and so $\X_j\cap\X_{n+1} \ne \emptyset$, which is again a contraddiction. 
This implies that $\X_i\cap\X_{n+1}=\emptyset$ and so $\Y_i = \X_i$. 
Therefore, applying again \cref{lem:closed-ctx}, we get $\Y_i = \X_i = \X_j = \Y_j$, thus $\link_2\in\Y_i$, as needed. 
\item If $\X_j\cap\X_{n+1}\ne\emptyset$, then we have $\Y_j = \bigcup\{\X_k \mid k \in 1..n+1,\X_k\cap\X_{n+1}\ne\emptyset\}$, hence 
$\link_1\in\X_h$ for some $h\in 1..n+1$ such that $\X_h\cap\X_{n+1}\ne\emptyset$. 
By a argument similar to the previous point, we get that  $\X_i\cap\X_{n+1}\ne\emptyset$, hence, by definition, $\Y_i = \Y_j$ that proves the thesis. 
\end{itemize}
\end{proof}


\brb
\begin{lemma}\label[lemma]{lem:InvCtx}[Inversion for context]
If $\IsWFExp{\Gamma}{\Ctx{\e}}{\T}$, then for some $\Gamma',\Delta$,  $\x\not\in\dom{\Gamma}$, $\X$ and $\T'$, $\Gamma = (\X\ctxmul\Delta)\shsum\Gamma'$,
$\IsWFExp{\Gamma'\shsum\VarTypeCoeffect{\x}{\T'}{\X}}{\Ctx{\x}}{\T}$ and
$\IsWFExp{\Delta}{\e}{\T'}$.
\end{lemma}
\begin{proof}  
The proof is by induction on $\ctx$. 
We only show some cases, the others are analogous. 
\begin{description}
\item[$\ctx = \emptyctx$] \label{lem:context:hole}
Just take $\Gamma'=\emptyset$,
$\x\not\in\dom{\Gamma}$, $\X=\{\res\}$, $\T'=\T$ and $\Delta=\Gamma$. 

\item [$\ctx = \FieldAssign{\ctx'}{\f}{\e'}$] \label{lem:context:assL}  
By \cref{lem:inversion}(\ref{lem:inversion:ass}) we have a context $\Gamma'$ such that $\Gamma'\prom\Gamma$, $\Gamma_1\shsum\Gamma_2=\Gamma'$, $\IsWFExp{\Gamma_1}{\ctx'[\e]}{\C}$ and $\IsWFExp{\Gamma_2}{\e'}{\T}$. 
By induction hypothesis we know that for some $\Gamma'',\Delta$,  $\x\not\in\dom{\Gamma}$, $\X$ and $\T'$, $\Gamma_1 = (\X\ctxmul\Delta)\shsum\Gamma''$,
$\IsWFExp{\Gamma''\shsum\VarTypeCoeffect{\x}{\T'}{\X}}{\ctx'[\x]}{\C}$ and
$\IsWFExp{\Delta}{\e}{\T'}$.
We can assume $\x\notin\dom{\Gamma}$ (doing a step of renaming if needed). 
We get 
$\Gamma_1\shsum\Gamma_2= (\X\ctxmul\Delta)\shsum\Gamma''\shsum\Gamma_2=\Gamma'$
and 
$\IsWFExp{(\Gamma_2\shsum\Gamma'')\shsum\VarTypeCoeffect{\x}{\T'}{\X}}{\FieldAssign{\ctx'[\x]}{\f}{\e'}}{\T}$. 
We have two cases:
\begin{itemize}
\item $\T=\C$ 
We know that the last rule applied to derive $\IsWFExp{\Gamma}{\Ctx{\e}}{\T}$ is \refToRule{t-assign}, so $\Gamma = \Gamma_1 \ctxsum \Gamma_2 = \Gamma'$, so we have the thesis
\item $\T=\PT$ 
We have two cases. If $\Gamma = \Gamma'$ we have the same situation as above so we have the thesis. If $\Gamma = \{\link\}\ctxmul\Gamma'$ then we have $\Gamma = \{\link\} \ctxmul ((\X\ctxmul\Delta)\shsum\Gamma'' \ctxsum \Gamma_2) = ((\{\link\} \ctxmul \X)\ctxmul\Delta)\shsum\{\link\} \ctxmul(\Gamma'' \ctxsum \Gamma_2)$. We can apply rule \refToRule{t-prim} to $\IsWFExp{(\Gamma_2\shsum\Gamma'')\shsum\VarTypeCoeffect{\x}{\T'}{\X}}{\FieldAssign{\ctx'[\x]}{\f}{\e'}}{\T}$ to obtain $\IsWFExp{\{\link\}\ctxmul((\Gamma_2\shsum\Gamma'')\shsum\VarTypeCoeffect{\x}{\T'}{\X})}{\FieldAssign{\ctx'[\x]}{\f}{\e'}}{\T}$. We know $\{\link\}\ctxmul((\Gamma_2\shsum\Gamma'')\shsum\VarTypeCoeffect{\x}{\T'}{\X}) = \{\link\}\ctxmul(\Gamma_2\shsum\Gamma'')\shsum\VarTypeCoeffect{\x}{\T'}{\{\link\}\ctxmul\X}$.
\end{itemize}

\item [$\ctx=\Block{\T'}{\x}{\ctx'}{\e'}$] \label{lem:context:bl} 
By \cref{lem:inversion}(\ref{lem:inversion:bl}) we have a context $\Gamma'$ such that $\Gamma'\prom\Gamma$,

$\Gamma'=(\X \shsum \{ \link \}) \ctxmul \Gamma_1 \shsum \Gamma_2$ where $ \link$ is fresh and $\IsWFExp{\Gamma_1}{\ctx'[\e]}{\T}$
and $\IsWFExp{\Gamma_2, \VarTypeCoeffect{\x}{\T'}{\X}}{\e'}{\T'}$.

By induction hypothesis we know that for some $\Gamma'',\Delta$,  $\y\not\in\dom{\Gamma}$, $\Y$ and $\T'$, $\Gamma_1 = (\Y\ctxmul\Delta)\shsum\Gamma''$,
$\IsWFExp{\Gamma''\shsum\VarTypeCoeffect{\y}{\T'}{\Y}}{\ctx'[\y]}{\C}$ and
$\IsWFExp{\Delta}{\e}{\T'}$.
We can assume $\y\notin\dom{\Gamma}$ (doing a step of renaming if needed). 
We get 
$(\X \shsum \{ \link \}) \ctxmul \Gamma_1 \shsum \Gamma_2 = (\X \shsum \{ \link \}) \ctxmul ( (\Y\ctxmul\Delta)\shsum\Gamma'') \shsum \Gamma_2 = (((\X \shsum \{ \link \}) \ctxmul \Y)\ctxmul\Delta)\shsum ((\X \shsum \{ \link \})\ctxmul\Gamma'') \shsum \Gamma_2 = \Gamma'$. By rule \refToRule{t-block} we have 
$\IsWFExp{((\X \shsum \{ \link \})\ctxmul\Gamma'')\shsum \Gamma_2 \shsum\VarTypeCoeffect{\y}{\T'}{((\X \shsum \{ \link \}) \ctxmul \Y)}}{\ctx[\y]}{\T}$. 
We have two cases:
\begin{itemize}
\item $\T=\C$ 
We know that the last rule applied to derive $\IsWFExp{\Gamma}{\Ctx{\e}}{\T}$ is \refToRule{t-block}, so $\Gamma = (\X \shsum \{ \link \}) \ctxmul \Gamma_1 \shsum \Gamma_2 = \Gamma'$, so we have the thesis
\item $\T=\PT$ 
We have two cases. If $\Gamma = \Gamma'$ we have the same situation as above so we have the thesis. If $\Gamma = \{\link'\}\ctxmul\Gamma'$ then we have $\Gamma = \{\link'\} \ctxmul ((\X \shsum \{ \link \}) \ctxmul \Gamma_1 \shsum \Gamma_2) =  (\{\link'\} \ctxmul (\X \shsum \{ \link \}) \ctxmul \Gamma_1) \shsum \{\link'\}\ctxmul\Gamma_2 = (((\{\link'\} \ctxmul (\X \shsum \{ \link \}) \ctxmul \Y)\ctxmul\Delta)\shsum((\{\link'\} \ctxmul (\X \shsum \{ \link \}) \ctxmul\Gamma'') \shsum \{\link'\}\ctxmul\Gamma_2$. We can apply rule \refToRule{t-prim} to $\IsWFExp{((\X \shsum \{ \link \})\ctxmul\Gamma'')\shsum \Gamma_2 \shsum\VarTypeCoeffect{\y}{\T'}{((\X \shsum \{ \link \}) \ctxmul \Y)}}{\ctx[\y]}{\T}$ to obtain $\IsWFExp{\{\link'\}\ctxmul(((\X \shsum \{ \link \})\ctxmul\Gamma'')\shsum \Gamma_2 \shsum\VarTypeCoeffect{\y}{\T'}{((\X \shsum \{ \link \}) \ctxmul \Y)})}{\ctx[\y]}{\T}$. We know $\{\link'\}\ctxmul(((\X \shsum \{ \link \})\ctxmul\Gamma'')\shsum \Gamma_2 \shsum\VarTypeCoeffect{\y}{\T'}{((\X \shsum \{ \link \}) \ctxmul \Y)}) = (\{\link'\}\ctxmul(\X \shsum \{ \link \})\ctxmul\Gamma'')\shsum \{\link'\}\ctxmul\Gamma_2 \shsum\VarTypeCoeffect{\y}{\T'}{(\{\link'\}\ctxmul(\X \shsum \{ \link \}) \ctxmul \Y)}$ By $\IsWFExp{(\{\link'\}\ctxmul(\X \shsum \{ \link \})\ctxmul\Gamma'')\shsum \{\link'\}\ctxmul\Gamma_2 \shsum\VarTypeCoeffect{\y}{\T'}{(\{\link'\}\ctxmul(\X \shsum \{ \link \}) \ctxmul \Y)}}{\ctx[\y]}{\T}$ we obtain the thesis.
\end{itemize}
\end{description}  
\end{proof}
\erb
 
\begin{lemma}\label{lem:ctx-subst}
If $\IsWFExp{\Delta}{\e}{\T}$ and $\IsWFExp{\Gamma\shsum\VarTypeCoeffect{\x}{\T}{\X}}{\Ctx{\x}}{\T'}$, with $\x\notin\dom\Gamma$, then 
$\IsWFExp{\Gamma\shsum\X\ctxmul\Delta}{\Ctx{\e}}{\T'}$. 
\end{lemma}
\begin{proof}
By induction on the derivation of $\IsWFExp{\Gamma\shsum\VarTypeCoeffect{\x}{\T}{\X}}{\Ctx{\x}}{\T'}$. 
\end{proof}


\begin{lemma}\label[lemma]{lem:closureGroup}
Let $\Gamma$ a closed context and $\Delta = \sum_{\x\in\dom{\Delta}}\VarTypeCoeffect{\x}{\T_{\x}}{\X}$ and such that $\dom{\Delta} \subseteq \dom{\Gamma}$ and 
$\links{\Gamma} \cap \links{\Delta} = \{\res\}$.
We have 2 cases:
\begin{itemize}
\item $\res \notin \X$\\
For all $\y \in \dom{\Gamma}$, if exists $\x \in \dom{\Delta}$ such that $\getCoeff{\Gamma}{\y} \cap \getCoeff{\Gamma}{\x} \neq \emptyset$ then
$\getCoeff{\Gamma \shsum \Delta}{\y} = \X \shsum \sum_{\x \in \Delta}\getCoeff{\Gamma}{\x}$, otherwise 
$\getCoeff{\Gamma \shsum \Delta}{\y} = \getCoeff{\Gamma}{\y}$.
\item $\res \in \X$\\
We define $\Y = \getCoeff{\Gamma}{\z}$ if $\res \in \getCoeff{\Gamma}{\z}$ and $\z \in \dom{\Gamma}$.
For all $\y \in \dom{\Gamma}$, if exists $\x \in \dom{\Delta}$ such that $\getCoeff{\Gamma}{\y} \cap \getCoeff{\Gamma}{\x} \neq \emptyset$ or $\res \in \getCoeff{\Gamma}{\y}$ then
$\getCoeff{\Gamma \shsum \Delta}{\y} = \X \shsum \sum_{\x \in \Delta}\getCoeff{\Gamma}{\x} \shsum \Y$, otherwise 
$\getCoeff{\Gamma \shsum \Delta}{\y} = \getCoeff{\Gamma}{\y}$.
\end{itemize}
\end{lemma}

\begin{proof}
\end{proof}

\begin{lemma}\label[lemma]{lem:timesPlus}
Let $\Restr{\Gamma}{\Delta}=\Delta$.
\begin{enumerate}
\item  $\Restr{(\X\ctxmul\Gamma)}{(\X\ctxmul\Delta)}=\X\ctxmul\Delta$

\item  If $\links{\Theta}\cap(\links{\Delta}\cup\links{\Gamma}) = \{\res\}$ or $\links{\Theta}\cap(\links{\Delta}\cup\links{\Gamma}) = \{\res\} \cup \X$ where exists $\x \in \dom{\Delta}$ such that $\getCoeff{\Delta}{\x} = \X$ and $\dom{\Theta}\subseteq\dom{\Delta}$, then $\Restr{(\Gamma\shsum\Theta)}{(\Delta\shsum\Theta)}=\Delta\shsum\Theta$.
\end{enumerate}
\end{lemma}
\begin{proof}
\begin{enumerate}
\item Let $\x\in\dom{\Delta}$ and $\getCoeff{\Gamma}{\x}=\X_{\Gamma}$ and $\getCoeff{\Delta}{\x}=\X_{\Delta}$.
By definition of $\Restr{\Gamma}{\Delta}$ and $\Restr{\Gamma}{\Delta}=\Delta$, 
we have that $\X_{\Gamma}\cap\links{\Delta}=\X_{\Delta}$. 
Let $\Y$ be such that $\getCoeff{\X\ctxmul\Gamma}{\x}=\Y$, then $\Y=\X\MulLnk\X_{\Gamma}$. We want to prove that
$\Y\cap\links{\X\MulLnk{\Delta}}=\X\MulLnk\X_{\Delta}$.
Consider two cases: $\res\not\in\X_{\Gamma}$ and $\res\in\X_{\Gamma}$.\\
In the first case $\Y=\X_{\Gamma}$ and $\Y\cap(\links{\X\MulLnk{\Delta}})=\Y\cap\links{\Delta}=\X_{\Delta}$. Moreover, since 
$\res\not\in\X_{\Delta}$ we have $\X\MulLnk\X_{\Delta}=\X_{\Delta}$. Therefore, $\Y\cap(\links{\X\MulLnk{\Delta}})=\X\MulLnk\X_{\Delta}$.\\ 
In the second case $\Y=\X\cup(\X_{\Gamma}-\{\res\})$ and since $\res\in\X_{\Delta}$ we have  $\X\MulLnk\X_{\Delta}=\X\cup(\X_{\Delta}-\{\res\})$. 
From $\X_{\Gamma}\cap\links{\Delta}=\X_{\Delta}$ we get $(\X_{\Gamma}-\{\res\})\cap\links{\Delta}=\X_{\Delta}-\{\res\}$.
Note that $\X\subseteq\links{\X\MulLnk{\Delta}}$. Therefore 
$\Y\cap(\links{\X\MulLnk{\Delta}})=\X\cup(\X_{\Delta}-\{\res\})$ which proves the result.

\item We know $\Theta = \hat{\cup}_{i=0}^{n}\Theta_i$, where for all $i \in [1...n]$, in $\Theta_i$ all variables have the same coeffect $\X_i$ and, for all $j,k \in [1...n]$ and $j\neq k$, $\links{\Theta_j} \cap \links{\Theta_k} = \{\res\}$.
To obtain the thesis therefore we just need to prove that the property holds summing one $\Theta_i$ at a time, since $\Gamma \shsum \Theta$ and $\Delta \shsum \Theta$ can be obtained summing iteratively the $\Theta_i$s. It suffices to prove the thesis only for the first sum.
By $\Restr{\Gamma}{\Delta} = \Delta$ we have that for all $\x \in \dom{\Delta}$, $\getCoeff{\Gamma}{\x} = \Z \cup \getCoeff{\Delta}{\x}$, where $\Z \cap \links{\Delta} = \emptyset$.
By \cref{lem:closureGroup} we have two cases for $\Gamma \shsum \Theta_i$ and $\Delta \shsum \Theta_i$:
\begin{itemize}
\item $\res \notin \X_i$\\
For all $\y \in \dom{\Gamma}$, if exists $\x \in \dom{\Theta_i}$ such that $\getCoeff{\Gamma}{\y} \cap \getCoeff{\Gamma}{\x} \neq \emptyset$ then
$\getCoeff{\Gamma \shsum \Theta_i}{\y} = \X \shsum \sum_{\x \in \Theta_i}\getCoeff{\Gamma}{\x} = \X \shsum \sum_{\x \in \Theta_i}(\getCoeff{\Delta}{\x}\cup \Z_{\x})$, otherwise 
$\getCoeff{\Gamma \shsum\Theta_i}{\y} = \getCoeff{\Gamma}{\y} = \Z_{\y} \cup \getCoeff{\Delta}{\y}$. 
For all $\y \in \dom{\Delta}$, if exists $\x \in \dom{\Theta_i}$ such that $\getCoeff{\Delta}{\y} \cap \getCoeff{\Delta}{\x} \neq \emptyset$ then
$\getCoeff{\Delta \shsum \Theta_i}{\y} = \X \shsum \sum_{\x \in \Theta_i}\getCoeff{\Delta}{\x}$, otherwise 
$\getCoeff{\Delta \shsum\Theta_i}{\y} = \getCoeff{\Delta}{\y}$. 
By $\Restr{\Gamma}{\Delta} = \Delta$ we have that, for all $\x,\y \in \dom{\Delta}$, $\getCoeff{\Gamma}{\x}\cap\getCoeff{\Gamma}{\y}\neq\emptyset$ if and only if $\getCoeff{\Delta}{\x}\cap\getCoeff{\Delta}{\y}\neq\emptyset$.
By these considerations we derive that, for all $\x \in \dom{\Gamma}$, $
\getCoeff{\Gamma \shsum \Theta_i}{\x} \cap \links{\Delta \shsum \Theta_i} = \getCoeff{\Delta \shsum \Theta_i}$, that is, $\Restr{\Gamma \shsum \Theta_i}{\Delta \shsum \Theta_i} = \Delta \shsum \Theta_i$.
\item $\res \in \X_i$\\
We define $\Y = \getCoeff{\Gamma}{\z}$ if $\res \in \getCoeff{\Gamma}{\z}$ and $\z \in \dom{\Gamma}$.
For all $\y \in \dom{\Gamma}$, if exists $\x \in \dom{\Theta_i}$ such that $\getCoeff{\Gamma}{\y} \cap \getCoeff{\Gamma}{\x} \neq \emptyset$ or $\res \in \getCoeff{\Gamma}{\y}$ then
$\getCoeff{\Gamma \shsum \Theta_i}{\y} = \X \cup \sum_{\x \in \Theta_i}\getCoeff{\Gamma}{\x} \cup \Y= \X \cup \sum_{\x \in \Theta_i}(\getCoeff{\Delta}{\x}\cup \Z_{\x}) \cup \Y$, otherwise 
$\getCoeff{\Gamma \shsum\Theta_i}{\y} = \getCoeff{\Gamma}{\y} = \Z_{\y} \cup \getCoeff{\Delta}{\y}$. 
We define $\Y = \getCoeff{\Delta}{\z}$ if $\res \in \getCoeff{\Delta}{\z}$ and $\z \in \dom{\Delta}$.
For all $\y \in \dom{\Delta}$, if exists $\x \in \dom{\Theta_i}$ such that $\getCoeff{\Delta}{\y} \cap \getCoeff{\Delta}{\x} \neq \emptyset$ or $\res \in \getCoeff{\Theta_i}{\y}$ then
$\getCoeff{\Delta \shsum \Theta_i}{\y} = \X \shsum \sum_{\x \in \Theta_i}\getCoeff{\Delta}{\x} \shsum \Y$, otherwise 
$\getCoeff{\Gamma \shsum\Theta_i}{\y} = \getCoeff{\Gamma}{\y}$. By $\Restr{\Gamma}{\Delta} = \Delta$ we have that, for all $\x,\y \in \dom{\Delta}$, $\getCoeff{\Gamma}{\x}\cap\getCoeff{\Gamma}{\y}\neq\emptyset$ if and only if $\getCoeff{\Delta}{\x}\cap\getCoeff{\Delta}{\y}\neq\emptyset$.
By these considerations we derive that, for all $\x \in \dom{\Gamma}$, $
\getCoeff{\Gamma \shsum \Theta_i}{\x} \cap \links{\Delta \shsum \Theta_i} = \getCoeff{\Delta \shsum \Theta_i}$, that is, $\Restr{\Gamma \shsum \Theta_i}{\Delta \shsum \Theta_i} = \Delta \shsum \Theta_i$. 
\end{itemize}

\end{enumerate}
\end{proof}

\begin{lemma}[Substitution]\label[lemma]{lem:substitutionLemma}
If  $\IsWFExp{\Delta}{\e'}{\T'}$ and $\IsWFExp{\Gamma,\VarTypeCoeffect{\x}{\T'}{\X}}{\e}{\T}$ and $\Delta'$ is $\Delta$ with fresh renamed links,
 then $\IsWFExp{\Gamma'}{\Subst{\e}{\e'}{\x}}{\T}$ with $\Gamma'$ such that $\Gamma' \shord \X\shmul\Delta'\shsum\Gamma$. 
\end{lemma} 

\begin{proof}
By induction on the derivation of $\IsWFExp{\Gamma, \VarTypeCoeffect{\x}{\T'}{\X}}{\e}{\T}$
\begin{description}
\item  [\refToRule{t-var}] 
We know that $\Gamma = \emptyset$ and $\X = \{\res\}$, hence the thesis follows from the hypothesis $\IsWFExp{\Delta}{\e'}{\T'}$, as 
$\emptyset\shsum \{\res\}\ctxmul\Delta = \Delta$ and $\Subst{\x}{\e'}{\x} = \e'$.

\item [\refToRule{t-invk}] We know that $\e = \MethCallTuple{\e_0}{\m}{\e}{n}$, $\Gamma \shord \shsum_{0\leq i\leq n} \X_{i} \ctxmul \Gamma_i$ and 
$\X \supseteq \cup_{0\leq i\leq n}\X_{i}\MulLnk\X_i'$, where $\IsWFExp{\Gamma_i,\VarTypeCoeffect{\x}{\T'}{\X_i'}}{\e_i}{\C_i}$ ($0\leq i\leq n$) and $\mtype{\C_0}{\m}=\funType{\X_{0},\coeffectType{\T_1}
{\X_1} \ldots \coeffectType{\T_n}{\X_n}}{\T}$. By induction hypothesis we obtain $\IsWFExp{\Gamma_i \shsum (\X_i' \ctxmul \Delta')}{\Subst{\e_i}{\e}
{\x}}{\T_i}$ ($0\leq i\leq n$). Applying rule \refToRule{t-invk} we obtain\\
\centerline{
$
\IsWFExp{\shsum_{0\leq i\leq n} \X_{i}\ctxmul(\Gamma_i\shsum\X'_i\ctxmul\Delta')}{\Subst{\MethCallTuple{\e_0}{\m}{\e}{n}}{\e}{\x}}{\T}
$}
Moreover,  $\shsum_{0\leq i\leq n} \X_{i}\ctxmul(\Gamma_i\shsum\X'_i\ctxmul\Delta')=\shsum_{0\leq i\leq n}(\X_i\ctxmul\Gamma_i)\shsum(\X_i\MulLnk\X_i')\ctxmul\Delta'$ and \\
$\shsum_{0\leq i\leq n}(\X_i\MulLnk\X_i')\ctxmul\Delta'=(\cup_{0\leq i\leq n} \X_i\MulLnk\X_i'  )\ctxmul\Delta'$.

We want to show that $\shsum_{0\leq i\leq n}(\X_i\ctxmul\Gamma_i)  \shsum (\cup_{0\leq i\leq n}\X_i\MulLnk\X_i' \ctxmul \Delta') \shord \X\shmul\Delta'\shsum\Gamma$.
Knowing that $\Gamma \shord \shsum_{0\leq i\leq n} \X_{i} \ctxmul \Gamma_i$ and $\X \supseteq \cup_{0\leq i\leq n}\X_i\MulLnk\X_i'$ we derive 
$\shsum_{0\leq i\leq n}(\X_i\ctxmul\Gamma_i)  \shsum (\cup_{0\leq i\leq n}\X_i\MulLnk\X_i' \ctxmul \Delta') \shord \X\shmul\Delta'\shsum\Gamma $.

\item [\refToRule{t-new}]We know that $\e = \ConstrCall{\C}{\e_1, \ldots, \e_n}$, $\Gamma \shord \shsum_{1\leq i\leq n} \Gamma_i$ and $\X \supseteq \cup_{1\leq i\leq n} \X_i$ and 
$\IsWFExp{\Gamma_i,\VarTypeCoeffect{\x}{\T'}{\X_i}}{\e_i}{\T_i}$ ($1\leq i\leq n$). 
By inductive hypothesis $\IsWFExp{\Gamma_i \shsum (\X_i \ctxmul \Delta')}{\Subst{\e_i}{\e'}{\x}}{\T_i}$, ($1\leq i\leq n$). 
Applying rule \refToRule{t-new} we obtain
$\IsWFExp{\shsum_{1\leq i\leq n}\Gamma_i \shsum (\X_i \ctxmul \Delta')}{\Subst{ \ConstrCall{\C}{\e_1, \ldots, \e_n}}{\e'}{\x}}{\T_1}$. \\

We know $\shsum_{1\leq i\leq n}\Gamma_i \shsum (\X_i \ctxmul \Delta')= \shsum_{1\leq i\leq n}\Gamma_i \shsum\shsum_{1\leq i\leq n}(\X_i \ctxmul \Delta')= \shsum_{1\leq i\leq n}\Gamma_i \shsum(\cup_{1\leq i\leq n}\X_i \ctxmul \Delta')$.

We want to show that $\shsum_{1\leq i\leq n}\Gamma_i \shsum(\cup_{1\leq i\leq n}\X_i \ctxmul \Delta') \shord \X\shmul\Delta'\shsum\Gamma$.
By $\X \supseteq \cup_{1\leq i\leq n}\X_i$ and $\Gamma \shord \shsum_{1\leq i\leq n} \Gamma_i$ we derive $ \X\shmul\Delta'\shsum\Gamma \shord \shsum_{1\leq i\leq n}\Gamma_i \shsum(\cup_{1\leq i\leq n}\X_i \ctxmul \Delta')$.

\item [\refToRule{t-block}]
By assumptions we know that 
$\e = {\Block{\T}{\x}{\e_1}{\e_2}}$, 
$\IsWFExp{\Gamma_1,\VarTypeCoeffect{\x}{\T'}{\X_1}}{\e_1}{\T''}$ and 
$\IsWFExp{\Gamma_2,\VarTypeCoeffect{\x}{\T'}{\X_2},\VarTypeCoeffect{\y}{\T''}{\Y}}{\e_2}{\T}$ and 
\[
\Gamma,\VarTypeCoeffect{\x}{\T'}{\X} = (\Y\cup\{\link\})\ctxmul(\Gamma_1,\VarTypeCoeffect{\x}{\T'}{\X_1}) \shsum (\Gamma_2,\VarTypeCoeffect{\x}{\T'}{\X_2}) 
\]
where $\link$ is fresh. 
Then, in particular, we have 
$\Gamma \shord (\Y \cup \{ \link \}) \ctxmul \Gamma_1 \shsum \Gamma_2$ and 
$\X \supseteq (\Y\cup\{\link\})\MulLnk\X_1\cup\X_2$. 
We want to prove that $\IsWFExp{\Gamma'}{\Subst{\Block{\T}{\x}{\e_1}{\e_2}}{\e'}{\x}}{\T}$ where 
\[ 
\Gamma'= ((\Y\cup\{\link\})\MulLnk\X_1\cup\X_2) \ctxmul \Delta' \shsum (\Y \cup \{ \link \}) \ctxmul \Gamma_1 \shsum \Gamma_2
\] 
then the thesis will follow by subsumption.  
By induction hypothesis, we get 
$\IsWFExp{\Gamma_1\shsum\X_1\ctxmul\Delta'}{\Subst{\e_1}{\e'}{\x}}{\T''}$ and 
$\IsWFExp{(\Gamma_2, \VarTypeCoeffect{\y}{\T''}{\Y})\shsum\X_2\ctxmul\Delta'}{\Subst{\e_2}{\e'}{\x}}{\T}$.  
By definition of $\shsum$ we have 
$(\Gamma_2,\VarTypeCoeffect{\y}{\T''}{\Y})\shsum\X_2\ctxmul\Delta' = \closure{(\Gamma_2\shsum\X_2\ctxmul\Delta',\VarTypeCoeffect{\y}{\T''}{\Y})} = \Theta,\VarTypeCoeffect{\y}{\T''}{\Y'}$ for some $\Theta$ and $\Y'$. 
By \cref{lem:links-closure} we know that $\Y'$ is the union of all coeffects in $\Gamma_2 \shsum \X_2\ctxmul\Delta'$ that are not  disjoint from $\Y$. 

Then, applying rule \refToRule{t-block}, we derive $\IsWFExp{\Gamma''}{\Subst{\Block{\T''}{\y}{\e_1}{\e_2}}{\e}{\x}}{\T}$ with 
\[
\Gamma''={(\Y'\cup\{\link\})\ctxmul(\Gamma_1\shsum\X_1\ctxmul\Delta')\shsum\Gamma_2\shsum\X_2\ctxmul\Delta'}
\] 
Since $\Y\subseteq\Y'$,
we have 
\[
\Gamma''=\Y'\ctxmul(\Gamma_1\shsum\X_1\ctxmul\Delta')\shsum (\Y\cup\{\link\})\ctxmul(\Gamma_1\shsum\X_1\ctxmul\Delta')\shsum\Gamma_2\shsum\X_2\ctxmul\Delta' 
\] 
Since the links in $\Delta'$ are fresh, we have that $\X_2\ctxmul\Delta' = \X_2\MulLnk\Delta'$, that is, 
$\X_2 \cap \links{\Delta'} \subseteq \{\res\}$ so $\X_2 \ctxmul \Delta'$ is equal to $\Delta'$ except that in coeffects $\res$ is replaced with links in $\X_2$. 
Moreover, by \cref{lem:closed-ctx}  $\X_2=\Y$ or $\X_2\cap\Y=\emptyset$. 
Therefore, for all $\z \in \dom{\X_2\ctxmul\Delta'}$, either 
$\Y \subseteq \getCoeff{\X_2\ctxmul\Delta'}{\z}$ or $\Y \cap \getCoeff{\X_2\ctxmul\Delta'}{\z} = \emptyset$. 
Applying \cref{lem:links-closure} we get 
$\Theta = \Gamma_2\shsum\X_2\ctxmul\Delta'$ and $\Y' = \getCoeff{\Theta}{\z}$ if $\Y \subseteq \getCoeff{\Theta}{\z}$ for all $\z \in \dom{\Theta}$. 
For all $\z \in \dom{\Gamma_1\shsum\X_1\ctxmul\Delta'}$ such that $\res \in \getCoeff{\Gamma_1\shsum\X_1\ctxmul\Delta'}{\z}$, we have that
$\Y \subseteq \getCoeff{(\Y\cup\{\link\})\ctxmul(\Gamma_1\shsum\X_1\ctxmul\Delta')}{\z}$. 
Since 
$(\Y\cup\{\link\})\ctxmul(\Gamma_1\shsum\X_1\ctxmul\Delta')\shsum\Gamma_2\shsum\X_2\ctxmul\Delta'$ is closed and
$\Y'$ is the union of all coeffects in $\Gamma_2 \shsum \X_2\ctxmul\Delta'$ that are not  disjoint from $\Y$ 
we get that
$\Y \subseteq \getCoeff{(\Y\cup\{\link\})\ctxmul(\Gamma_1\shsum\X_1\ctxmul\Delta')}{\z}$ implies $\Y' \subseteq \getCoeff{(\Y\cup\{\link\})\ctxmul(\Gamma_1\shsum\X_1\ctxmul\Delta')\shsum\Gamma_2\shsum\X_2\ctxmul\Delta'}{\z}$. 

On the other hand, if $\res \not\in \getCoeff{\Gamma_1\shsum\X_1\ctxmul\Delta'}{\z}$, then
$\getCoeff{\Y'\hatshmul(\Gamma_1\shsum\X_1\ctxmul\Delta')}{\z}=\getCoeff{\Gamma_1\shsum\X_1\ctxmul\Delta'}{\z}$.
Therefore 
\[
\Y'\hatshmul(\Gamma_1\shsum\X_1\ctxmul\Delta') \shord (\Y\cup\{\link\})\ctxmul(\Gamma_1\shsum\X_1\ctxmul\Delta')\shsum\Gamma_2\shsum\X_2\ctxmul\Delta'
\] 
and so $\Gamma''=(\Y\cup\{\link\})\ctxmul(\Gamma_1\shsum\X_1\ctxmul\Delta')\shsum\Gamma_2\shsum\X_2\ctxmul\Delta'$. 

Finally 
\[ 
\begin{array}{lcl}
\Gamma'' & = & (\Y\cup\{\link\})\ctxmul(\Gamma_1\shsum\X_1\ctxmul\Delta')\shsum\Gamma_2\shsum\X_2\ctxmul\Delta'\\
 & = &((\Y\cup\{\link\})\ctxmul\X_1)\ctxmul\Delta'\shsum\X_2\ctxmul\Delta' \shsum (\Y\cup\{\link\})\ctxmul\Gamma_1\shsum\Gamma_2\\
 & = &((\Y\cup\{\link\})\MulLnk\X_1\cup\X_2) \ctxmul \Delta' \shsum (\Y \cup \{ \link \}) \ctxmul \Gamma_1 \shsum \Gamma_2\\
\end{array}
\] 

We want to show that $((\Y\cup\{\link\})\MulLnk\X_1\cup\X_2) \ctxmul \Delta' \shsum (\Y \cup \{ \link \}) \ctxmul \Gamma_1 \shsum \Gamma_2 \shord \X \ctxmul \Delta' \shsum \Gamma$.
Since $\X \supseteq (\Y\cup\{\link\})\MulLnk\X_1\cup\X_2$ and $\Gamma \shord (\Y \cup \{ \link \}) \ctxmul \Gamma_1 \shsum \Gamma_2$
 we derive 
$((\Y\cup\{\link\})\MulLnk\X_1\cup\X_2) \ctxmul \Delta' \shsum (\Y \cup \{ \link \}) \ctxmul \Gamma_1 \shsum \Gamma_2 \shord \X \ctxmul \Delta' \shsum \Gamma$.

\end{description}
\end{proof}

\begin{proof}[Proof of \cref{thm:subj-red}]
\brb
If $\IsWFConf{\Gamma'}{\e}{\mem}{\T}$, we have 
$\Gamma' = \Gamma_1'\shsum\Gamma_2$, 
$\IsWFExp{\Gamma_1'}{\e}{\T}$ and $\IsWFMem{\Gamma_2}{\mem}$.
We know that 
$\Gamma_2 = \Gamma_\mem \shsum \Theta$, where 
$\Gamma_\mem = \VarTypeCoeffect{\x_1}{\T_1}{\{\link_1\}},\ldots,\VarTypeCoeffect{\x_n}{\T_n}{\{\link_n\}}$, 
$\Theta = \sum_{i = 1}^n \link_i \ctxmul \Theta_i$, 
$\IsWFExp{\Theta_i}{\mem(\x_i)}{\T_u}$ and 
$\link_1,\ldots,\link_n$ are fresh links. 
We also know that 
$\Gamma_2 = \Gamma_\mem \shsum \Theta$, where 
$\Gamma_\mem = \VarTypeCoeffect{\x_1}{\T_1}{\{\link_1\}},\ldots,\VarTypeCoeffect{\x_n}{\T_n}{\{\link_n\}}$, 
$\Theta = \sum_{i = 1}^n \link_i \ctxmul \Theta_i$, 
$\IsWFExp{\Theta_i}{\mem(\x_i)}{\T_u}$ and 
$\link_1,\ldots,\link_n$ are fresh links. 
We also know that by \cref{lem:nonStruct} $\IsWFExp{\Gamma_1}{\e}{\T}$ with $\Gamma_1\prom\Gamma_1'$. In the proof below we prove the theorem with as hypothesis $\IsWFConf{\Gamma}{\e}{\mem}{\T}$ where $\Gamma = \Gamma_1\shsum\Gamma_2$ obtaining that $\IsWFConf{\Delta}{\e'}{\amem}{\T}$ and $\Restr{\Gamma\shsum\Delta}{\Gamma}= \Gamma$. We also know that $\Delta = \Gamma_3 \shsum \Gamma_4$ and $\IsWFExp{\Gamma_3}{\e}{\T}$ and $\IsWFMem{\Gamma_4}{\amem}$. If $\Gamma_1=\Gamma_1'$ we have also the thesis. If $\Gamma_1'=\{\link\}\ctxmul\Gamma_1$ then we can apply the same non syntx directed rules applied to $\IsWFExp{\Gamma_1}{\e}{\T}$
obtaining $\IsWFExp{\{\link\}\shmul\Gamma_3}{\e'}{\T}$ and so, since $\res \notin \links{\Gamma_4}$, $\IsWFConf{\{\link\}\shmul\Delta}{\e}{\mem}{\T}$. By \cref{lem:timesPlus} we have that $\Restr{\{\link\}\ctxmul(\Gamma\shsum\Delta)}{\{\link\}\ctxmul\Gamma}=\{\link\}\ctxmul\Gamma$, that is, the thesis.
The proof is by induction on the reduction relation.
\erb
\begin{description}

\item[\refToRule{field-access}] 
	We know that $\reduce{\ExpMem{\FieldAccess{\x}{\f_k}}{\mem}}{\ExpMem{\val_k}{\mem}}$, hence 
	$\e = \FieldAccess{\x}{\f_k}$, $\e' = \val_k$ and $\amem = \mem$.
	We know that $\mem(\x) = \ConstrCall{\C}{\val_1, \ldots, \val_k, \ldots, \val_{m}}$ with k $\in 1..m$.
	Since $\x \in \dom{\mem}$ we know exists an $h \in 1..n$ such that $\x = \x_h$ and so $\IsWFExp{\Theta_h}{\mem(\x)}{\T_h}$. By \cref{lem:inversion}(\ref{lem:inversion:new}) we know that $\Theta_h = \sum_{i=1}^{m}\Theta'_i$ such that $\IsWFExp{\Theta'_i}{\val_i}{\T_i}$. By this consideration we derive that $\IsWFExp{\Theta'_k}{\val_k}{\T_k}$.
	We know that $\Gamma \shord\{\link_h\}\ctxmul\Theta_h \shord \{\link_h\}\MulLnk\Theta_h$. For all $\y \in \dom{\Theta_h}$ we know $\getCoeff{\{\link_h\}\MulLnk\Theta_h}{\y}$ = $(\getCoeff{\Theta_h}{\y}\setminus \{\res\}) \cup \{\link_h\}$ if $\res \in \getCoeff{\Theta_h}{\y}$, $\getCoeff{\{\link_h\}\MulLnk\Theta_h}{\y} = \getCoeff{\Theta_h}{\y}$ otherwise. 
	We know that $\{\res,\link_h\} \subseteq \getCoeff{\Gamma}{\x}$, so, since $\Gamma$ is closed and $\Gamma \shord \{\link_h\}\MulLnk\Theta_h$, for all $\y \in \dom{\{\link_h\}\MulLnk\Theta_h}$, $\link_h \in \getCoeff{\{\link_h\}\MulLnk\Theta_h}{\y}$ implies $\res \in \getCoeff{\Gamma}{\y}$ and so $\getCoeff{\Gamma}{\y} \supseteq (\getCoeff{\Theta_h}{\y} \setminus \{\res\}) \cup \{\link_h\} \cup \{\res\} \supseteq \getCoeff{\Theta_h}{\y}$. By these considerations we derive that $\Gamma \shord \Theta_h$.
	We can apply rule \refToRule{t-conf} obtaining $\IsWFConf{\Gamma_2 \shsum \Theta'_{k}}{\val_k}{\amem}{\T_k}$.
	Since $\Gamma = \Gamma \shsum \Gamma_2 \shsum \Theta'_{k}$ we have $\Restr{(\Gamma \shsum \Gamma_2 \shsum \Theta'_{k})}{\Gamma} = \Restr{\Gamma}{\Gamma} = \Gamma$ we have the thesis.

\item[\refToRule{field-assign}] 
	We know that  $\reduce{\ExpMem{\FieldAssign{\x}{\f_k}{\val}}{\mem}}{\ExpMem{\val}{\UpdateMem{\mem}{\x}{k}{\val}}}$, hence 
	$\e = \FieldAssign{\x}{\f_k}{\val}$, $\e' = \val$ and $\amem = \UpdateMem{\mem}{\x}{k}{\val}$. 
	By \cref{lem:inversion}(\ref{lem:inversion:ass}), we get 
	$\Gamma_1 = \Delta_1\shsum\Delta_2$, 
	$\IsWFExp{\Delta_1}{\x}{\C}$, 
	$\IsWFExp{\Delta_2}{\val}{\T}$, and 
	$\fields{\C} = \Field{S_1}{\f_1},\ldots, \Field{S_m}{\f_m}$ with $k \in 1..m$. 
	Again by \cref{lem:inversion}(\ref{lem:inversion:x}),  
	we get $\Delta_1 = \VarTypeCoeffect{\x}{\C}{\{\res\}}$.  
	We know that $\mem(\z) = \amem(\z)$ for all $\z \in \dom{\mem}\setminus\{\x\}$ and 
	$\mem(\x) = \ConstrCall{\C}{\val_1, \ldots, \val_m}$ and $\amem(\x) = \ConstrCall{\C}{\val_1, \ldots, \val_{k-1},\val,\val_{k+1}, \ldots, \val_m}$. 
	Since $\x\in\dom{\mem}$, there is $h \in 1..n$ such that $\x = \x_h$ and so 
	$\IsWFExp{\Theta_h}{\mem(\x)}{\T_h}$ holds, with $\T_h = \C$. 
	Applying \cref{lem:inversion}(\ref{lem:inversion:new}), we derive that $\Theta_h = \Sigma_{i=1}^{m} \Theta'_i$ and $\IsWFExp{\Theta'_i}{\val_i}{S_i}$, for all $i \in 1..m$.
	By these considerations and applying rule \refToRule{t-new} we obtain $\IsWFExp{\Sigma_{i=1}^{k-1} \Theta'_i \shsum \Delta_2 \shsum \Sigma_{i=k+1}^{m} \Theta'_i}{\ConstrCall{\C}{\val_1, \ldots, \val_{k-1},\val,\val_{k+1}, \ldots, \val_m}}{\C}$. 
	
	Applying rule \refToRule{t-mem} we can type memory $\amem$, deriving $\IsWFMem{\Gamma_{\mem} \shsum (\sum_{i = 1}^{h-1} \{\link_i\} \ctxmul \Theta_i) \shsum (\{\link_h\}\ctxmul(\Sigma_{i=1}^{k-1} \Theta'_i \shsum \Delta_2 \shsum \Sigma_{i=k+1}^{m} \Theta'_i)) \shsum (\sum_{i = h+1}^{n} \{\link_i\} \ctxmul \Theta_i)}
	{\amem}$ with $\getCoeff{\Gamma_{\mem}}{\x} = \{\link_{\x}\}$.
	We know $\Gamma_{\mem} \shsum (\sum_{i = 1}^{h-1} \{\link_i\} \ctxmul \Theta_i) \shsum (\{\link_h\}\ctxmul(\Sigma_{i=1}^{k-1} \Theta'_i \shsum \Delta_2 \shsum \Sigma_{i=k+1}^{m} \Theta'_i)) \shsum (\sum_{i = h+1}^{n} \{\link_i\} \ctxmul \Theta_i) = \Gamma_{\mem} \shsum (\sum_{i = 1}^{h-1} \{\link_i\} \ctxmul \Theta_i) \shsum \{\link_h\}\ctxmul\Sigma_{i=1}^{k-1} \Theta'_i \shsum \{\link_h\}\shmul\Delta_2 \shsum \{\link_h\}\shmul\Sigma_{i=k+1}^{m} \Theta'_i \shsum (\sum_{i = h+1}^{n} \{\link_i\} \ctxmul \Theta_i)$.
	We have two cases for all $\y \in \dom{\Delta_2}$:
	\begin{itemize} 
	\item $\res \in \getCoeff{\Delta_2}{\y}$\\
	We have $\getCoeff{\{\link_h\}\MulLnk\Delta_2}{\y} = (\getCoeff{\Delta_2}{\y}\setminus\{\res\})\cup\{\link_h\}$. We have $\{\res,\link_h\} \subseteq \getCoeff{\Gamma}{\x}$, and, since $\Gamma \shord \Delta_2$ we know $\getCoeff{\Gamma}{\y} \supseteq \getCoeff{\Delta_2}{\y}$ and in particular $\res \in \getCoeff{\Gamma}{\y}$. By the fact that $\Gamma$ is closed we know that $\link_h \in \getCoeff{\Gamma}{\y}$. We can also conclude that $\getCoeff{\Gamma}{\y} \supseteq \getCoeff{\Delta_2}{\y} \cup \{\link_h\} \supseteq  (\getCoeff{\Delta_2}{\y}\setminus\{\res\})\cup\{\link_h\}$.
	\item $\res \notin \getCoeff{\Delta_2}{\y}$\\
	We have $\getCoeff{\{\link_h\}\MulLnk\Delta_2}{\y} = \getCoeff{\Delta_2}{\y}$.
	Since $\Gamma \shord \Delta_2$ we derive  $\getCoeff{\Gamma}{\y} \supseteq \getCoeff{\Delta_2}{\y}$.
	\end{itemize}
	By these two considerations we have that $ \{\link_h\}\MulLnk\Delta_2\shord \Gamma$, so $\{\link_h\}\shmul\Delta_2\shord \Gamma$.

	Since $\Gamma_{\mem} \shord \Gamma$, $(\sum_{i = 1}^{h-1} \{\link_i\} \ctxmul \Theta_i) \shord \Gamma$, $\{\link_h\}\ctxmul\Sigma_{i=1}^{k-1} \Theta'_i\shord \Gamma$, $\{\link_h\}\shmul\Delta_2\shord \Gamma$,$\{\link_h\}\shmul\Sigma_{i=k+1}^{m} \Theta'_i\shord \Gamma$, $(\sum_{i = h+1}^{n} \{\link_i\} \ctxmul \Theta_i)\shord \Gamma$. We derive $\Gamma_{\mem} \shsum (\sum_{i = 1}^{h-1} \{\link_i\} \ctxmul \Theta_i) \shsum \{\link_h\}\ctxmul\Sigma_{i=1}^{k-1} \Theta'_i \shsum \{\link_h\}\shmul\Delta_2 \shsum \{\link_h\}\shmul\Sigma_{i=k+1}^{m} \Theta'_i \shsum (\sum_{i = h+1}^{n} \{\link_i\} \ctxmul \Theta_i)\shord\Gamma$. Applying rule \refToRule{t-conf}, we obtain $\Delta = \Gamma_{\mem} \shsum (\sum_{i = 1}^{h-1} \{\link_i\} \ctxmul \Theta_i) \shsum \{\link_h\}\ctxmul\Sigma_{i=1}^{k-1} \Theta'_i \shsum \{\link_h\}\shmul\Delta_2 \shsum \{\link_h\}\shmul\Sigma_{i=k+1}^{m} \Theta'_i \shsum (\sum_{i = h+1}^{n} \{\link_i\} \ctxmul \Theta_i)\shsum\Delta_2$. We know that $\Gamma\shsum \Delta = \Gamma$.
	By the fact that $\Restr{\Delta\shsum\Gamma}{\Gamma} = \Restr{\Gamma}{\Gamma} = \Gamma$ we obtain the thesis.

\item[\refToRule{new}] 
	We know that $\reduce{\ExpMem{\ConstrCallTuple{\C}{\val}{n}}{\mem}}{\ExpMem{\loc}{\Subst{\mem}{\ConstrCallTuple{\C}{\val}{n}}{\loc}}}$ with $\loc\not\in\dom{\mem}$ hence 
	$\e = \ConstrCallTuple{\C}{\val}{n}$, $\e' = \x$ and $\amem = \Subst{\mem}{\ConstrCallTuple{\C}{\val}{n}}{\loc}$. 
	Since we have $\mem(\y) = \amem(\y)$ for all $\y \in \dom{\mem}$ and $\dom{\amem} = \dom{\mem} \cup \{\x\}$ and $\amem(\x) = \ConstrCallTuple{\C}{\val}{n}$ we can apply rule \refToRule{t-mem}  deriving $\IsWFMem{\Gamma_{\mem} \shsum \VarTypeCoeffect{\x}{\C}{\link_{m+1}} \shsum \Theta \shsum (\link_{m+1} \ctxmul \Gamma_1)}{\amem}$ , where $\IsWFExp{\Gamma_1}{\amem(\x)}{\C}$, $\link_{m+1}$ fresh and $\link_{m+1}\notin\links{\Gamma}$. By rule \refToRule{t-var} we know $\IsWFExp{\VarTypeCoeffect{\x}{\C}{\{\res\}}}{\x}{\C}$. By rule \refToRule{t-conf} we derive $\Delta = \Gamma_{\mem} \shsum \VarTypeCoeffect{\x}{\C}{\{\link_{m+1}\}} \shsum \Theta \shsum (\{\link_{m+1}\} \ctxmul \Gamma_1) \shsum 
\VarTypeCoeffect{\x}{\C}{\{\res\}}$ and
	$\IsWFConf{\Gamma_{\mem} \shsum \VarTypeCoeffect{\x}{\C}{\{\link_{m+1}\}} \shsum \Theta \shsum (\{\link_{m+1}\} \ctxmul \Gamma_1) \shsum 
\VarTypeCoeffect{\x}{\C}{\{\res\}}}{\x}{\amem}{\C}$. We know $\Gamma \shsum \Delta = \Gamma_{\mem} \shsum \VarTypeCoeffect{\x}{\C}{\{\link_{m+1}\}} \shsum \Theta \shsum (\{\link_{m+1}\} \ctxmul \Gamma_1) \shsum \VarTypeCoeffect{\x}{\C}{\{\res\}} \shsum \Gamma = \Gamma \shsum \VarTypeCoeffect{\x}{\C}{\{\res,\link_{m+1}\}} \shsum (\{\link_{m+1}\} \ctxmul \Gamma_1)$.
	We know that $\link_{m+1} \notin \getCoeff{\Gamma}{\y}$ for all $\y \in \dom{\Gamma}$ and we know that, since $\Gamma$ is closed, $\res \in 
	\getCoeff{\Gamma}{\y}$ and $\res \in \getCoeff{\Gamma}{\z}$ implies  $\getCoeff{\Gamma}{\y} = \getCoeff{\Gamma}{\z}$ for all $\y,\z \in 
	\dom{\Gamma}$. We also know that $\x \notin \dom{\Gamma}$, so $\Gamma \shsum \VarTypeCoeffect{\x}{\C}{\{\link_{m+1},\res\}}$ = $
	\closure{(\Gamma,\VarTypeCoeffect{\x}{\C}{\{\link_{m+1},\res\}})} = \Gamma',\VarTypeCoeffect{\x}{\C}{\X}$. By \cref{lem:links-closure} and by the observations above we have 2 cases for all $\y \in \dom{\Gamma}$:
	\begin{itemize}
	\item $\res \notin \getCoeff{\Gamma}{\y}$ implies $\getCoeff{\Gamma'}{\y} = \getCoeff{\Gamma}{\y}$
	\item $\res \in \getCoeff{\Gamma}{\y}$ implies $\getCoeff{\Gamma'}{\y} = \getCoeff{\Gamma}{\y} \cup \{\link_{m+1},\res\} = \getCoeff{\Gamma}{\y} \cup \{\link_{m+1}\}$
	\end{itemize}
	We know $\Gamma_1 \shord \Gamma \shord \Gamma'$ and $(\Gamma',\VarTypeCoeffect{\x}{\C}{\X}) \shsum (\{\link_{m+1}\} \ctxmul \Gamma_1) = (\Gamma',\VarTypeCoeffect{\x}{\C}{\X}) \shsum (\{\link_{m+1}\} \MulLnk \Gamma_1)$.
	
	If $\res \in \getCoeff{\Gamma_1}{\y}$ then $\getCoeff{\{\link_{m+1}\}\MulLnk\Gamma_1}{\y} = (\getCoeff{\Gamma_1}{\y}\setminus \{\res\}) \cup \{\link_{m+1}\}$ and, since $\res \in \getCoeff{\Gamma}{\y}$, $\getCoeff{\Gamma'}{\y} = \getCoeff{\Gamma}{\y} \cup \{\link_{m+1}\} \supseteq \getCoeff{\Gamma_1}{\y} \cup \{\link_{m+1}\}$ for all $\y \in \dom{\Gamma_1}$. 
	if $\res \notin \getCoeff{\Gamma_1}{\y}$ then $\getCoeff{\{\link_{m+1}\}\MulLnk\Gamma_1}{\y} = \getCoeff{\Gamma_1}{\y} \subseteq \getCoeff{\Gamma'}{\y}$ for all $\y \in \dom{\Gamma_1}$.
	By these observations we can conclude that $(\Gamma',\VarTypeCoeffect{\x}{\C}{\X}) \shsum (\{\link_{m+1}\MulLnk\Gamma_1\}) = \Gamma',\VarTypeCoeffect{\x}{\C}{\X} = \Gamma \shsum \Delta$.
	By the fact that $\link_{m+1} \notin \links{\Gamma}$ and $\getCoeff{\Gamma \shsum \Delta}{\z} = \getCoeff{\Gamma}{\z}$ or $\getCoeff{\Gamma \shsum \Delta}{\z} = \getCoeff{\Gamma}{\z} \cup \{\link_{m+1}\}$ for all $\z \in \dom{\Gamma}$ we derive that $\Restr{(\Gamma \shsum \Delta)}{\Gamma} = \Gamma$, that is, the thesis.

\item[\refToRule{ctx}] 
	We have $\e=\Ctx{\e_1}$ and  $\e'=\Ctx{\e'_1}$ and  ${\reduce{\ExpMem{\e_1}{\mem}}{\ExpMem{\e_1'}{\amem}}}$. 
	By \cref{lem:InvCtx}, $\Gamma_1 = \Delta_1\shsum\X\ctxmul\Delta_2$, 
	$\IsWFExp{\Delta_1\shsum\VarTypeCoeffect{\x}{\T'}{\X}}{\Ctx{\x}}{\T}$ and $\IsWFExp{\Delta_2}{\e_1}{\T'}$ and we can impose that if $\link\in\getCoeff{\Delta_1}{\x}$ and  $\link\in\X$ then $\getCoeff{\Delta_1}{\x}= \X$.
	We get $\Gamma = \Gamma_1\shsum\Gamma_2 = \Delta_1\shsum\X\ctxmul\Delta_2\shsum\Gamma_2$ and, 
	since $\res\notin\getCoeff{\Gamma_2}{\y}$ for every $\y\in\dom{\Gamma_2}$ by rule \refToRule{t-mem}
	we have $\Gamma_2 = \X\ctxmul\Gamma_2$, hence 
	$\Gamma = \Delta_1 \shsum\X\ctxmul(\Delta_2\shsum\Gamma_2)$. 
	We set $\Gamma' = \Delta_2\shsum\Gamma_2$. 
	By induction hypothesis, we get $\IsWFConf{\Delta'}{\e_1'}{\amem}{\T'}$ with $\Restr{\Delta'\shsum\Gamma'}{\Gamma'} = \Gamma'$. 
	By rule \refToRule{t-conf}, we know that $\Delta' = \Delta'_1\shsum\Delta'_2$ with 
	$\IsWFExp{\Delta'_1}{\e_1'}{\T'}$ and $\IsWFMem{\Delta'_2}{\amem}$. 
	By \cref{lem:ctx-subst}, we get $\IsWFExp{\Delta_1\shsum\X\ctxmul\Delta'_1}{\Ctx{\e_1'}}{\T}$ and so 
	$\IsWFConf{\Delta_1\shsum\X\ctxmul\Delta'_1\shsum\Delta'_2}{\Ctx{\e_1'}}{\amem}{\T}$. 
	We have $\Delta =  \Delta_1\shsum\X\ctxmul\Delta'_1\shsum\Delta'_2= \Delta_1\shsum\X\ctxmul\Delta'_1\shsum\X\ctxmul\Delta'_2 = \Delta_1\shsum\X\ctxmul\Delta'$. We get the thesis, that is, $\Restr{\Delta\shsum\Gamma}{\Gamma} = \Gamma$, by \cref{lem:timesPlus}, since  
	$\links{\Delta_1}\cap(\links{\X\shmul(\Delta'\shsum\Gamma')}\cup\links{\X\shmul\Gamma'}) = \{\res\}$ or $\links{\Delta_1}\cap(\links{\X\shmul(\Delta'\shsum\Gamma')}\cup\links{\X\shmul\Gamma'}) = \{\res\} \cup \X$ if exists $\x \in \dom{\Delta}$ such that $\getCoeff{\Delta}{\x} = \X$
	$\dom{\Delta_1}\subseteq\dom{\Gamma_{\res}}\subseteq\dom{\X\ctxmul\Gamma'}$.

\item[\refToRule{block}] 
	$\reduce{\ExpMem{\Block{\T}{\loc}{\val}{\e}}{\mem}}{\ExpMem{\Subst{\e}{\val}{\loc}}{\mem}}$
	Applying \cref{lem:inversion}(\ref{lem:inversion:bl}) we obtain $\IsWFExp{\Gamma',\VarTypeCoeffect{\x}{\T}{\X}}{\e}{\T'}$ and $\IsWFExp{\Sigma}{\val}{\T}$ such that $\Gamma_1 = (\X \cup \{\link\}) \ctxmul \Sigma \shsum \Gamma'$. If $\val$ is a reference then by \cref{lem:inversion}(\ref{lem:inversion:x}) we know $\Sigma = \VarTypeCoeffect{\x}{\T}{\{\res\}}$ otherwise if $\val$ is a primitive value then we know by \refToRule{t-primitive} $\Sigma = \emptyset$. Applying \cref{lem:substitutionLemma} we obtain $\IsWFExp{\Gamma''}{\Subst{\e}{\val}{\loc}}{\T'}$ and $\Gamma'' \shord (\X \cup \{\link\}) \ctxmul \Sigma \shsum \Gamma'$. Knowing that $\Gamma'' \shsum \Gamma_2 \shsum \Gamma = \Delta \shsum \Gamma = \Gamma$ we obtain $\Restr{\Gamma}{\Gamma} = \Gamma$,that is, the thesis.

\end{description}
\end{proof}


\section{Proof of Theorem~\ref{thm:subj-red-extended}}

\begin{definition}
$\Gamma\prom\Gamma'$ if  
\begin{enumerate}
\item $\Gamma'= {\Gamma}$ or
\item $\Gamma'= \Sealed{\Gamma}$ or
\item $\Gamma'=\{\link\}\ctxmul\Gamma$ with $\link$ fresh, or
\item $\Gamma'= \{\link\}\ctxmul\Sealed{\Gamma}$ with $\link$ fresh.
\end{enumerate}
\end{definition}

\begin{lemma}\label[lemma]{lem:nonStructModif}
If $\Deriv:\IsWFExp{\Gamma}{{\e}}{\T}$, then there is a subderivation $\Deriv':\IsWFExp{\Gamma'}{{\e}}{\T'}$ of 
$\Deriv$ ending with a syntax-directed rule and $\Gamma'\prom\Gamma$
\end{lemma}

\brb
\begin{lemma}\label[lemma]{lem:leqCtx}
If for all $\x\in(\dom{\Gamma}\cap\dom{\Delta})$, $\getModif{\Gamma}{\x}\leq\getModif{\Delta}{\x}$ then for all $\x\in(\dom{\Gamma}\cap\dom{\Delta})$, $\getModif{\Sealed{\Gamma}}{\x}\leq\getModif{\Sealed{\Delta}}{\x}$.
\end{lemma}

\begin{proof}
We consider only type contexts for terms, so modifiers can be only $\imm$,$\mut$ and $\seal$. We also consider $\VarTypeCoeffect{\x}{\TypeMod{\T}{\modif}}{\X} \in \Gamma$ and $\VarTypeCoeffect{\x}{\TypeMod{\T}{\modif'}}{\X'} \in \Delta$. We have 3 cases:

\begin{itemize}
\item $\modif = \modif'$\\ Immediate
\item $\modif=\seal$ and $\modif'=\imm$\\ We know $\seal[\seal']=\seal$ and $\imm[\seal'] = \imm$, so we have the thesis
\item $\modif=\seal$ and $\modif'=\mut$\\ We know $\seal[\seal']=\seal$ and $\mut[\seal'] = \seal'$, so we have the thesis
\end{itemize}
\end{proof}
\erb

\bfd 
\begin{lemma}\label{lem:replace-mod}
Let $\IsWFMem{\Gamma,\Delta}{\mem}$ where, 
$\Delta = \VarTypeCoeffect{\x_1}{\TypeMod{\C_1}{\seal}}{\X},\ldots,\VarTypeCoeffect{\x_n}{\TypeMod{\C_n}{\seal}}{\X}$ and, 
for all $\x\in\dom{\Gamma}$, $\getCoeff{\Gamma}{\x}\cap\X  = \emptyset$. 
Let $\modif\notin\{\readonly,\capsule\}$ be a modifier and 
$\Theta = \VarTypeCoeffect{\x_1}{\TypeMod{\C_1}{\modif}}{\X_1},\ldots,\VarTypeCoeffect{\x_n}{\TypeMod{\C_n}{\modif}}{\X_n}$ be such that 
\begin{itemize}
\item $\modif = \imm$ implies $\X_i = \{\link_i\}$, with $\link_i$ fresh for all $i \in 1..n$. 
\item $\modif \neq \imm$ implies $\X_i=\X$ for all $i \in 1..n$. 
\end{itemize}
Then, $\IsWFMem{\Gamma,\Theta}{\mem}$ holds. 
\end{lemma}
\efd 

\bfd 
\begin{lemma}\label{lem:sameModif}
Let  $\IsWFMem{\Gamma\shsum\Gamma_\mem}{\mem}$, where
\begin{itemize}
\item $\Gamma_{\!\mem} = \sum_{\z\in\dom\mem}\VarTypeCoeffect{\z}{\TypeMod{\C_z}{\modif_\z}}{\{\link_\z\}}$ with $\link_\z$ fresh for all $\z \in \dom\mem$; 
\item $\Gamma = \sum_{\z\in\dom\mem} \{\link_\z\}\shmul \Gamma_\z$ where 
$\IsWFInMem{\Gamma_\z}{\mem(\z)}{\TypeMod{\C_\z}{\modif_\z}}$ for all $\z \in \dom\mem$; 
\item $\x,\y \in \dom\mem$ and 
$\mem(\x) = \Object{\C_\x}{\val_1,\ldots,\val_m}$,  
$\fields{\C_\x}=\Field{\T_1}{\f_1} \ldots \Field{\T_m}{\f_m}$, 
$\Gamma_\x = \sum_{j = 1}^m \Gamma'_j$, 
$\IsWFInMem{\Gamma'_j}{\val_j}{\Modif{\T_j}{\modif_\x}}$, for all $j \in 1...m$, and 
$\Modif{\T_i}{\modif_\x} = \TypeMod{\C_\y}{\modif_\y}$ and $\Gamma'_i = \VarTypeCoeffect{\y'}{\TypeMod{\C_\y}{\modif_\y}}{\Y}$ for some $i \in 1...m$. 
\end{itemize}
Let $\Delta$ be  such that
$\Delta = \{\link_\x\}\shmul\Delta_\x \shsum \sum_{\z\in\dom\mem\setminus\{\x\}} \{\link_\z\}\shmul\Gamma_\z$ with 
$\Delta_\x = \VarTypeCoeffect{\y}{\TypeMod{\C_\y}{\modif_\y}}{\Y} \shsum \sum_{j = 1}^{i-1} \Gamma'_j \shsum \sum_{j = i+1}^m \Gamma'_j$. 
Then, $\IsWFMem{\Delta\shsum\Gamma_\mem}{\UpdateMem{\mem}{\x}{i}{\y}}$ holds and $\getModif{\Gamma\shsum\Gamma_\mem}{\z} = \getModif{\Delta\shsum\Gamma_\mem}{\z}$ for all $\z \in \dom{\mem}$
\end{lemma}
\efd 

\begin{proof}[Proof of \cref{thm:subj-red-extended}]
Since the proof for condition $\Restr{(\Gamma'\shsum\Delta')}{\Gamma'} = \Gamma'$, for $\Gamma'=\Erase{\Gamma}$ and $\Delta'=\Erase{\Delta}$ is analogous to the proof for \cref{thm:subj-red} in this proof we focus on the condition for all $\x\in\dom{\Gamma}$, $\getModif{\Gamma}{\x}\leq\getModif{\Delta}{\x}$ .
If $\IsWFConf{\Gamma}{\e}{\mem}{\T}$, we have 
$\Gamma = \Gamma_1\shsum\Gamma_2$, 
$\IsWFExp{\Gamma_1}{\e}{\T}$ and $\IsWFMem{\Gamma_2}{\mem}$. 
We also know that 
$\Gamma_2 = \Gamma_\mem \shsum \Theta$, where 
$\Gamma_\mem = \VarTypeCoeffect{\x_1}{\T_1}{\{\link_1\}},\ldots,\VarTypeCoeffect{\x_n}{\T_n}{\{\link_n\}}$, 
$\Theta = \sum_{i = 1}^n \link_i \ctxmul \Theta_i$, 
$\IsWFExp{\Theta_i}{\mem(\x_i)}{\T_u}$ and 
$\link_1,\ldots,\link_n$ are fresh links. 
The proof is by induction on the reduction relation.
\begin{description}
\brb
\item[\refToRule{field-assign}] 
	 By \cref{lem:nonStructModif} and rule \refToRule{t-assign} we have  $\IsWFExp{\Gamma_1''}{\x}{\T_i}$ and $\IsWFExp{\Gamma_2''}{\val}{\T_i}$ such that $\Gamma_1'' \shsum \Gamma_2''\prom\Gamma_1$ with $\fields{\C}=\Field{\T_1}{\f_1} \ldots \Field{\T_n}{\f_n} i\in 1..n$. We have two interesting cases considering $\val=\y$:
\begin{itemize}
\item $\T_i=\TypeMod{\D}{\imm}$\\
Since $\IsWFExp{\Gamma_2''}{\y}{\TypeMod{\D}{\imm}}$ we know that $\getModif{\Gamma_2''}{\y} \neq \mut$ and so $\getModif{\Gamma}{\y} \neq \mut$. By rule \refToRule{t-var} we have $\IsWFExp{\VarTypeCoeffect{\y}{\TypeMod{\D}{\imm}}{\Y}}{\y}{\TypeMod{\D}{\imm}}$. By applying the same non syntax directed rules applied to 
$\IsWFExp{\Gamma_1'' \shsum \Gamma_2''}{\FieldAssign{\x}{\f_i}{\y}}{\TypeMod{\D}{\imm}}$ we have $\IsWFExp{\VarTypeCoeffect{\y}{\TypeMod{\D}{\imm}}{\Y}}{\y}{\T}$.
We have two cases:
\begin{itemize}
\item $\getModif{\Gamma}{\y} = \imm$\\
	By \cref{lem:sameModif} and applying rule \refToRule{t-conf} we have the thesis
\item $\getModif{\Gamma}{\y} = \seal$\\
	By \cref{lem:replace-mod} applied to variables with coeffect $\getCoeff{\Gamma_2}{\y}$ and with $\modif = \imm$, by \cref{lem:sameModif} and by applying rule \refToRule{t-conf} we have the thesis
\end{itemize}
\item $\T_i=\TypeMod{\D}{\mut}$\\
Since $\IsWFExp{\Gamma_1''}{\x}{\TypeMod{\C}{\mut}}$ and  $\IsWFExp{\Gamma_2''}{\y}{\TypeMod{\D}{\mut}}$ we know that $\getModif{\Gamma_1''}{\x} \neq \mut$ and $\getModif{\Gamma_2''}{\y} \neq \imm$ and so $\getModif{\Gamma}{\x} \neq \imm$ and $\getModif{\Gamma}{\y} \neq \imm$.
We have 4 cases:
\begin{itemize}
\item $\getModif{\Gamma}{\x} = \mut$ and $\getModif{\Gamma}{\y} = \mut$\\
	We can apply the same non syntax directed rules applied to $\IsWFExp{\Gamma_1'' \shsum \Gamma_2''}{\FieldAssign{\x}{\f_i}{\y}}{\TypeMod{\D}{\imm}}$ to $\IsWFExp{\Gamma_2''}{\val}{\T_i}$. By \cref{lem:sameModif} and applying rule \refToRule{t-conf} we have the thesis

\item $\getModif{\Gamma}{\x} = \mut$ and $\getModif{\Gamma}{\y} = \seal$\\
	 We know by \cref{lem:nonStructModif} and rule \refToRule{t-var} that $\IsWFExp{\VarTypeCoeffect{\x}{\T''}{\X''}}{\x}{\T''}$. We can apply rule \refToRule{t-var} and the same non syntax-directed rules applied to $\x$ on $\y$ to obtain $\IsWFExp{\VarTypeCoeffect{\y}{\TypeMod{\D}{\getModif{\Gamma_1'' \shsum \Gamma_2''}{\x}}}{\X'}}{\y}{\TypeMod{\D}{\mut}}$. If we apply the same non syntax-directed rules applied to $\IsWFExp{\Gamma_1'' \shsum \Gamma_2''}{\FieldAssign{\x}{\f_i}{\y}}{\T_i}$ we have $\IsWFExp{\VarTypeCoeffect{\y}{\TypeMod{\D}{\mut}}{\X}}{\y}{\T}$. By \cref{lem:replace-mod} applied to variables with coeffect $\getCoeff{\Gamma_2}{\y}$ and with $\modif = \mut$, by \cref{lem:sameModif} and applying rule \refToRule{t-conf} we have the thesis.

\item $\getModif{\Gamma}{\x} = \seal$ and $\getModif{\Gamma}{\y} = \seal'$\\
 	Reasoning similar to that above. We also apply \cref{lem:replace-mod} with $\modif=\seal$

\item $\getModif{\Gamma}{\x} = \seal$ and $\getModif{\Gamma}{\y} = \mut$\\
	We can apply the same non syntax-directed rules applied to $\IsWFExp{\Gamma_1'' \shsum \Gamma_2''}{\FieldAssign{\x}{\f_i}{\y}}{\T_i}$ to $\IsWFExp{\Gamma_2''}{\y}{\T_i}$ to have $\IsWFExp{\VarTypeCoeffect{\y}{\TypeMod{\D}{\mut}}{\Y}}{\y}{\T}$. By \cref{lem:replace-mod} applied to variables with coeffect $\getCoeff{\Gamma_2}{\x}$ and with $\modif = \mut$, by \cref{lem:sameModif} and applying rule \refToRule{t-conf} we have the thesis.

\end{itemize}
\end{itemize}
\erb

\brb
\item[ \refToRule{block}] 
We have $\Block{\T'}{\x}{\val}{\e}$.
By \cref{lem:nonStructModif} and rule \refToRule{t-block} we have a contexts $(\X\cup\{\link\})\ctxmul\Gamma_1''\ctxsum\Gamma_2'' \prom \Gamma_1$ such that $\IsWFExp{(\X\cup\{\link\})\ctxmul\Gamma_1''\ctxsum\Gamma_2''}{\Block{\T'}{\x}{\val}{\e}}{\T}$, $\IsWFExp{\Gamma_1''}{\val}{\T'}$ and $\IsWFExp{\Gamma_2'',\VarTypeCoeffect{\x}{\T'}{\X}}{\e}{\T}$. By \cref{lem:substitutionLemma} we have that $\IsWFExp{\Delta''}{\Subst{\e'}{\val}{\x}}{\T}$ with $\Delta'' \shord \X\ctxmul\Gamma_1'\ctxsum\Gamma_2''\shord (\X\cup\{\link\})\ctxmul\Gamma_1''\ctxsum\Gamma_2''$. By applying the same non syntax directed rules applied to $\IsWFExp{(\X\cup\{\link\})\ctxmul\Gamma_1''\ctxsum\Gamma_2''}{\Block{\T'}{\x}{\val}{\e}}{\T}$ since $\Delta'' \shord (\X\cup\{\link\})\ctxmul\Gamma_1''\ctxsum\Gamma_2''$ we have the thesis.

\item[\refToRule{field-access}] 
	We know that  $\reduce{\ExpMem{\FieldAccess{\x}{\f}}{\mem}}{\ExpMem{\val}{\mem}}$, hence 
	$\e = \FieldAccess{\x}{\f}$, $\e' = \val$. By \cref{lem:nonStructModif} and rule \refToRule{t-field-access} we have  $\IsWFExp{\Gamma''}{\FieldAccess{\x}{\f_i}}{\Modif{\T_i}{\modif}}$ and $\IsWFExp{\Gamma''}{\x}{\TypeMod{\C}{\modif}}$ such that $\Gamma''\prom\Gamma_1$ and $\T' = \Modif{\T_i}{\modif} \leq \T$ with $\fields{\C}=\Field{\T_1}{\f_1} \ldots \Field{\T_n}{\f_n} i\in 1..n$. We have two interesting cases:
\begin{itemize}
\item $\T_i=\TypeMod{\D}{\imm}$\\
By rule \refToRule{t-var} we can derive $\IsWFExp{\VarTypeCoeffect{\y}{\TypeMod{\D}{\imm}}{\X'}}{\y}{\TypeMod{\D}{\imm}}$. Since $\T=\Modif{\T_i}{\modif}=\TypeMod{\D}{\imm}$ we know that to $\IsWFExp{\Gamma''}{\FieldAccess{\x}{\f_i}}{\TypeMod{\D}{\imm}}$ is not applied rule \refToRule{t-caps}, so, if we apply to $\IsWFExp{\VarTypeCoeffect{\y}{\TypeMod{\D}{\imm}}{\X}}{\y}{\TypeMod{\D}{\imm}}$ the same non syntax-directed rules applied to $\IsWFExp{\Gamma''}{\FieldAccess{\x}{\f_i}}{\TypeMod{\D}{\imm}}$ we obtain $\IsWFExp{\VarTypeCoeffect{\y}{\TypeMod{\D}{\imm}}{\X}}{\y}{\T}$, where $\X=\getCoeff{\Gamma_1}{\x}$. By rule \refToRule{t-mem} and \refToRule{t-obj} we have $\IsWFInMem{\Delta_1\shsum\cdots\shsum\Delta_n}{\Object{\C}{\val_1,\ldots,\val_n}}{\TypeMod{\C}{\modif'}}$ and $\IsWFInMem{\Delta_i}{\val_i}{\Modif{\T_i}{\modif'}}$ where $\mem(\x)=\Object{\C}{\val_1,\ldots,\val_n}$, $\fields{\C}=\Field{\T_1}{\f_1} \ldots \Field{\T_n}{\f_n}$ and exists $i$ such that $\val_i=\y$. Since $\T_i=\TypeMod{\D}{\imm}$ we know $\Modif{\T_i}{\modif}=\TypeMod{\D}{\imm}$ and $\IsWFInMem{\Gamma_i}{\y}{\TypeMod{\D}{\imm}}$.By rule \refToRule{t-imm-ref} we obtain $\Delta_i = \VarTypeCoeffect{\y}{\TypeMod{\D}{\imm}}{\{\link\}}$ with $\link$ fresh. Since memory does not change applying rule \refToRule{t-conf} we obtain the thesis.
\item $\T_i=\TypeMod{\D}{\mut}$\\
 We know by \cref{lem:nonStructModif} and rule \refToRule{t-var} that $\IsWFExp{\VarTypeCoeffect{\x}{\T''}{\X''}}{\x}{\T''}$ and $\Gamma''= \VarTypeCoeffect{\x}{\TypeMod{\C}{\modif'}}{\X'}$. We can apply rule \refToRule{t-var} and the same non syntax-directed rules applied to $\x$ on $\y$ to obtain $\IsWFExp{\VarTypeCoeffect{\y}{\Modif{\T_i}{\modif'}}{\X'}}{\y}{\Modif{\T_i}{\modif}}$. If we apply the same non syntax-directed rules applied to $\IsWFExp{\Gamma''}{\FieldAccess{\x}{\f_i}}{\Modif{\T_i}{\modif}}$ we have $\IsWFExp{\VarTypeCoeffect{\y}{\Modif{\T_i}{\modif''}}{\X}}{\y}{\T}$. We know $\getModif{\Gamma}{\x}=\getModif{\VarTypeCoeffect{\y}{\Modif{\T_i}{\modif''}}{\X}}{\y}=\modif''$ and $\getCoeff{\Gamma_1}{\x}=\X$. By rule \refToRule{t-mem} and \refToRule{t-obj} we have $\IsWFInMem{\Delta_1\shsum\cdots\shsum\Delta_n}{\Object{\C}{\val_1,\ldots,\val_n}}{\TypeMod{\C}{\modif''}}$ and $\IsWFInMem{\Delta_i}{\val_i}{\Modif{\T_i}{\modif''}}$ where $\mem(\x)=\Object{\C}{\val_1,\ldots,\val_n}$, $\fields{\C}=\Field{\T_1}{\f_1} \ldots \Field{\T_n}{\f_n}$ and exists $i$ such that $\val_i=\y$. By rule \refToRule{t-imm-ref} we obtain $\Delta_i = \VarTypeCoeffect{\y}{\Modif{\T_i}{\modif''}}{\{\link\}}$ with $\link$ fresh. Since memory does not change applying rule \refToRule{t-conf} we obtain the thesis.
\end{itemize}

\item[\refToRule{ctx}] 
	We have $\e=\Ctx{\e_1}$ and  $\e'=\Ctx{\e'_1}$ and  ${\reduce{\ExpMem{\e_1}{\mem}}{\ExpMem{\e_1'}{\amem}}}$
	and $\IsWFExp{\Gamma_1}{\Ctx{\e_1}}{\T}$. 

We prove the thesis by induction on evaluation context $\ctx$ and we show only the relevant cases:
\begin{itemize}
\item $\ctx= \emptyctx$
	We know that $\IsWFConf{\Gamma}{\Ctx{\e_1}}{\mem}{\T}$ and that $\Ctx{\e_1}=\e_1$, so $\IsWFConf{\Gamma}{\e_1}{\mem}{\T}$. Applying the primary induction hypothesis to ${\reduce{\ExpMem{\e_1}{\mem}}{\ExpMem{\e_1'}{\amem}}}$ and $\IsWFConf{\Gamma}{\e_1}{\mem}{\T}$ and by knowing that $\Ctx{\e'_1}=\e'_1$ we obtain the thesis.\\

\item $\ctx = \FieldAssign{\ctx'}{\f}{\e'}$ By \cref{lem:nonStructModif} we know that exists a context $\Gamma_1'$ such that $\Gamma_1'\prom\Gamma_1$ and a derivation $\Deriv':\IsWFExp{\Gamma_1'}{\Ctx{\e_1}}{\T'}$ ending with a syntax directed rule. We know that the last applied rule in $\Deriv'$ must be \refToRule{t-assign}, so we know $\Gamma_1' = \Gamma_1''\shsum\Gamma_2''$ such that $\IsWFExp{\Gamma_1''}{\ctx'[\e_1]}{\TypeMod{\C}{\mut}}$ and $\IsWFExp{\Gamma_2''}{\e'}{\T_i}$ with $\fields{\C}=\Field{\T_1}{\f_1} \ldots \Field{\T_n}{\f_n}$ and $i\in 1..n$. We can apply rule \refToRule{t-conf} to obtain $\IsWFConf{\Gamma_1''\shsum\Gamma_2}{\ctx'[\e_1]}{\mem}{\TypeMod{\C}{\mut}}$. By \refToRule{ctx} we have ${\reduce{\ExpMem{\ctx'[\e_1]}{\mem}}{\ExpMem{\ctx'[\e'_1]}{\amem}}}$. By secondary induction hypothesis on this and on $\IsWFConf{\Gamma_1''\shsum\Gamma_2}{\ctx'[\e_1]}{\mem}{\TypeMod{\C}{\mut}}$ we have that $\IsWFConf{\Theta'}{\ctx'[\e'_1]}{\amem}{\TypeMod{\C}{\mut}}$ such that
 \begin{itemize}
 \item $\Restr{(\Erase{\Gamma_1''\shsum\Gamma_2}\shsum\Erase{\Theta'})}{\Erase{\Gamma_1''\shsum\Gamma_2} }= \Erase{\Gamma_1''\shsum\Gamma_2}$, 
 \item for all $\x\in\dom{\Gamma_1''\shsum\Gamma_2}$, $\getModif{\Gamma_1''\shsum\Gamma_2}{\x}\leq\getModif{\Theta'}{\x}$
 \end{itemize}
By rule \refToRule{t-conf} we have $\Theta' = \Delta'_1 \shsum \Gamma'_2$ such that $\IsWFExp{\Delta'_1}{\ctx'[\e'_1]}{\T'}$ and $\IsWFMem{\Gamma'_2}{\amem}$. By rule \refToRule{t-assign} we have $\IsWFExp{\Delta'_1\shsum\Delta_2}{\Ctx{e'_1}}{\T'}$. We can prove that

 \begin{itemize}
 \item for all $\x\in(\dom{\Gamma'_1}\cap\dom{\Delta'_1\shsum\Delta_2})$, $\getModif{\Gamma'_1}{\x}\leq\getModif{\Delta'_1\shsum\Delta_2}{\x}$
 \end{itemize}
Applying the same non syntax directed rules applied to $\IsWFExp{\Gamma_1'}{\Ctx{\e_1}}{\T'}$ and rule \refToRule{t-conf}, by the lemma \cref{lem:leqCtx}  we obtain the thesis.

\item $\ctx=\Block{\T_\x}{\x}{\ctx'}{\e'}$ By \cref{lem:nonStructModif} we know that exists a context $\Gamma_1'$ such that $\Gamma_1'\prom\Gamma_1$ and a derivation $\Deriv':\IsWFExp{\Gamma_1'}{\Ctx{\e_1}}{\T'}$ ending with a syntax directed rule. We know that the last applied rule in $\Deriv'$ must be \refToRule{t-block}, so we know $\Gamma_1' = ((\X\cup\{\link\})\ctxmul\Gamma_1'')\ctxlinsum\Gamma_2''$ such that $\IsWFExp{\Gamma_1''}{\ctx'[\e_1]}{\T_\x}$ and $\IsWFExp{\Gamma_2'',\VarTypeCoeffect{\x}{\T_\x}{\X}}{\e'}{\T_i}$. We can apply rule \refToRule{t-conf} to obtain $\IsWFConf{\Gamma_1''\shsum\Gamma_2}{\ctx'[\e_1]}{\mem}{\T_\x}$. By \refToRule{ctx} we have ${\reduce{\ExpMem{\ctx'[\e_1]}{\mem}}{\ExpMem{\ctx'[\e'_1]}{\amem}}}$. By secondary induction hypothesis on this and on $\IsWFConf{\Gamma_1''\shsum\Gamma_2}{\ctx'[\e_1]}{\mem}{\T_z}$ we have that $\IsWFConf{\Theta'}{\ctx'[\e'_1]}{\amem}{\T_z}$ such that
 \begin{itemize}
 \item $\Restr{(\Erase{\Gamma_1''\shsum\Gamma_2}\shsum\Erase{\Theta'})}{\Erase{\Gamma_1''\shsum\Gamma_2} }= \Erase{\Gamma_1''\shsum\Gamma_2}$, 
 \item for all $\x\in\dom{\Gamma_1''\shsum\Gamma_2}$, $\getModif{\Gamma_1''\shsum\Gamma_2}{\x}\leq\getModif{\Theta'}{\x}$
 \end{itemize}
By rule \refToRule{t-conf} we have $\Theta' = \Delta'_1 \shsum \Gamma'_2$ such that $\IsWFExp{\Delta'_1}{\ctx'[\e'_1]}{\T'}$ and $\IsWFMem{\Gamma'_2}{\amem}$. By rule \refToRule{t-block} we have $\IsWFExp{((\X\cup\{\link\})\ctxmul\Delta'_1)\ctxlinsum\Gamma_2''}{\Ctx{e'_1}}{\T'}$. We can prove that

 \begin{itemize}
 \item for all $\x\in(\dom{\Gamma'_1}\cap\dom{((\X\cup\{\link\})\ctxmul\Delta'_1)\ctxlinsum\Gamma_2''})$, $\getModif{\Gamma'_1}{\x}\leq\getModif{((\X\cup\{\link\})\ctxmul\Delta'_1)\ctxlinsum\Gamma_2''}{\x}$
 \end{itemize}
Applying the same non syntax directed rules applied to $\IsWFExp{\Gamma_1'}{\Ctx{\e_1}}{\T'}$ and rule \refToRule{t-conf}, by the lemma \cref{lem:leqCtx}  we obtain the thesis.
\end{itemize}
\erb
\end{description}
\end{proof}

\section{Module of fixpoints}

\begin{definition}[Idempotent]\label{def:idem} 
A homomorphism \fun{\ip}{\MM}{\MM} is \emph{idempotent} if 
$\ip \circ \ip = \ip$. 
\end{definition} 

Let \fun{\ip}{\MM}{\MM} be an idempotent homomorphism on the \RR-module \MM, 
and consider the  subset 
${\MSet_\ip = \{ \mel \in \MSet \mid \ip(\mel) = \mel \}}$ of fixpoints of $\ip$. 
Define the following operations, for all $\mel,\amel \in \MSet_\ip$ and $\rel \in \RSet$: 
\[
\mel \msum^\ip \amel := \ip(\mel\msum\amel) 
\qquad 
\mzero^\ip := \ip(\mzero) 
\qquad 
\rel\mmul^\ip\mel := \ip(\rel\mmul\mel) 
\]
We can prove the following result 

\begin{proposition}\label{prop:module-fix}
$\MM_\ip = \ple{\MSet_\ip, \mord,\msum^\ip,\mzero^\ip,\mmul^\ip}$ is an \RR-module. 
\end{proposition}

\begin{proof}
All the operations are trivially well-defined as $\ip$ is idempotent. 
Moreover, note that, since $\ip$ is an homomorphism, we have $\mzero^\ip = \ip(\mzero) = \mzero$. 
We check the module axioms. 
\begin{align*}
&\text{associativity of }\msum^\ip & 
  \mel\msum^\ip(\amel\msum^\ip\bmel) &= \ip(\mel\msum \ip(\amel\msum\bmel)) = \ip(\ip(\mel)\msum \ip(\amel\msum\bmel))\\ 
&&&\mord \ip(\ip(\mel\msum\amel\msum\bmel)) = \ip(\mel\msum\amel\msum\bmel) \\
&&&= \ip(\ip(\mel)\msum \ip(\amel) \msum \bmel) \mord \ip(\ip(\mel\msum\amel)\msum \bmel) \\
&&&= (\mel\msum^\ip\amel)\msum^\ip \bmel  \\ 
&&(\mel\msum^\ip\amel)\msum^\ip\bmel &= \ip(\ip(\mel\msum\amel)\msum\bmel) = \ip(\ip(\mel\msum\amel)\msum \ip(\bmel)) \\
&&&\mord \ip(\ip(\mel\msum\amel\msum\bmel)) = \ip(\mel\msum\amel\msum\bmel) \\
&&&= \ip(\mel\msum \ip(\amel) \msum \ip(\bmel)) \mord \ip(\mel\msum \ip(\amel\msum\bmel)) \\
&&&=\mel\msum^\ip(\amel\msum^\ip\bmel)  \\ 
&\text{commutativity of }\msum^\ip 
  &\mel\msum^\ip\amel &=\ip(\mel\msum\amel) = \ip(\amel\msum\mel) = \amel\msum^\ip\mel \\ 
&\text{neutrality of }\mzero^\ip 
  &\mel\msum^\ip\mzero^\ip &= \ip(\mel\msum \ip(\mzero)) = \ip(\mel\msum\mzero) = \ip(\mel) = \mel \\
&\text{linearity of }\mmul^\ip 
  &\rel\mmul^\ip (\mel\msum^\ip\amel) &= \ip(\rel\mmul \ip(\mel\msum\amel)) = \ip(\ip(\rel\mmul(\mel\msum\amel))) \\
&&&= \ip(\rel\mmul(\mel\msum\amel)) = \ip(\rel\mmul\mel\msum\rel\mmul\amel) \\
&&&= \ip(\rel\mmul \ip(\mel) \msum \rel \mmul \ip(\amel)) = \ip(\ip(\rel\mmul\mel)\msum \ip(\rel\mmul\amel)) \\ 
&&&= \rel\mmul^\ip \mel \msum^\ip \rel \mmul^\ip \amel  \\ 
&&\rel\mmul^\ip\mzero^\ip &= \ip(\rel\mmul \ip(\mzero)) = \ip(\rel\mmul\mzero) = \ip(\mzero) = \mzero^\ip \\
&&(\rel\rsum\arel)\mmul^\ip\mel &= \ip((\rel\rsum\arel)\mmul\mel) = \ip(\rel\mmul\mel\msum\arel\mmul\mel) \\
&&&= \ip(\rel\mmul \ip(\mel) \msum \arel\mmul \ip(\mel)) = \ip(\ip(\rel\mmul\mel)\msum \ip(\arel\mmul\mel)) \\
&&&= \rel\mmul^\ip\mel \msum^\ip \arel\mmul^\ip\mel \\
&&\rzero\mmul^\ip\mel &=\ip(\rzero\mmul\mel) = \ip(\mzero) = \mzero^\ip \\
&&(\rel\rmul\arel)\mmul^\ip\mel &= \ip((\rel\rmul\arel)\mmul\mel) = \ip(\rel\mmul(\arel\mmul\mel))  \\
&&&= \ip(\rel\mmul(\arel\mmul \ip(\mel)) = \ip(\rel\mmul \ip(\arel\mmul\mel)) \\
&&&= \rel\mmul^\ip(\arel\mmul^\ip\mel) \\
&&\rone\mmul^\ip\mel &= \ip(\rone\mmul\mel) = \ip(\mel) = \mel 
\end{align*}
\end{proof}

\end{document}